\documentclass[11pt]{article}
\usepackage{amsmath,amsthm,amssymb}
\usepackage[utf8]{inputenc}
\usepackage{fullpage}
\usepackage{hyperref}
\usepackage{natbib}
\usepackage{eqparbox}
\usepackage{array}

\usepackage{url}
\usepackage{graphicx}
\usepackage{soul}
\usepackage[usenames,dvipsnames]{xcolor}
\usepackage{tikz}
\usepackage{float}
\usepackage{tcolorbox}
\usepackage{amsmath,mathtools}

\usepackage{algorithmic}
\usepackage{algorithm}
\usepackage[normalem]{ulem}
\usepackage{todonotes}
\usepackage{physics}
\usepackage{caption}
\usepackage{subcaption}
\usepackage{enumitem}

\newcommand*\samethanks[1][\value{footnote}]{\footnotemark[#1]}

\newcommand{\F}{\mathcal{F}}

\newcommand{\E}{\mathbf{E}}
\usepackage{xspace}

\newcommand{\std}[1]{#1-$\mathcal{D}_{\mathtt{sp}}$\xspace}
\newcommand{\lstd}{\std{$\lambda$}}
\newcommand{\stdmath}[1]{#1-\mathcal{D}_{\mathtt{sp}}}
\newcommand{\lstdmath}{\stdmath{\lambda}}
\newcommand{\spandist}{\ensuremath{\mathcal{D}_{\mathtt{sp}}}\xspace}

\usepackage{thmtools, thm-restate}
\declaretheorem{theorem}
\declaretheorem[sibling=theorem]{claim}
\declaretheorem[sibling=theorem]{lemma}
\declaretheorem[sibling=theorem]{corollary}

\newtheorem{conjecture}{Conjecture}

\newtheorem{observation}[theorem]{Observation}
\newtheorem{example}[theorem]{Example}
\theoremstyle{definition}
\newtheorem{definition}{Definition}

\usepackage{xcolor}
\usepackage[export]{adjustbox}
\usepackage{caption}
\usepackage{subcaption}

\newtheorem{subclaim}{Subclaim}[claim]

\usepackage[capitalize]{cleveref}
\crefname{subclaim}{subclaim}{subclaims}
\Crefname{subclaim}{Subclaim}{Subclaims}

\crefname{claim}{claim}{claims}
\crefname{observation}{observation}{observations}
\Crefname{observation}{Observation}{Observations}

\renewcommand{\sp}[1]{\texttt{sp}(#1)}

\title{Balanced Spanning Tree Distributions Have Separation Fairness}
\author{
  Harry Chen\thanks{This work was done while the author was at Duke  University.} \\
  MIT \\
  \texttt{harryc60@mit.edu} \\
  \and
 Kamesh Munagala \thanks{Supported by NSF grant IIS-2402823.} \qquad Govind S. Sankar\samethanks[2] \\
  Duke University\\
  \texttt{\{munagala,govind.subash.sankar\}@duke.edu} 
}

\date{}

\begin{document}
\maketitle
\renewcommand{\F}{\mathcal{F}}
\newcommand{\R}{\mathbb{R}}
\newcommand{\linv}{L^+}
\newcommand{\linvsq}{L^{2+}}
\newcommand{\dbias}{\mathcal{D}_\mathrm{biased}}
\newcommand{\dspan}{\mathcal{D}_\mathrm{uniform}}
\newcommand{\psep}{\mathrm{Pr}_{\mathrm{sep}}}
\newcommand{\lr}[1]{\left(#1\right)}

\newcommand{\gmn}{\boxplus_{m\times n}}
\newcommand{\dual}[1]{#1^*}
\newcommand{\dualvertex}[1]{\dual{#1}}
\newcommand{\dualgraph}[1]{\dual{#1}}
\newcommand{\outerface}{\dualvertex{v_0}}
\newcommand{\gmnd}{\dualgraph{\gmn}}

\newcommand{\dualedge}[2]{\langle\dualvertex{#1},\dualvertex{#2}\rangle}
\newcommand{\dualcorresp}[2]{#1=\mathrm{dual}(#2)}

\begin{abstract}

Sampling-based methods such as ReCom are widely used to audit redistricting plans for fairness, with the balanced spanning tree distribution playing a central role since it favors compact, contiguous, and population-balanced districts. However, whether such samples are truly representative or exhibit hidden biases remains an open question. In this work, we introduce the notion of separation fairness, which asks whether adjacent geographic units are  separated with at most a constant probability (bounded away from one) in sampled redistricting plans. Focusing on grid graphs and two-district partitions, we prove that a smooth variant of the balanced spanning tree distribution satisfies separation fairness. Our results also provide theoretical support for popular MCMC methods like ReCom, suggesting that they maintain fairness at a granular level in the sampling process.
Along the way, we prove a novel local-interchangeability lemma for 2-partitions on grids, showing that any separation of adjacent vertices can be undone via a constant-sized modification. This lemma, along with our other tools for analyzing the structure of partitions and loop-erased random walks, may be of independent interest.



\end{abstract}

\section{Introduction}

Redistricting, the process of redrawing the boundaries of electoral units, is a fundamental aspect of representative democracies.
These boundaries are typically redrawn at regular intervals to reflect changes in population demographics or in response to legislative mandates. While intended to ensure that each district contains a roughly equal number of people, this process presents an opportunity for political parties to sway the outcome of an election\footnote{Throughout this paper, we only consider first-past-the-post voting rules.} to their benefit.
Such manipulation is called gerrymandering.

Gerrymandering typically manifests in two primary forms: ``packing" and ``cracking"~\cite{StephanopoulosMcGhee2015}.
Packing concentrates voters of a particular demographic or political affiliation into a single district, minimizing their influence in  other districts. Conversely, cracking disperses a cohesive group of voters across several districts, preventing them from achieving a majority in any single district.  Both tactics undermine the principle of fair representation, effectively diluting the voting power of targeted groups. 
Gerrymandering has recently found particular interest both academically and legislatively in several states in the US \cite{McGann_Smith_Latner_Keena_2016}, leading to several high-profile trials, with varying outcomes.
Academic research has become increasingly vital in these trials, directly providing data-driven evidence used by expert witnesses in legal challenges to allegedly gerrymandered maps \cite{Mattingly2019_expert,Pegden2017_expert}.

Since the power to draw the districts often rests with incumbent legislators or appointed committees, academic research has largely focused on evaluating proposed plans for fairness and identifying gerrymandering. One significant line of work along this direction is through Monte-Carlo Markov Chains (MCMC) that produce an ensemble of redistricting plans from ``unbiased''
distributions over plans that satisfy the basic requirements enshrined in the U.S. Constitution. If the characteristics of the proposed plan fall outside the typical range observed in the ensemble of randomly generated plans, it can provide statistical evidence that the plan is an outlier and potentially the result of gerrymandering. These characteristics are either scores based on compactness of the plan (such as the Reock~\cite{reock1961note} and Polsby-Popper~\cite{polsby1991third} scores), or scores based on partisan outcomes generated by the plan (such as the efficiency gap~\cite{StephanopoulosMcGhee2015}, mean-median gap~\cite{wang2016three}, partisan symmetry~\cite{warrington2018quantifying}, and the GEO metric~\cite{campisi2022geography}), or scores based on competitiveness of the plan~\cite{deford2020computational}. Many of these measures are used in publicly available tools~\cite{GEO,princeton}.

\subsection{Redistricting Plans and the Spanning Tree Distribution}
The above sampling problem is the focus of this paper. Formally, the state is modeled as a planar graph $G(V,E)$ on precincts\footnote{The term used to refer to the geographic unit varies depending on the state and country. For example, Voting Tabulation District (VTD) is another term used in the United States.}, where there is an edge between two nodes (or precincts) if they are adjacent. A node is weighted by its population. A redistricting plan is a partitioning of this graph into $k$ regions or {\em districts}, each of which elects one representative.
The basic requirements on a redistricting plan are threefold --- {\em contiguity}, meaning each district forms a connected component; {\em balance}, meaning the total weight of nodes in each district is approximately equal; and {\em compactness}, meaning a region is not too oblong. The last property is measured in several ways, one popular method being to minimize the total number of edges cut by the district boundaries.

The challenge now is to define a distribution over redistricting plans satisfying the above desirable properties, in a way that sampling from this distribution is efficient. These samples are then used to audit the characteristics of the proposed plan as described above. The most popular sampling distribution is called the {\em spanning tree distribution}, $\spandist$. In this distribution, a partition (or redistricting plan) $P$ is sampled proportional to its spanning tree score $\sp{P}$. For a graph $G$, define $\sp{G}$ to be the number of spanning trees of $G$. The spanning tree score of a partition $P=G_1,\ldots,G_k$ is defined as the product of the spanning tree scores $\sp{G_i}$ of each part (or district) $G_i$. 

\newcommand{\imb}[1]{\ensuremath \texttt{imb}(#1)}

Clearly, any partition is sampled with non-zero probability only if its individual parts have a non-zero number of spanning trees and hence, form a connected component. It is shown in~\cite{procaccia2022compact} that this distribution is biased towards plans with small number of cut edges between districts, and are hence more compact.  Properties like compactness and balance can also be achieved by assigning a score $f(\cdot)$ to partitions and sampling a partition $P$ proportional to $f(P) \cdot \sp{P}$~\cite{forest-recom-herschlag21,deford2019recombination,georgia-herschlag-mattingly}. For balance, the scoring function $f$ can assign $1$ to partitions that satisfy population balance up to a certain tolerance and $0$ otherwise. This would correspond to simply rejecting partitions that are imbalanced beyond the prescribed tolerance. This is called the {\em balanced spanning tree} distribution. The work of~\cite{jamie_grid24} samples exactly from the balanced spanning tree distribution and shows that vanilla rejection sampling runs in polynomial time for constant $k$ on a grid graph. However, this method is not practical since the running time depends exponentially on $k$. Instead, MCMC methods such as ReCom are the method of choice in practice.

\paragraph{The ReCom Markov Chain.} The well-known ReCom heuristic~\cite{deford2019recombination} attempts to sample from the balanced spanning tree distribution (though the exact distribution it samples from cannot be characterized). At each step of the Markov chain, given a redistricting plan with $k$ districts, two adjacent districts are chosen at random, and merged into a super-district. This super-district is split into two districts by sampling from the balanced spanning tree distribution of this super-district with $2$ parts. Though this Markov chain is not reversible, and mixes in exponential time in the worst case~\cite{Charikar-Recom}, for maps that arise in practice, this method is easy to implement and efficient.
Several variations of ReCom exist, including one whose steady state distribution is exactly the balanced spanning tree distribution~\cite{cannon2022spanning}.

\subsection{Separation Fairness and Our Result}
    Though the balanced spanning tree distribution is widely used, one desideratum that is under-explored is whether the distribution over plans it generates is ``unbiased'' in any sense. Given these samples are used for auditing the fairness characteristics of a given redistricting plan, it would be imperative to show that the samples themselves are not biased in an unfair way.

The hurdle with exploring this question is the complexity of characterizing the balanced spanning tree distribution itself. We therefore study a simple bias-freeness property based on the marginals this distribution induces on individual edges of the graph being cut. We call this property \textit{separation fairness}.

\begin{definition}
Given a graph $G(V,E)$ and constant $\alpha \in (0,1]$, a distribution over redistricting plans is $\alpha$-fair if every pair of adjacent nodes (precincts) is not separated, {\em i.e.}, assigned to the same district, with probability at least $\alpha$. 
\end{definition}

As a motivating example, if two adjacent precincts are separated in every redistricting plan in the support of the distribution, then this could be perceived as an example of ``cracking'' as mentioned earlier. In other words, we would like the local picture from the perspective of any edge in the graph to look random, so that the edge is not always separated. This is one way to capture the notion of ``unbiasedness'' alluded to above.
We therefore ask: 

\begin{quote}
    Is the balanced spanning tree distribution $\alpha$-fair for constant $\alpha > 0$?
\end{quote}

\paragraph{Main Result.} We answer this question in the affirmative for a smooth version of this distribution, under the restrictions that $k = 2$ and the graph is a $m \times n$ grid.  For $k = 2$ and constant $\lambda > 0$, the $\lambda$-{\em smooth} spanning tree distribution (henceforth \lstd) samples partitions the following way:
\begin{equation}
\Pr[P \text{ is sampled}]\propto \sp{P} e^{-{\lambda \cdot \imb{P}}} \label{eqn:biased-distribution}    
\end{equation}
where for a partition $P=(G_1,G_2)$, the imbalance $\imb{P}=\frac{\abs{\abs{G_1}-\abs{G_2}}}{2}$.  It follows from~\cite{jamie_grid24} that for a $m \times n$ grid graph, the imbalance of samples from this distribution is $O\left(\frac{\ln mn}{\lambda}\right)$ with probability $1 - \frac{1}{\mbox{poly}(mn)}$. (See \Cref{sec:imbalance} for a proof.) As $m,n \rightarrow \infty$, the imbalance therefore goes to $0$ as a fraction of $mn$. 

We note that the smooth distribution is well-motivated in practice: Several existing works sample from similar distributions that down-weight partitions with high population imbalance \cite{forest-recom-herschlag21,mattingly_will_of_the}.
We remark that the specific term incorporating the imbalance is not important for our arguments; we only require that it decays exponentially with imbalance.

Our main result is the following theorem:
\begin{theorem}\label{thm:main-intro}
For any constant $\lambda \geq 0$, there is a constant $f_{\lambda} > 0$ such that for the distribution \lstd on a $n \times m$ grid (with $m,n = \Omega(1)$) with $k = 2$ is $f_{\lambda}$-fair. 
\end{theorem}

To interpret the technical significance of our result, suppose we define the fractional imbalance as $\frac{\imb{P}}{mn}$. Our result shows constant separation fairness when the fractional imbalance vanishes as the size of the grid increases.  \Cref{thm:main-intro} provides, to our knowledge, the first property proven for this challenging regime. 

The proof of \cref{thm:main-intro} hinges on the following local-interchangeability result which is of interest on its own. Such a result is non-trivial even for 2-partitions on a grid. Our proof requires a careful case-analysis, aided by tools we develop in \cref{sec:tools} that are essential in pruning the number of cases to a tractable level.

\begin{theorem}[Informal version of \cref{thm:main-reconnect}]
    Given any 2-partition $P$ of a $n\times m$ grid (with $n,m=\Omega(1)$), and two adjacent vertices $u,v$ that are in different parts in $P$, there is a partition $P'$ such that $u,v$ are in the same part in $P'$, and $P,P'$ differ only in a constant-sized neighborhood around $u,v$.
\end{theorem}

Our results also have implications for groups of such separation events. We show that the probability that any $q$ sufficiently separated edges are simultaneously separated decays rapidly as $q$ increases. 

\begin{corollary}[Informal version of \cref{thm:group-separation}]\label{thm:group-separation-informal}
    For any set of $q$ sufficiently separated edges of a $n\times m$ grid (with $n,m=\Omega(1)$), the probability that a 2-partition sampled from \lstd
    separates all the $q$ edges is at most $e^{-\Omega(q)}$.
\end{corollary}

This shows a strong independence property for separation events in spanning tree distributions.

\paragraph{Connection to ReCom.} Though we restrict ourselves to the case of $k=2$ on a grid, our results have practical implications for the ReCom MCMC method. Note that at any step of the ReCom chain, the process samples from the balanced spanning tree distribution on two parts. Informally, our result suggests that any individual step of the chain will be $\alpha$-fair (for constant $\alpha$) for the super-district that is being partitioned. Since the super-district is chosen at random, an edge that was a boundary edge in the previous step becomes an interior edge of the super-district with probability $\Omega\left(\frac{1}{k}\right)$. Conditioned on this, the $\alpha$-fairness condition implies this edge stays an interior edge with constant probability. This means that in steady state, the distribution that ReCom samples from is likely to be $\alpha$-fair for $\alpha = \Omega\left(\frac{1}{k}\right)$.

Though the above is {\em not} a formal statement, we show that it holds true empirically. We run the ReCom chain on the maps of North Carolina and Pennsylvania, and on a $50 \times 50$ grid graph, with $k = 14,18$ and $10$ respectively. We run the chain $N = 10^5$ steps to obtain $1000$ maps, and consider the fraction of times an edge is separated, yielding the empirical probability of separation.  In \Cref{fig:experiment}, we present the histogram of this empirical probability. Note that the maps are $\alpha$-fair for $\alpha \ge 0.6$. This suggests that MCMC methods like ReCom offer strong separation fairness on real-world maps, and \Cref{thm:main-intro} provides a theoretical basis for explaining this phenomenon.

\begin{figure}[htbp]
 \centering
     \hfill
     \begin{subfigure}[b]{0.3\textwidth}
         \centering
         \includegraphics[width=\textwidth]{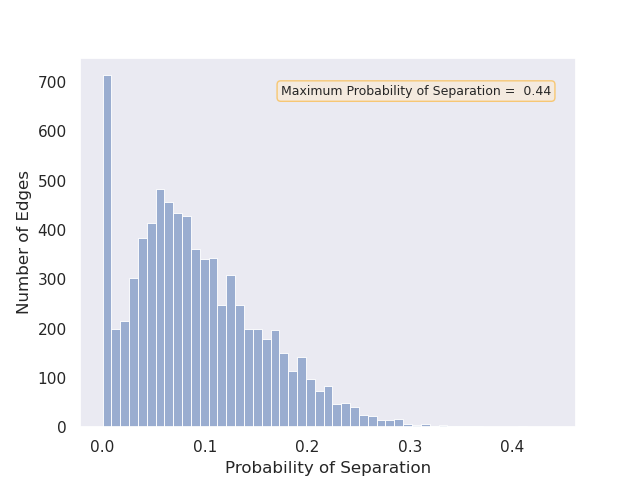}
         \caption{The North Carolina precinct graph\footnotemark with $k=14$ districts.}
     \end{subfigure}
     \hfill
     \begin{subfigure}[b]{0.3\textwidth}
         \centering
         \includegraphics[width=\textwidth]{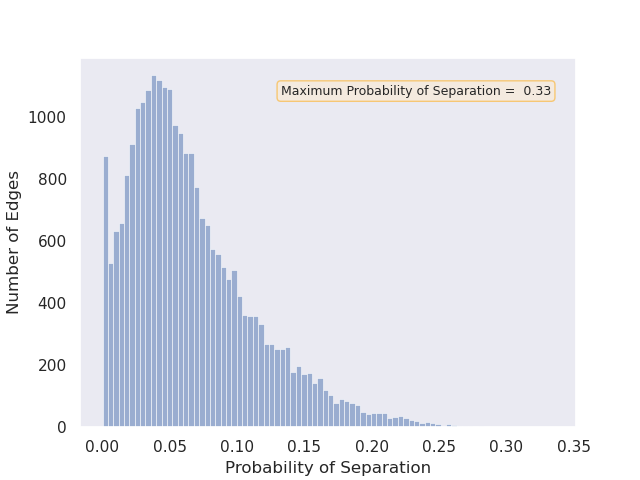}
         \caption{The Pennsylvania precinct graph with $k=18$ districts.}
     \end{subfigure}
     \hfill
    \begin{subfigure}[b]{0.3\textwidth}
         \centering
         \includegraphics[width=\textwidth]{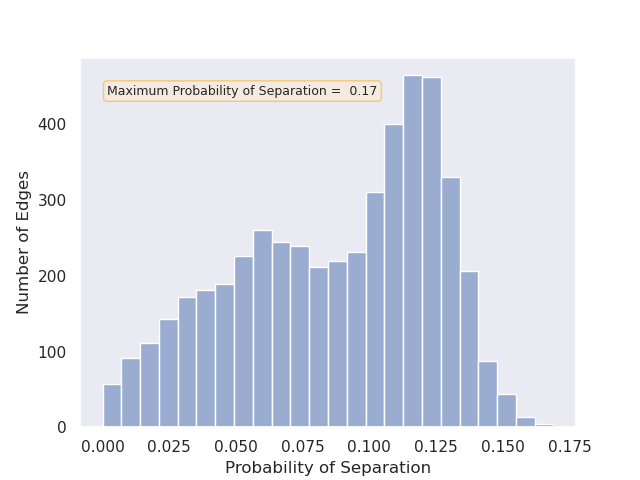}
         \caption{A $50\times 50$ grid graph with $k=10$ districts.}
     \end{subfigure}
     \hfill
    \caption{Histogram of separation probabilities of the endpoints of edges in three different graphs over 1000 redistricting plans sampled from ReCom with a population imbalance of $1\%$.}
    \label{fig:experiment}
\end{figure}
\footnotetext{Both the North Carolina and Pennsylvania precinct graphs, along with the population data, were obtained from \cite{mgggstates}.}

Finally, we show that our results also imply constant separation fairness for the standard spanning tree distribution restricted to exactly balanced partitions, or to partitions where the imbalance is an additive 3, under some conjectures on the properties of the spanning tree distribution.
We discuss this further in \Cref{sec:discussion}.

\paragraph{Technical Contribution.} The balanced spanning tree distribution is a tricky object --- even properties that seem intuitively true are not easy to show, and can often be false. For instance, though the ReCom chain seems quite intuitive, it is not only  irreversible, but in the worst-case, also torpidly mixing~\cite{Charikar-Recom} and not ergodic~\cite{tuckerfoltz2024lockedpolyominotilings}. On the other hand, rejection sampling --- where a random spanning tree is chosen and partitioned into $k$ equal parts (with rejection if such a partition is not possible) --- runs in polynomial time for constant $k$.

In the same vein, though separation fairness may seem quite intuitive, we are considering (almost) balanced partitions generated by the spanning tree distribution, which biases towards compactness~\cite{procaccia2022compact}. This creates a fundamental tension. The bias towards compactness favors short, straight partition boundaries, which would cut the same central edges in every instance. However, to ensure fairness, the boundary must be sufficiently random, so that it does not cut any single edge deterministically. The core of our technical work is to show that the distribution is not \textit{so} compact that this randomness is lost. We show that the separation probability is a constant bounded away from 1 for grids, even with vanishingly small fractional imbalance, which is a priori not obvious.


Similarly, though the spanning tree distribution has a close connection to random walks on the underlying graph, a significant amount of the randomness is ``lost'' when constraints such as balance are added. For example, consider the probability that two vertices are separated in the spanning tree distribution \spandist. For $k=2$, this probability can be shown to be directly proportional to the \textit{effective resistance} between the two vertices \cite{barrett2019spanning2forestsresistancedistance,wu15-sepprob,wu23-sepprob}.
On the other hand, no general expression is known for this probability even under mild modifications of the distribution towards balance.

Building upon the established connection between random walks on the dual graph and random partitions~\cite{jamie_grid24}, our work introduces a novel perspective: viewing these partitions as originating from random cycles on the dual graph – an idea subtly present in their work. 
This shift in viewpoint allows us to develop a modified version of Wilson's Algorithm~\cite{wilson_randomwalk} that directly samples these cycles, while preserving the same underlying distribution.
We then analyze this modified process and show that such a random cycle that separates two adjacent vertices $u,v$ can be modified with a few changes to a cycle that does not separate them. 
We seek such a modification that preserves balance between the partitions.
Though one would expect that this kind of modification is always possible locally near $u,v$ while keeping the balance constraint intact, this is not true: it may be the case that changes need to be made far away from $u,v$ (see \Cref{fig:bad-rebalance}). 

\Cref{def:unsep} and \Cref{thm:main-exact} form the heart of our technical contributions. The former and the associated \Cref{thm:walk-collapse} shows that our modified process indeed allows us to analyze partitions purely from the perspective of random cycles.  Our main contribution here is a constructive mapping between random walks: given a walk whose erasure separates two vertices, we devise an intricate algorithm that builds a corresponding walk that does not, carefully re-splicing all the original loops to ensure the mapping is one-to-one, and the number of different directed edges is small. This formalizes the intuition that a local change to a partition's boundary should correspond to a local change in the underlying random walk.
This result can be viewed as a complement to the main result from \cite{procaccia2022compact} - they show that two partitions whose boundaries greatly differ (in terms of length) must have greatly different probabilities of being sampled whereas ours shows that two partitions with very similar boundaries have similar probabilities of being sampled.
The latter result, \Cref{thm:main-exact}, proves that such a local modification is always possible. To do this, we develop a suite of graph-theoretic tools in \cref{sec:tools} specifically for analyzing 2-partitions on grid graphs. 
Our key insight is the following: After switching vertices from one part to the other, either the original partitions were both    ``thick'' and this did not disconnect any part, or they were ``thin'' and there must exist some structures that we can identify and perform local surgery on (\Cref{lem:1-thin,lem:2-thin}).
Through an extensive case analysis facilitated by these tools, we prove that any separating partition can be modified in a small neighborhood to become non-separating, with only a constant change to the imbalance.


\subsection{Related Work}

\paragraph{Redistricting as Optimization.} The idea of using computational tools in redistricting dates back to the 1960s~\cite{vickrey1961prevention, hess1965nonpartisan}. Since then, an extensive line of work (see~\cite{becker2020redistricting} for a comprehensive survey) has cast the redistricting task as an optimization problem. As mentioned before, the objective and constraints capture the population balance, contiguity, and compactness criteria of the districts. If we consider a property such as compactness, minimizing number of cut-edges immediately makes the problem NP-Complete by itself, for $k>2$, since it is precisely the $k$-Cut problem, and the additional requirement of population balance makes it NP-Complete even for $k=2$, since it is the Graph Bisection problem~\cite{cohen-addad2021computational,cohen-addad2018balanced,GareyJohnson1979}.  Multiple efficient algorithmic approaches have been proposed~\cite{altman2011bard, gawrychowski2021voronoi,king2015efficient,levin2019automated,liu2016pear} --- each of these methods comes with trade-offs in the guarantees that can be achieved.

\paragraph{Redistricting as Ensembles.} Instead of optimizing and finding a single best redistricting plan, another line of work focuses on generating a large \emph{ensemble} of redistricting plans, with the goal of auditing a given plan for fairness, or even finding a plan that is fair. These methods include Flood Fill~\cite{cirincione2000assessing, magleby2018new}, Column Generation~\cite{gurnee2021fairmandering}, and the widely adopted Markov Chain Monte Carlo (MCMC) approach~\cite{deford2019recombination,fifield2020automated, lin2022auditing,tam2016toward}.  These ensemble based approaches provide a natural, statistical way of auditing a given plan for fairness~\cite{ herschlag2020quantifying,herschlag2017evaluating}. 

Though MCMC methods like ReCom~\cite{deford2019recombination} are empirically very successful and hence widely used, they offer little in terms of formal guarantees on properties like compactness of the resulting distribution. Nevertheless, there have been some positive and negative results. On the positive side, it is shown in~\cite{procaccia2022compact} that ReCom provably biases towards compact plans. On the negative side, it is known~\cite{Charikar-Recom} that this chain mixes in exponential time in the worst case. Our work is in this vein, where we show a different (and positive) property --- separation fairness --- for this method.

\paragraph{Separation Fairness.} Similar notions to separation fairness have been considered before, although not in a probabilistic sense. For example, Chen et al \cite{chen2022turning} discuss the importance of keeping ``Communities of Interest'' together in electoral districts, highlighting the impact of gerrymandering on reducing the influence of the Asian-American voting bloc in city-level elections in New York City. In fact, in certain jurisdictions, there are already legal requirements to keep such Communities of Interest intact~\cite{aceproject_communityofinterest}.
In a probabilistic sense, a small separation probability is a key feature of low-diameter decompositions~\cite{LinialSaks1991}, which have found broad applications in metric embeddings. Such properties have also been studied under names such as ``pairwise fairness'' or ``community cohesion'' in graph clustering problems~\cite{brubach2020pairwise,munagala-individualfair}.

Most related to our paper is the work of~\cite{jamie_grid24}, who present a polynomial time algorithm (for constant $k$) for sampling from the balanced spanning tree distribution. Their work implies constant separation fairness, but only in the regime where the fractional imbalance is a non-vanishing constant. In contrast, our work addresses the more 
challenging regime where the fractional imbalance asymptotically vanishes. This necessitates the development of novel analysis techniques. The analysis in~\cite{jamie_grid24} lower bounds the distribution's ``global'' randomness by showing a constant probability that the boundary of a sampled 2-partition of the $n\times n$ grid is within $\epsilon n$ of certain guiding curves. On the other hand, our work analyzes the distribution at a much finer resolution, proving that it also behaves ``locally'' randomly.

\paragraph{Scaling Limits.}
A different line of work examines the limits of the random walks underlying the sampling methods used for spanning tree distributions. It is a classical result that random walks on a grid, when properly scaled, converge to Brownian motion~\cite{durrett2019probability}. 
On the other hand, loop-erased random walks central to Wilson's algorithm, when properly scaled, converge to a stochastic process known as Schramm-Loewner Evolution (SLE). Specifically, the seminal work of Lawler, Schramm, and Werner~ \cite{LSW2004} established that the scaling limit of a loop-erased random walk is SLE$_2$. While these limiting approaches provide insight into the large-scale geometric properties of random partitions, our focus is at a finer level that is not directly addressed by this framework.
\section{Preliminaries}\label{sec:prelim}
We will begin with basic graph-theoretic definitions, introduce the key redistricting distributions, and then describe the connection between these distributions and random walks on a graph's dual, which is central to our analysis.

Our input is a connected graph $G=(V(G),E(G))$, where $V(G)$ is the set of vertices representing precincts and $E(G)$ is the set of edges connecting adjacent precincts. A \textbf{$k$-partition} of this graph is a collection of $k$ vertex-disjoint subgraphs, $(G_1, \ldots, G_k)$, whose vertices partition $V(G)$. We will often refer to a partition by its vertex sets $(S_1, \ldots, S_k)$, which induce the subgraphs $G[S_1], \ldots, G[S_k]$. For a 2-partition $(S,T)$, the set of \textbf{cut edges}, denoted $E(S,T)$, consists of all edges with one endpoint in $S$ and the other in $T$.

Throughout this work, our primary example is the $m \times n$ \textbf{grid graph}, denoted $\gmn$. Its vertices are identified by coordinates $(i,j)$ for $i \in \{1,\ldots,m\}$ and $j \in \{1,\ldots,n\}$, and an edge exists between two vertices if their coordinates differ by exactly 1 in one dimension and are identical in the other. This definition also implies a natural planar embedding of the graph where vertices are embedded in their corresponding coordinates in $\mathbb{R}^2$.

Our analysis relies on the \textbf{dual graph}, $\dualgraph{G}$, of a planar graph $G$. The dual graph has a vertex for each face of $G$'s planar embedding and an edge connecting two dual vertices if their corresponding faces in $G$ share a primal edge. 

\begin{definition}[Dual Graph]
    The dual graph $\dualgraph{G}$ of a planar graph $G$ is a graph with a vertex $\dualvertex{f}$ for each face $f$ of the planar embedding of $G$, and an edge $\dual{e}=\dualedge{f_1}{f_2}$ for every edge $e\in E(G)$ that lies on the boundary of faces $f_1$ and $f_2$ of $G$.
    We say that $\dual{e}$ is the dual edge corresponding to $e$ and indicate this as $\dualcorresp{\dual{e}}{e}$.
\end{definition}
The dual of a dual edge is a primal edge. We will primarily distinguish between these two using the type of bracket used to refer to them, or with an asterisk.
We will use $\gmnd$ to refer to the dual of the grid $\gmn$. We will refer to the dual node corresponding to the outer face in the canonical embedding of $\gmn$ as $\outerface$.

\subsection{Key Definitions for Redistricting}

To evaluate and sample redistricting plans, we rely on several key metrics and probability distributions.

\begin{definition}[Imbalance $\imb{\cdot}$]
For a 2-partition $P=(S,T)$ of $G$, we will define the imbalance of $P$ as $\imb{P}=\frac{\abs{\abs{S}-\abs{T}}}{2}$.
\end{definition}

\begin{definition}[Spanning Tree Score and $\sp{\cdot}$]
For any graph $G$, define $\sp{G}$ to be the number of spanning trees of $G$.
For a $k$-partition $P=(G_1,\ldots,G_k)$ of $G$, define $\sp{P}=\prod_{i=1,\ldots,k}\sp{G_i}$ to be the spanning tree score of $P$.
\end{definition}

This score is used to define the probability distributions over partitions that we study.
The \textbf{Spanning Tree Distribution ($\spandist$)} samples a partition $P$ with probability proportional to its spanning tree score, i.e., $\Pr[P] \propto \sp{P}$. This distribution naturally favors compact districts, as partitions with fewer cut edges tend to have higher spanning tree scores.
However, on its own, $\spandist$ may sample partitions that are highly imbalanced.
Balance is typically achieved in practice by modifying this distribution in one of two ways : (1) rejecting sampled plans if they are imbalanced or (2) explicitly favoring balanced plans in the probability distribution.

The \textbf{$\lambda$-smooth Spanning Tree Distribution} (\lstd) is a variation that also accounts for population balance, falling in the latter category. It samples a partition $P$ with probability proportional to $\sp{P} e^{-\lambda \cdot \imb{P}}$. This can also be viewed as sampling from $\spandist$ and then applying a rejection step, accepting a partition $P$ with probability $e^{-\lambda \cdot \imb{P}}$.

\subsection{The Spanning Tree Distribution via Random Walks}
In this section, we present an adaptation of Wilson's algorithm~\cite{wilson_randomwalk,jamie_grid24} for sampling 2-partitions from the spanning tree distribution \spandist. Although there are already relatively simple algorithms to sample 2-partitions, we modify this specific algorithm to show our results. 
We first start with some simple definitions.

\begin{definition}[Walk~\cite{diestel}]
    A walk $W$ in a graph $G$ is an alternating sequence $v_1e_1v_2e_2v_3\ldots v_{m}$ of vertices and edges of $G$ such that $e_i=\{v_i,v_{i+1}\}$. If $W$ is non-empty, we say that the walk starts at $v_1$ and ends at $v_{m}$.
\end{definition}
Note that since dual graphs of simple graphs contain parallel edges, it will not suffice to consider a walk as merely a sequence of vertices. 
We now define the \textit{loop-erasure} of a walk. Put simply, it is the walk with the cycles removed one-by-one sequentially from the start of the walk. We will discuss loop-erasures and loops at length in \cref{sec:theorem5}. Until then, we will treat the loops slightly informally, for clarity of exposition. See \cref{fig:loop-example} for an example of a loop-erasure.

\begin{definition}[Loop-Erasure]\label{defn:loop-erasure}
    The loop-erasure $D$ of a walk $W=v_1e_1v_2e_2\ldots v_n$ is a directed path defined the following way : Initialize $D=\phi$. 
    For $i=1,\ldots, n$, add the directed edge $(v_{i-1},v_{i})$ to $D$.
    If this created a cycle in $D$, remove the edges of the cycle from $D$.
\end{definition}

\begin{figure}[htbp]
    \centering
    \includegraphics[width=0.3\linewidth]{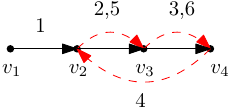}
    \caption{A walk from $v_1$ to $v_4$. The numbers above an edge represent the indices of the sequence where that edge appears. That is, the walk is $W=v_1(v_1,v_2)v_2(v_2,v_3)v_3(v_3,v_4)v_4(v_4,v_2)v_2(v_2,v_3)v_3(v_3,v_4)v_4$.
    The only loop of $W$ is drawn dashed in red.
    That is, the walk $v_2(v_2,v_3)v_3(v_3,v_4)v_4(v_4,v_2)v_2$ is a loop of $W$ whereas $v_3(v_3,v_4)v_4(v_4,v_2)v_2(v_2,v_3)v_3$ is not. The loop erasure is the directed path $\{(v_1,v_2),(v_2,v_3),(v_3,v_4)\}$.}
    \label{fig:loop-example}
\end{figure}

Wilson's algorithm samples a uniformly random spanning tree from a graph. This is shown in \cref{alg1}. We iteratively grow the tree $T$ starting from an arbitrary root whilst adding loop-erased random walks from some vertex not in $T$ so far. 
This returns a uniform spanning tree of any graph \cite{wilson_randomwalk}. We then recall the standard bijection between spanning trees of $G$ and its planar dual $\dualgraph{G}$: Given a spanning tree $T$ of $G$, we can construct a spanning tree $\dualgraph{T}$ of $\dualgraph{G}$ by adding dual edges $\dualedge{x}{y}$ if and only if $\dualcorresp{(u,v)}{\dualedge{x}{y}}$ is \textit{not} in $T$. Thus, we can also obtain a random spanning tree of $G$ by running Wilson's algorithm on the dual graph $\dualgraph{G}$ to get a random spanning tree of the dual graph. Selecting the corresponding primal spanning tree ensures that this selection is also uniformly random.
For the rest of this work, we will deal with random walks only on the dual graph $\dualgraph{G}$. As such, we will present Wilson's algorithm as running on the dual instead.

\begin{algorithm}[htbp]
\caption{Wilson's Algorithm on Dual Graph\label{alg1}}
\begin{algorithmic}[1]
    \STATE  Let $V(\dualgraph{T})\leftarrow \{\dualvertex{x}\}$ for an arbitrary root vertex $\dualvertex{x}\in V(\dualgraph{G})$.
    \WHILE{$V(\dualgraph{T})\neq V(\dualgraph{G})$}
    \STATE Start a random walk $\Pi$ from an arbitrary vertex $\dualvertex{y}\in V(\dualgraph{G})\setminus V(\dualgraph{T})$, until it reaches some vertex of $V(\dualgraph{T})$.
    \STATE
    Let $\pi$ be the loop-erasure of $\Pi$. Add all the edges of the $\pi$ to $\dualgraph{T}$.
    \ENDWHILE
    \RETURN primal spanning tree $T$ constructed as follows: For each edge $e \in E(G)$, $e \in T$ if and only if $\dualcorresp{\dualvertex{e}}{e}$
    is {\em not in} $\dualgraph{T}$.
    
\end{algorithmic}
\end{algorithm}

To proceed from a random spanning tree to a 2-partition, we first run Wilson's algorithm to get some spanning tree $\dualgraph{T}$. We then sample a random edge $e$ of the corresponding primal tree $T$. The subgraph $T\setminus e$ has two components, say $X,Y$. We output $(X,Y)$ with probability $\frac{1}{\abs{E(X,Y)}}$ and restart otherwise.
The probability that we get this partition before the rejection step is proportional to the number of spanning trees $T$ that are formed by the union of a spanning tree of $X$, a spanning tree of $Y$, and a cut edge between $(X,Y)$.
This is precisely $\sp{X}\cdot \sp{Y}\cdot E(X,Y)$. Since we reject this with probability $\frac{1}{E(X,Y)}$, we sample according to $\spandist$.

\section{Separation Fairness of \lstd on a Grid}

We will now  prove \Cref{thm:main-intro}. 
At a high level our approach has two steps. We first modify the random walk algorithm (\Cref{alg1}) for sampling from \spandist so that we can focus on the event that a particular edge of $G$ is cut.
We show a one-to-one correspondence between partitions of the primal grid and cycles in the dual graph. We then observe that if a partition assigns two specified vertices $u,v$ to different parts, then the corresponding dual cycle must also separate the vertices in a natural way.
We present this modification in \Cref{sec:modified-wilson}.
We then show that given a cycle $C$ that separates $u,v$, the random walk algorithm had a similar probability to output a cycle $C'$ that did \textit{not} separate $u,v$ while also corresponding to a (slightly) different partition.
For technical clarity, we divide this step into two sub-arguments.
In \Cref{sec:sep-prob} and \Cref{sec:theorem5}, we first demonstrate that the existence of a specific type of mapping between separating and non-separating partitions would imply a small separation probability for \lstd. Then, in \Cref{sec:reconnect}, we establish the existence of this mapping. The last two steps are the most technically involved and form the majority of the proof.

\subsection{Modified Random Walk Algorithm}\label{sec:modified-wilson}
We now present a modification of \Cref{alg1} that will be used in our analysis. This is shown in \Cref{alg2}. In this algorithm, we run a loop-erased random walk on the dual graph, where the root is a vertex $\dualvertex{x}$ and the random walk is started from a vertex $\dualvertex{y}$ adjacent to $\dualvertex{x}$. If the walk ever includes the edge $\dualedge{x}{y}$, we restart the walk. This process eventually outputs a path $\pi$ from $\dualvertex{y}$ to $\dualvertex{x}$ that does not include the edge $\dualedge{x}{y}$. 
If we continue Wilson's algorithm, we eventually obtain a spanning tree $\dual{T}$ that does not contain the edge $\dualedge{x}{y}$. The main observation, however, is that we will sample such a tree uniformly at random (among those that do not contain $\dualedge{x}{y}$). This follows from the fact that in Wilson's algorithm, the choice of the root and the vertex to start the walk from are arbitrary. Thus, this is equivalent to running Wilson's algorithm as described above, and discarding $\dual{T}$ if it includes $\dualedge{x}{y}$.

\begin{figure}[htbp]
    \centering
    \includegraphics[width=0.5\linewidth]{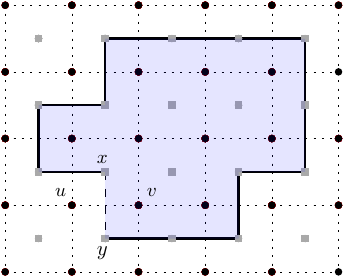}
    \caption{An example of a random walk $\pi$ (shown in solid black) started from $\dualvertex{x}$ to $\dualvertex{y}$. The shaded region is the interior of the cycle $\pi\cup \{\dualedge{x}{y}\}$.
    Here, $\dualcorresp{\dualedge{x}{y}}{(u,v)}$. The grey squares are dual nodes, and the black circles are primal nodes. }
    \label{fig:random-cycle}
\end{figure}

However, in \Cref{alg2}, we will not continue Wilson's algorithm after obtaining $\pi$. We will instead output a partition defined by $\pi$.
Consider the cycle $C=\pi\cup \{\dualedge{x}{y}\}$ in the dual graph $\gmnd$.
This cycle divides the primal vertices of the grid graph $\gmn$ into two sets, typically denoted $X$ and $Y$, corresponding to the vertices on the ``interior'' and ``exterior'' of $C$ (see \Cref{fig:random-cycle}).
Formally, in the standard planar embedding of $\gmn$, the dual cycle $C$ corresponds to a simple closed curve. By the Jordan Curve Theorem, this curve divides the plane into a bounded interior region and an unbounded exterior region. The primal vertices $(X,Y)$ are partitioned based on whether their embeddings lie in the interior or exterior region, respectively. We call $(X,Y)$ the 2-partition induced by the dual cycle $C$, or conversely, we call $C$ the dual cycle corresponding to the partition $(X,Y)$.

\begin{definition}[Dual Cycle of a 2-Partition]
\label{def:dual-cycle-of-partition}
Let $P=(X,Y)$ be a 2-partition of the primal graph $\gmn$ such that both $G[X]$ and $G[Y]$ are connected subgraphs. The \emph{dual cycle corresponding to $P$} is the set of dual edges $\dual{e} \in E(\gmnd)$ such that their corresponding primal edges $e \in E(\gmn)$ have one endpoint in $X$ and the other in $Y$ (i.e., $e$ is a cut edge of the partition $P$).
\end{definition}

We will thus view $(X,Y)$ as the partition corresponding to $C$.
We will also say that $C$ is the cycle sampled by \Cref{alg2}, if its corresponding partition $X,Y$ is sampled. We can then argue that \Cref{alg2} samples $(X,Y)$ proportional to its spanning score. For this, it is critical to sample the starting edge $\dualedge{x}{y}$ uniformly at random, and perform a normalization at the end.
Our normalization factors are defined as follows : For an edge $(u,v)\in E(G)$, let 
$\sp{G,2}$ denote the sum
$$
\sum_{X,Y}\sp{X}\cdot \sp{Y}
$$
where the sum is over all connected partitions $X,Y$.
Similarly, for any pair of vertices $u,v\in V(G)$, let $\sp{G,2}_{u\not\sim v}$  denote the same sum except the sum is over all connected partitions $X,Y$ such that $u\in X$, $v\in Y$.  
We note that these normalization factors do not need to be computed, since we do not intend to actually run our algorithm, only analyze it.\footnote{Nevertheless, the required factor \textit{can} be computed since $\sp{G,2}_{u\not\sim v}$ is simply the number of spanning trees of $G$ containing the edge $(u,v)$. For our purpose, we can replace $\sp{G,2}$ in the denominator with the total number of spanning trees instead. The ratio of $\sp{G,2}_{u\not\sim v}$ and the number of spanning trees is the \textit{effective resistance} between $u,v$ \cite{barrett2019spanning2forestsresistancedistance}, which can be computed in polynomial time.}
We are now ready to state \Cref{alg2}.

\begin{algorithm}[htbp]
\caption{Modified Random Walk Algorithm \label{alg2}}
\label{alg:balanced-partition}
\begin{algorithmic}[1]

    \STATE Sample a dual edge $\dualedge{x}{y}$ uniformly at random.
    \label{alg:sample-edge}
    \STATE Start a random walk $\Pi$ in the dual graph from $\dualvertex{x}$. Run the random walk till it hits vertex $\dualvertex{y}$\footnotemark. Let $\pi$ be the loop-erasure of $\Pi$.
    \IF{$\pi=\dualedge{x}{y}$\footnotemark}  
    \STATE Restart from Step 2.
    \ENDIF
    \STATE Let $C=\pi\cup \{\dualedge{x}{y}\}$. Let $X,Y$ be the vertices in the interior and exterior of $C$. Let $\dualcorresp{(u,v)}{\dualedge{x}{y}}$.
    \label{alg:random-walk-over}
    \STATE Accept with probability
    \[
    \frac{\sp{G,2}_{u\not\sim v}}{\sp{G,2}}\cdot \frac{1}{\abs{E(X,Y)}}.
    \]
    Otherwise, reject and restart.\label{alg:accept}
\end{algorithmic}
\end{algorithm}

\addtocounter{footnote}{-1}

\footnotetext{Although not relevant for the case of the grid, we remark that the algorithm is also applicable when the input graph has cut edges which lead to self-loops in the dual. We let $\pi=\{\}$ when these self-loops are sampled.}
\addtocounter{footnote}{1}
\footnotetext{More precisely, we mean the copy of $\dualedge{x}{y}$ corresponding to $(u,v)$. In the grid, this is only relevant when one of $u,v$ is a corner vertex. In this case, there are two parallel edges to $\outerface$ that correspond to different primal edges.}

\begin{claim}\label{claim:cycles-give-spanning-tree-distribution}
    \Cref{alg:balanced-partition} samples from \spandist.
\end{claim}
\begin{proof}

First, consider a fixed choice of $\dualedge{x}{y}$.
Suppose that $\dual{T}$ is any spanning tree that is a completion of $\pi$. That is, it is the spanning tree we would have obtained had we let Wilson's algorithm run until all dual vertices were visited. 
As mentioned above, since the choice of the roots and the starting vertex are arbitrary in Wilson's algorithm, the spanning tree $\dual{T}$ obtained will be uniformly chosen among the set of all spanning trees (of $\dual{G}$) that do not contain the edge $\dualedge{x}{y}$.
This implies that the corresponding primal spanning trees $T$ are uniformly chosen among those that contain the edge $\dualcorresp{(u,v)}{\dualedge{x}{y}}$.
Thus, the probability of getting a particular path $\pi$ is exactly proportional to the number of spanning trees $\dual{T}$ that are a completion of $\pi$, since all such trees do not contain the edge $\dualedge{x}{y}$ and any spanning tree is the completion of some path.

Now, consider any primal tree $T$ corresponding any dual tree $\dual{T}$ that is the completion of $\pi$.
Observe that the edge $(u,v)\in T$ splits $T$ into two connected components, where one component connects $X$ and the other $Y$. This is because none of the primal edges corresponding to $\pi$ belong to $T$.
Therefore, the number of such spanning trees that lead to the partition $(X,Y)$ is exactly $\sp{X}\cdot \sp{Y}$, the product of the number of spanning trees in the induced graphs of each region. Thus, we get a partition $X,Y$ of the primal vertices such that the following hold: 

\begin{enumerate}
     \item $(X,Y)$ is chosen with probability proportional to $\sp{X}\cdot \sp{Y}$, and 
     \item If $(u,v)$ is the primal edge corresponding to the dual edge $\dualedge{x}{y}$, then $u\in X$ and $v\in Y$.
 \end{enumerate}

Thus, for a given choice of $(u,v)$, a specific partition $P=X\cup Y$ is chosen with probability
\begin{align*}
    \frac{\sp{X}\cdot \sp{Y}}{\sum_{X,Y,u\in X,v\in Y}\sp{X}\cdot \sp{Y}}
\end{align*}
where the summation in the denominator is over all connected partitions $X,Y$. 
This term is precisely $\sp{G,2}_{u\not\sim v}$, one of the normalization terms we defined before.
Thus, conditioned on a particular $(u,v)$ being chosen in Step~1 of \Cref{alg2}, the probability of a partition $(X,Y)$ being accepted after the rejection step in Line~\ref{alg:accept} is
    \begin{align}\label{eqn:sampling-probability}
        \frac{\sp{X}\cdot \sp{Y}}{\sp{G,2}_{u\not\sim v}}\cdot \frac{\sp{G,2}_{u\not\sim v}}{\sp{G,2}}\cdot \frac{1}{\abs{E(X,Y)}} = \frac{\sp{X}\cdot \sp{Y}}{\abs{E(X,Y)}\sp{G,2}}
    \end{align}
The edge $(u,v)$ is chosen with probability $\frac{1}{E(G)}$. Further, the number of such edges $(u,v)$ that can lead to a particular $(X,Y)$ being sampled is $E(X,Y)$, since these are the only edges whose endpoints are in different parts of $(X,Y)$. Because these events are disjoint (since only one $(u,v)$ is chosen), the total probability of sampling $(X,Y)$ is
    \begin{align*}
    \frac{\sp{X}\cdot \sp{Y}}{\abs{E(G)}\sp{G,2}}.
    \end{align*}
    This completes the proof.
\end{proof}

\Cref{alg:balanced-partition} also has the following property, which will be useful later.

\begin{claim}\label{claim:uniform-edge-prob}
    For partition $P=(X,Y)$, let $C$ be the corresponding dual cycle. For any edge $\dualedge{x}{y}\in C$, conditioned on $P$ being accepted at Step~\ref{alg:accept} of \Cref{alg2}, edge $\dualedge{x}{y}$ was chosen at Step~\ref{alg:sample-edge} with probability $\frac{1}{\abs{C}}$.
\end{claim}
\begin{proof}
    From \Cref{eqn:sampling-probability}, the quantity $\Pr[\text{\Cref{alg2} outputs } P\mid \dualedge{x}{y} \text{ is sampled at Step~\ref{alg:sample-edge}}]$ is the same for every $\dualedge{x}{y}\in C$.
    Since every such $\dualedge{x}{y}$ is sampled with the same probability, and these are disjoint events, and $P$ is only sampled if some $\dualedge{x}{y}\in C$ is chosen,  it follows that $\Pr[\dualedge{x}{y} \text{ is sampled at Step~\ref{alg:sample-edge}}\mid \text{\Cref{alg2} outputs } P]=\frac{1}{\abs{C}}$.
\end{proof}

\subsection{The Unseparating Mapping and Cycles in the Dual}\label{sec:sep-prob}

\newcommand{\unsep}{unseparating\xspace}
\newcommand{\Unsep}{Unseparating\xspace}

In the last section, we showed that we can sample partitions by sampling cycles instead. Now, we argue how this can be used to show that \lstd has $\alpha$-separation fairness for constant $\alpha > 0$. 
Throughout this section, we will consider some arbitrary edge $(u,v)$ of $\gmn$. 
We will show that the probability that $u,v$ are separated is bounded away from 1.
In particular, we show that this will hold if an \textit{\unsep mapping} (defined below) exists.

Let $\mathcal{C}$ be some set of cycles of $\gmnd$.
For a cycle $C\in \mathcal{C}$, let $P(C)=(S,T)$ be the partition of $\gmn$ corresponding to $C$. 
With some abuse of notation, define $\imb{C}=\imb{P(C)}$ to be the imbalance of the partition corresponding to $C$.
For vertices $u,v\in \gmn$,
let $\mathcal{C}_{uv}$ be the subset of $\mathcal{C}$ such that for every $C\in \mathcal{C}_{uv}$, the two vertices $u,v$ are in the same part of $P(C)$.
To define the \unsep mapping, we first define a locality property.
\newcommand{\Local}{Local Difference\xspace}
\newcommand{\local}{local difference\xspace}
\newcommand{\locally}{locally different\xspace}

\begin{definition}[\Local]\label{def:local}
    Given two sets $S_1,S_2$ of (possibly directed) edges of $\gmnd$, we say that $S_1$ and $S_2$ are \locally on $\dualgraph{H}$ if $\dualgraph{H}$ is a subgrid\footnote{A subgrid of a graph is any subgraph that is isomorphic to $\boxplus_{a,b}$ for some positive integers $a,b$.} of $\gmnd\setminus \{\outerface\}$ such that every edge in $S_1\Delta S_2$ lies in $\dualgraph{H}$, except for edges incident on the outer face vertex $\outerface$.
\end{definition}

In other words, the sets agree everywhere outside $\dualgraph{H}$, except for edges incident on $\outerface$. However, we specifically apply the local difference property to sets $S_1,S_2$ that correspond to (un)directed simple paths that start and end outside $\dualgraph{H}$. This requirement also ensures that for some edge $\langle\dualvertex{v},\outerface\rangle\in S_1\Delta S_2$, $\dualvertex{v}$ must lie in $\dualgraph{H}$. See \cref{fig:local-difference}.

\begin{figure}[htbp]
    \centering
    \begin{subfigure}{0.49\textwidth}
        \centering
        \includegraphics[width=0.55\linewidth]{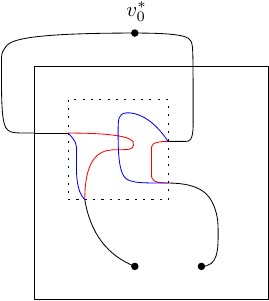}
        \caption{}
\end{subfigure}
\hfill
\begin{subfigure}{0.49\textwidth}
        \centering
        \includegraphics[width=0.55\linewidth]{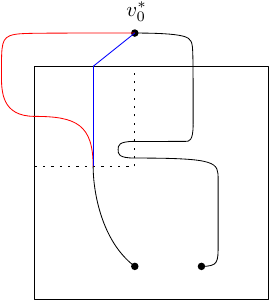}
        \caption{}
\end{subfigure}
\caption{Paths that are locally different on a subgrid $\dualgraph{H}$ (drawn with a dotted boundary). The red and blue portions are part of exactly one path, whereas the black portions are part of both.
Note that if $\dualgraph{H}$ was not adjacent (and part of) the boundary of the grid, then the paths cannot differ on any edges to $\outerface$.}
\label{fig:local-difference}
\end{figure}

\begin{definition}[\Unsep Mapping]
\label{def:unsep}
    An \unsep mapping for vertices $u,v\in \gmn$ on $\mathcal{C}$ with parameters $(\beta,\delta)$ is a mapping $f:\mathcal{C}\to \mathcal{C}_{uv}$ such that
    \begin{enumerate}
        \item If $f(C)=C'$, then $C,C'$ are \locally on some subgrid $\dualgraph{H}$ with at most $\beta$ edges. \label{unsep:differ-few-edges}
        \item If $f(C)=C'$, then the number of edges incident on $\outerface$ in $C'$ is at most the number in $C$. \label{unsep:dual-edges-unchanged}
        \item If $f(C)=C'$, then $\abs{\imb{C}-\imb{C'}}\leq \delta$. \label{unsep:imbalance}
        \item For any $C'\in \mathcal{C}_{uv}$, $\abs{\{C\in \mathcal{C} \mid f(C)=C' \}}\leq \beta$. \label{unsep:not-too-many-mapped}
    \end{enumerate} 
\end{definition}

Recall that we sample a cycle via a random walk. The ratio of the probabilities between two random walks is an inverse exponential function of the difference in the number of edges. 
Though we will use \unsep mappings on the \textit{loop-erasures} of random walks and not the walks themselves, we can show that properties~\ref{unsep:differ-few-edges} and \ref{unsep:dual-edges-unchanged} are sufficient to charge the probability of a random walk $W$ whose loop-erasure is $C\setminus\{u,v\}$ and separates $u,v$ to a random walk $W'$ whose loop-erasure is $C'\setminus \{u,v\}$ that does not separate $u,v$.

After the random walk, we have an additional rejection step proportional to $e^{-\lambda\cdot \imb{C}}$. Property~\ref{unsep:imbalance} ensures that these rejection factors are also similar in magnitude. These properties combine to ensure that given an \unsep mapping, we can lower bound the probability that two vertices \textit{are not} separated with some function of $\delta,\beta$ times the probability that they \textit{are} separated. An example of an \unsep mapping for two cycles is described in \Cref{fig:unsep}. 
In \Cref{sec:reconnect}, we prove \cref{thm:main-exact}, the existence of an unseparating mapping with constant parameters.

\begin{figure}[htbp]
    \centering
    \includegraphics[width=0.65\linewidth]{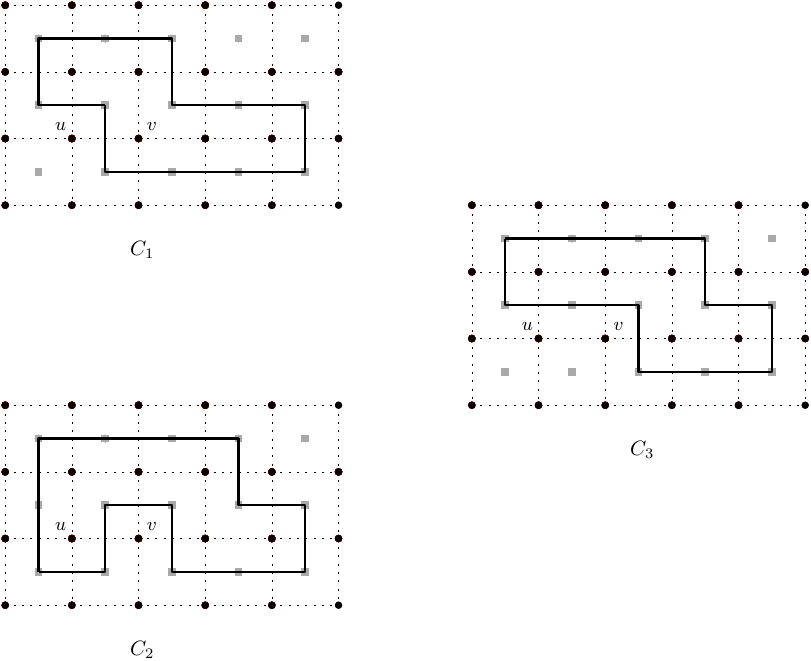}
    \caption{Two separating cycles $C_1,C_2$ that could potentially map to the same unseparating cycle $C_3$. In this case, $\delta=1$ since $\abs{\imb{C_2}-\imb{C_3}}=1$. We also have $\beta=\max(12,2)$ since edges are modified in a $3\times 3$ subgrid to go from $C_1$ to $C_3$ and two cycles map to $C_3$.}
    \label{fig:unsep}
\end{figure}

\begin{theorem}\label{thm:main-exact}
    There are universal constants $\beta,n_0$ such that for any $m,n\geq n_0$, there is a \unsep mapping for every pair of adjacent vertices $(u,v)$ in $\gmn$ on the set of all cycles of $\gmnd$ of length at least $2\beta$ with parameters $(\beta,2)$.
\end{theorem}

\subsection{Unseparating Mappings imply \Cref{thm:main-intro}}\label{sec:implies}
In the rest of this section, we assume \Cref{thm:main-exact} is true, and complete the proof of \Cref{thm:main-intro}. Again, we fix some edge $(u,v)$ in the primal grid. Recall that \Cref{alg2} constructs a loop-erased random walk in $\gmnd$ from a start vertex $\dualvertex{x}$ to a neighboring end vertex $\dualvertex{y}$, and completes the loop by adding edge $\dualedge{x}{y}$ to obtain a cycle. We will also need the following theorem relating unseparating mappings to random walks, whose proof we present in \Cref{sec:theorem5}.

\begin{restatable}{theorem}{walkcollapse} \label{thm:walk-collapse}
    Consider a pair of cycles $C$ and $C'$ of $\gmnd$ that are \locally on a subgrid $\dualgraph{H}$ with at most $\beta$ edges. 
    Let $\dualvertex{e}=\dualedge{x}{y}\in C\cap C'$ such that $\dualvertex{x},\dualvertex{y}$ lie outside $\dualgraph{H}$.
    Let $D,D'$ be the directed paths from $\dualvertex{x}$ to $\dualvertex{y}$ along $C\setminus\{\dualvertex{e}\},C'\setminus\{\dualvertex{e}\}$ respectively.
    Then there is a one-to-one function $g$ such that for every walk $W$ whose loop-erasure is $D$, there is a walk $\widehat{W} = g(W)$ whose loop-erasure is $D'$, such that $W,\widehat{W}$ differ in at most $3\beta^2$ edges\footnote{That is, the multiset difference of the directed edges in $W,\widehat{W}$ has size $3\beta^2$.}, and the number of edges incident to $\outerface$ in $\widehat{W}$ is at most that in $W$.
\end{restatable}

We assume \Cref{thm:main-exact,thm:walk-collapse} are true.
    \Cref{claim:cycles-give-spanning-tree-distribution} implies that that we can sample from \lstd as follows: First, run \Cref{alg:balanced-partition}. Accept the resulting partition $P$ with probability $e^{-\lambda\cdot \imb{P}}$ and restart otherwise.
    Fix some pair of adjacent vertices $u,v$ and let $f$ be the promised \unsep mapping for them. Let $\dualgraph{H}$ be the subgrid in Property~\ref{unsep:differ-few-edges}.

    We will analyze the event that in one run of
    \Cref{alg:balanced-partition}, we output a cycle $C$ that separates $u,v$.
    Our goal is to show that there was a similar probability to output a cycle $C'$ that did \textit{not} separate $u,v$. As we will show, this implies \cref{thm:main-intro}.
    First, assume that $C$ has length at least $2\beta$. In the unseparating mapping $f$ for $(u,v)$, let $f(C)=C'$. Since we assume \Cref{thm:main-exact} is true, this mapping exists. 
    Further suppose that in Step~\ref{alg:sample-edge}, we started with some dual edge $\dualvertex{e}=\dualedge{x}{y}\in C\cap C'$ outside $\dualgraph{H}$.
    Call this event $\E_{\dualvertex{e}\in C\cap C'}$.
    Since $\dualgraph{H}$ has size at most $\beta$, and  $\beta\leq \frac{\abs{C}}{2}$,
    there are at least $\frac{\abs{C}}{2}$
    possibilities for such an edge.
    
    Let $D,D'$ be the directed paths from $\dualvertex{x}$ to $\dualvertex{y}$ along $C\setminus\{\dualvertex{e}\},C'\setminus\{\dualvertex{e}\}$ respectively.
    \Cref{alg2} sampled $C$ only if it walked along some walk $\Pi$ 
    from $\dualvertex{x}$ to $\dualvertex{y}$ in Step~\ref{alg:random-walk-over} such that the loop-erasure of $\Pi$ was $D$.
    We will modify every such random walk to a walk $\Pi'$ whose loop-erasure is $D'$ by applying \Cref{thm:walk-collapse}.
    Note that  $\Pi'$ has at most $3\beta^2$ additional edges that were not in $\Pi$, and none of these edges were incident to $\outerface$.
    Thus, $$\Pr[\text{\Cref{alg2} walks along }\Pi'\mid \E_{\dualvertex{e}\in C\cap C'}]\geq \frac{\Pr[\text{\Cref{alg2} walks along }\Pi \mid \E_{\dualvertex{e}\in C\cap C'}]}{4^{3\beta^2}}$$
    where the probability is over the choices of \Cref{alg2} until Step~\ref{alg:random-walk-over}.
    This follows from the observation that the degree of the vertices of $\gmnd$ is at most 4, except for $\outerface$.
    
    Adding this up across all possible random walks $\Pi$, we get that the probability that \Cref{alg2} obtains $C'$ at Step~\ref{alg:random-walk-over} is at least $\frac{1}{4^{3\beta^2}}$ times the probability of getting $C$, again conditioning $\E_{\dualvertex{e}\in C\cap C'}$.
    Note that this step crucially requires that the function $h$ such that $h(\Pi)=\Pi'$ is one-to-one.
    Afterwards, in Step~\ref{alg:accept}, $C$ and $C'$ may be accepted with different probabilities since $|E(X,Y)|$ in this step could be different.
    If $P(C)=(X,Y)$ is the partition corresponding to $C$, note that $\abs{E(X,Y)}=\abs{C}$. Since $\abs{C'}-\abs{C}\leq \beta\leq \frac{\abs{C}}{2}$ because of \local, the acceptance probabilities are different by at most a factor of $3/2$.
    Therefore, we obtain
    \[
    \Pr[C' \text{ is obtained at Step~\ref{alg:accept}} \mid \E_{\dualvertex{e}\in C\cap C'} ] \geq \frac{2}{3\cdot 4^{3\beta^2}}\cdot \Pr[C \text{ is obtained at Step~\ref{alg:accept}} \mid \E_{\dualvertex{e}\in C\cap C'}].
    \] 
    We also have
    \begin{align*}
        \Pr[C \text{ is obtained at Step~\ref{alg:accept}} \mid \E_{\dualvertex{e}\in C\cap C'} ]&=\frac{\Pr[C \text{ is obtained at Step~\ref{alg:accept}  }\wedge \E_{\dualvertex{e}\in C\cap C'} ]}{\Pr[\E_{\dualvertex{e}\in C\cap C'}]}.
        \intertext{From \Cref{claim:uniform-edge-prob}, we get}
        \Pr[C \text{ is obtained at Step~\ref{alg:accept}} \mid \E_{\dualvertex{e}\in C\cap C'} ]&=\frac{\Pr[C \text{ is obtained at Step~\ref{alg:accept}}]\cdot \frac{\abs{C\cap C'}}{\abs{C}}}{\Pr[\E_{\dualvertex{e}\in C\cap C'}]}.
    \end{align*}
    The above similarly holds for $C'$.
    Combining this with our previous equation, we get
    \[
    \Pr[C' \text{ is obtained at Step~\ref{alg:accept}} ] \geq \frac{1}{4^{3\beta^2}}\cdot \frac{2}{3}\cdot \frac{1}{2} \Pr[C \text{ is obtained at Step~\ref{alg:accept}} ]
    \]
    because $\abs{C}-\abs{C'}\leq \frac{\beta}{2}\implies \frac{\abs{C'}}{\abs{C}}\geq \frac{1}{2}$.
    Now, we note that if $C$ is accepted with probability $p=e^{-\lambda\cdot \imb{C}}$, then $C'$ is accepted with probability at least $pe^{-\delta\lambda}$, from Property (\ref{unsep:imbalance}) in \Cref{def:unsep}.
    \begin{equation}\label{eqn:C-to-Cprime}
        \Pr_{\lstdmath}[ C' \text{ is sampled} ] \geq \frac{1}{3e^{\delta\lambda}4^{3\beta^2}}\Pr_{\lstdmath}[ C \text{ is sampled} ]
    \end{equation}
    We can make the above argument for every such $C$ of length at least $2\beta$ that separates $u,v$.
    Let $\E_{2\beta}$ denote the event that \lstd samples a cycle of length at least $2\beta$.
    Adding up the LHS of \Cref{eqn:C-to-Cprime} across all such $C$ with length at least $2\beta$, and noting that each $C'$ can appear at most $\beta$ times on the left hand side (from Property (\ref{unsep:not-too-many-mapped}) of \unsep mappings).
    \begin{equation}\label{eqn:C-to-Cprime2}
    \Pr_{\lstdmath}[u,v \text{ are not separated} \wedge \E_{2\beta}]\geq \frac{1}{3\beta e^{\delta\lambda}4^{3\beta^2}}\Pr_{\lstdmath}[u,v \text{ are separated} \wedge \E_{2\beta}].
    \end{equation}

    For the case when the sampled cycle has length at most $2\beta$, we simply assume that $u,v$ are separated.
    This will not significantly affect the separation probability if we can show that the probability of such an event, i.e. $\Pr[\neg \E_{2\beta}]$, is small.
    We will argue that for small $\beta$, the number of vertices contained inside $C$ is also small and thus, the probability that \lstd samples such a partition is small.
    This follows from the following claim, whose proof we leave to \Cref{app:missing-proofs}.
    
    \begin{restatable}{claim}{shortcycle}\label{claim:-short-cycle}
        For any cycle $C$ on $\gmnd$, if $\abs{C}\leq \min(m,n)-2$, then the number of primal nodes enclosed by $C$ is at most  $3\abs{C}^2$.
    \end{restatable}

    Since the number of vertices enclosed by $C$ is at most $12\beta^2$, then $\imb{C}\geq \frac{mn}{2}-12\beta^2$.
    For large enough $m,n$ and a constant $c$, such that $\beta\leq \sqrt{mn-\frac{c\log(mn)}{\lambda}}$, the probability that we ever sample such a cycle $C$ from \lstd can be upper bounded by $\frac{1}{2}$, by following the same arguments that we used to show that \lstd is almost-balanced with high probability (See \Cref{sec:imbalance}).
    Then we have
    \begin{align*}
        \Pr[u,v \text{ are separated}]-\Pr[u,v \text{ are separated}\wedge \E_{2\beta}]=\Pr[u,v \text{ are separated}\wedge \neg \E_{2\beta}]
        \leq \Pr[ \neg \E_{2\beta}] \leq \frac{1}{2}.
    \end{align*}
    Plugging this in \Cref{eqn:C-to-Cprime2}, we get
    \begin{align*}
    \Pr_{\lstdmath}[u,v \text{ are not separated} \wedge \E_{2\beta}]&\geq \frac{1}{3\beta e^{\delta\lambda}4^{3\beta^2}}(\Pr_{\lstdmath}[u,v \text{ are separated}] -\frac{1}{2})\\
    \implies \Pr_{\lstdmath}[u,v \text{ are not separated}]&\geq \frac{1}{3\beta e^{\delta\lambda}4^{3\beta^2}}(\Pr_{\lstdmath}[u,v \text{ are separated}]-\frac{1}{2})\\
    \implies\Pr_{\lstdmath}[u,v \text{ are not separated}]&\geq \frac{1}{2+6\beta e^{\delta\lambda}4^{3\beta^2}}.
    \end{align*}

    Since our choice of $u,v$ was arbitrary,  combining the above arguments with \Cref{thm:main-exact} implies \Cref{thm:main-intro}.

\subsection{Extension to Groups of Events}
In this section, we extend \cref{thm:main-intro} by showing that the probability that multiple edges are separated in a partition sampled from \lstd is also small.
We follow an argument similar to the proof of \cref{thm:main-intro}, except we apply the unseparating mapping multiple times. More specifically, instead of comparing the probability that an edge $e$ is separated with the probability that it is not separated, we will compare the probability that $q$ edges are simultaneously separated and the probability that only a subset of them are separated. 

\begin{theorem}\label{thm:group-separation}
    For constant $\lambda$, and $q=o(\sqrt{mn})$, and any set of $q$ edges of $\gmn$, each distance at least $\gamma+1$ apart (where $\gamma$ is the constant from \cref{thm:main-reconnect}), the probability that a 2-partition sampled from \lstd separates all the $q$ edges is $e^{-qc}$, for some constant $c$.
\end{theorem}
\begin{proof}
\newcommand{\edgeset}{\mathcal{E}}
    As before, we will assume the correctness of  \Cref{thm:main-exact,thm:walk-collapse}.
    Fix some set $Q$ of $q$ edges that are each at distance at least $\gamma+1$ from each other.
    Let $\edgeset$ be an arbitrary subset of $Q$ of size $\ell$ (to be determined later).
    Order the edges of $\edgeset$ as $e_1,\ldots, e_\ell$.
    Let $f_1,f_2,\ldots,f_\ell$ be the \unsep mappings for them.

    Consider some cycle $C$ output by 
    \Cref{alg:balanced-partition} that separates every edge in $\edgeset$.
    We assume that $C$ has length at least $2\ell\beta$. 
    Define a set of functions $F_i$ such that $F_0(C)=C$ and $F_i(C)=f_i(F_{i-1}(C))$ for all $1\leq i\leq \ell$.
    Since each unseparating mapping can decrease the length of the cycle by at most $\beta$ (because of Property~\ref{unsep:differ-few-edges}), this is well-defined.
    Since the subset $\edgeset$ and an arbitrary ordering of the initial $q$ edges uniquely defines $F_\ell(C)$, we will also refer to $F_\ell(C)$ as $F(\edgeset,C)$.
    Define $H_i$ to be the subgrid (from  Property~\ref{unsep:differ-few-edges}) where the cycles $F_i(C)$ and $F_{i-1}(C)$ differ. Note that since the edges of $\edgeset$ are at distance at least $\gamma+1$ from each other, the unseparating mappings have the additional property that the graphs $H_i$ are disjoint.
    
    We assume that in Step~\ref{alg:sample-edge}, we started with some dual edge from $C \cap F(\edgeset,C)$ outside of $\bigcup_{i=1}^{\ell} \dualgraph{H_i}$.
    Since $\bigcup_{i=1}^{\ell} \dualgraph{H_i}$ has size at most $\ell\beta$, and  $\abs{C}\geq 2\ell\beta$,
    there are at least $\frac{\abs{C}}{2}$
    possibilities for such an edge. 
    
    Like before, we now apply \Cref{thm:walk-collapse} to $F_i(C)$ and $F_{i+1}(C)$, sequentially for $i=0,1,\ldots, q$, so that for any walk $\Pi$ whose loop erasure was $C$, we get a unique walk $\Pi'$ whose loop erasure is $F(\edgeset,C)$, such that $\Pi'$ has at most  $3\ell\beta^2$ additional edges that were not in $\Pi$, and none of these edges were incident to $\outerface$.
    Repeating the same arguments as before, across all possible random walks $\Pi$,
    and noting that the population imbalance can increase by at most $\ell\delta$, we get
    \begin{equation}\label{eqn:group-C-to-Cprime}
        \Pr_{\lstdmath}[ F(\edgeset,C) \text{ is sampled} ] \geq \frac{1}{3e^{\ell\delta\lambda}4^{3\ell\beta^2}}\Pr_{\lstdmath}[ C \text{ is sampled} ]
    \end{equation}
    
    We now want to add up the above equation over different cycles $C$ and subsets $\edgeset$ of $Q$. To this end, consider when $F(\edgeset,C)=F(\edgeset',C')$ for some other subset $\edgeset'\subseteq Q$ of size $\ell$ and some other cycle $C'$ that separates all the edges in $Q$. Firstly, since every edge from $Q$ except those in $\edgeset$ are still separated in $F(\edgeset,C)$, and both $\edgeset$ and $\edgeset'$ have size $\ell$, it follows that $\edgeset=\edgeset'$. Thus, the only cycles $C'$ for which $F(\edgeset,C)=F(\edgeset,C')$ is when the unseparating mapping maps two cycles that separate one edge in $\edgeset$ to the same cycle that does not separate it. From Property (\ref{unsep:not-too-many-mapped}) of \unsep mappings, there are at most $\beta$ such cycles per edge in $\edgeset$, meaning that there are at most $\beta^\ell$ cycles $C'$ such that $F(\edgeset,C)=F(\edgeset,C')$.

    Recall that we assumed that $C$ has length at least $2\ell\beta$. Let $\E_{2\ell\beta}$ denote the event that \lstd samples a cycle of length at least $2\ell\beta$.
    Adding up \cref{eqn:group-C-to-Cprime} over all cycles $C$ of length at least $2\ell\beta$ that separate every edge in $Q$, and over all subsets $\edgeset$ of $Q$ of size $\ell$, we get 

    \begin{align*}
        \sum_{\edgeset}\sum_C\Pr_{\lstdmath}[ F(\edgeset,C) \text{ is sampled} ] &\geq \binom{q}{\ell}\sum_{C} \frac{1}{3e^{\ell\delta\lambda}4^{3\ell\beta^2}}\Pr_{\lstdmath}[ C \text{ is sampled} ]
    \end{align*}
    The RHS is precisely the probability that every edge in $Q$ is separated intersected with the event $\E_{2\ell\beta}$. The LHS is a lower bound on the complementary event (i.e. not every edge in $Q$ is separated), except we may be over counting each cycle $\beta^\ell$ times. We can bound this above by $\beta^\ell$. Thus,
    \begin{align}\label{eqn:group-c-cprime2}
        \beta^\ell &\geq \binom{q}{\ell} \frac{1}{3e^{\ell\delta\lambda}4^{3\ell\beta^2}}\Pr_{\lstdmath}[ \text{Every edge in } Q \text{ is separated} \wedge \E_{2\ell\beta}]
    \end{align}
    Next, \cref{claim:-short-cycle} implies that a cycle with length smaller than $2\ell\beta$ encloses at most $12\ell^2\beta^2$ vertices, implying $\imb{C}\geq \frac{mn}{2}-12\ell^2\beta^2$.
    Since $\beta,\lambda$ are constants, for large enough $m,n$ and $q= o(\sqrt{mn})$, we can argue similarly as in \Cref{sec:imbalance} that $\Pr[\neg \E_{2\ell\beta}]\leq \frac{1}{2}e^{-qc}$ for some constant $c$.
    Then, since we have
    \begin{align*}
        \Pr[ \text{Every edge in } Q \text{ is separated} \wedge \E_{2\ell\beta}] &\geq \Pr[ \text{Every edge in } Q \text{ is separated}] -\Pr[ \neg \E_{2\ell\beta}],
    \end{align*}
    it must hold that either $\Pr[ \text{Every edge in } Q \text{ is separated}]\leq e^{-qc}$ is true, in which case the proof is done, or it is the case that
    \begin{align*}
        \Pr[ \text{Every edge in } Q \text{ is separated} \wedge \E_{2\ell\beta}] &\geq \frac{1}{2} \Pr[ \text{Every edge in } Q \text{ is separated}].
    \end{align*}
    Plugging this back into \cref{eqn:group-c-cprime2}, we get
    \begin{align*}
        \Pr_{\lstdmath}[ \text{Every edge in } Q \text{ is separated}] &\leq \frac{6}{\binom{q}{\ell}} \left(\beta  e^{\delta\lambda} 4^{3\beta^2}\right)^\ell.
    \end{align*}
    Since $\beta,\lambda$ are constants, $c=\beta  e^{\delta\lambda} 4^{3\beta^2}$ is also a constant. We lower bound $\binom{q}{\ell}$ with $\left(\frac{q}{\ell}\right)^\ell$. The expression is then minimized at $\ell=\frac{q}{ec}$, giving us
    \begin{align*}
        \Pr_{\lstdmath}[ \text{Every edge in } Q \text{ is separated}] &\leq 6e^{-\frac{q}{ec}}
    \end{align*}
    as required.
\end{proof}

\section{Modifying Loop-Erased Random Walks : Proof of \Cref{thm:walk-collapse}}
\label{sec:theorem5}

To fully appreciate the non-trivial nature of \Cref{thm:walk-collapse}, we first outline several obstacles towards a possible proof.
The first obstacle is the distinction between undirected and directed cycles. Recall that \Cref{alg2}, after loop-erasure, samples a directed path $D$ between two adjacent dual vertices, $\dualvertex{x}$ and $\dualvertex{y}$.
We then interpret this as sampling the undirected version of the cycle formed by $D \cup \{\dualedge{x}{y}\}$. We do this for technical reasons, as undirected cycles significantly simplify the proof of existence unseparating mappings in \Cref{sec:reconnect}. However, a significant issue arises: two undirected cycles, $C_1$ and $C_2$, that differ by only a few edges can, when oriented, exhibit differences across a substantially larger number of directed edges. Consider the example (not on a grid) depicted in \Cref{fig:differing-cycles}.

\begin{figure}[H]
    \centering
    \includegraphics[width=0.35\linewidth]{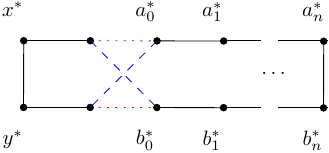}
    \caption{The two (undirected) cycles agree on the black solid edges. One cycle additionally uses the red dotted edges, while the other uses the blue dashed edges.}
    \label{fig:differing-cycles}
\end{figure}

If these cycles are directed to agree on the orientation of $\dualedge{x}{y}$, the number of directed edges in their set difference can become very large, as the direction of every edge of the form $\dualedge{a_i}{a_{i+1}}$ or $\dualedge{b_i}{b_{i+1}}$ would flip. This potentially introduces some complications in the proof of \Cref{thm:main-intro}. However, as we demonstrate in \Cref{lem:edge-cannot-flip}, this specific scenario cannot occur on a grid graph.

A second obstacle relates to the mapping between $W$ and $\widehat{W}$ in the statement of \Cref{thm:walk-collapse}. For this mapping to be one-to-one, every loop\footnote{The notion of a loop will be formally defined in \Cref{def:loop}. Informally, loops are the cycles in a walk that are removed by loop-erasure.} of $W$ must influence the construction of $\widehat{W}$. If, for example, we ignored the loops on some vertex $v$ in $D'\setminus D$ in the process of constructing $\widehat{W}$, this mapping would fail to be one-to-one because several walks that agree with $W$ everywhere except $v$ would all be mapped to $\widehat{W}$.
The natural solution is to traverse all the loops of $W$ in $\widehat{W}$. However, there may be a loop in $W$ that starts from some vertex in $D\setminus D'$ and never visits $D'$.
To walk along this loop, we need to also first visit some vertex on it.
We solve this issue by ensuring that $\widehat{W}$ visits every vertex in $D$. In particular, it will visit the vertices in $D\setminus D'$ on loops so that they do not appear on the loop-erasure of $\widehat{W}$.
We show this in \Cref{lem:directed-cycle}.

Thirdly, we must be careful in how we traverse these loops. For example, consider the instance in \Cref{fig:loop-wwhat1}, where $D,D'$ only differ on the path between $\dualvertex{x}_1$ and $\dualvertex{x}_2$. To resolve the previously discussed obstacle, $\widehat{W}$ can first walk along the red dashed path (from $D$), backtrack along it, then take the blue dotted path (from $D'$) and finally continue along the path common to both $D,D'$. This ensures that the loop-erasure of $\widehat{W}$ is $D'$ while also ensuring that we visit the vertices of $D$.

\begin{figure}[htbp]
    \centering
    \begin{subfigure}{0.49\textwidth}
        \centering
        \includegraphics[width=0.75\linewidth]{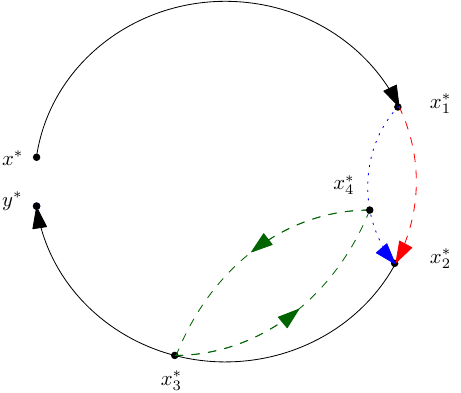}
        \caption{}
\label{fig:loop-wwhat1}
\end{subfigure}
\hfill
\begin{subfigure}{0.49\textwidth}
        \centering
        \includegraphics[width=0.75\linewidth]{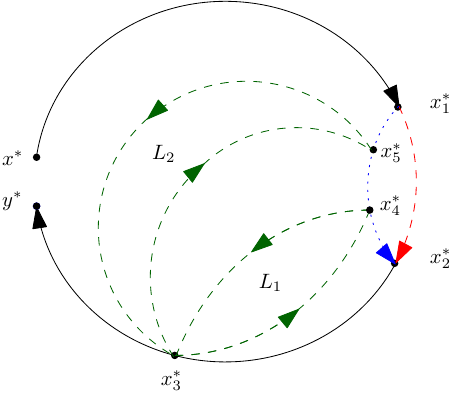}
        \caption{}
        \label{fig:loop-wwhat2}
\end{subfigure}
\caption{We consider two directed paths $D,D'$ from $\dualvertex{x}$ to $\dualvertex{y}$. Both of them share the black solid paths. The paths differ only between $\dualvertex{x}_1$ to $\dualvertex{x}_2$: $D$ contains the red dashed path whereas $D'$ contains the blue dotted path.}
    
\end{figure}

However, this strategy fails if $W$ contains the green dashed loop from $\dualvertex{x_3}$ to $\dualvertex{x_4}$ (which lies on the blue dotted path) and back to $\dualvertex{x_3}$.
If $\widehat{W}$ attempts to trace this green loop starting from $\dualvertex{x_3}$ \textit{after} the blue path segment, this creates a loop that includes the path from $\dualvertex{x_4}$ to $\dualvertex{x_2}$ to $\dualvertex{x_3}$. Since this path must appear on the loop-erasure, we will need to walk along this path again. This can add many edges to $\widehat{W}$ since the length of the $\dualvertex{x_2}$-$\dualvertex{x_3}$ path may be arbitrarily long.
The natural way to fix this is to walk along the green loop from $\dualvertex{x_4}$ instead. 

Our last obstacle, illustrated in \Cref{fig:loop-wwhat2}, arises when $W$ contains multiple loops, say $L_1$ and $L_2$, originating from the same vertex ($\dualvertex{x_3}$), both intersecting $D'\setminus D$ (the blue path).
A naive approach where $\widehat{W}$ incorporates these loops as it encounters their respective intersections on $D'$ (e.g., $L_1$ at $\dualvertex{x_4}$, $L_2$ at $\dualvertex{x_5}$) will not preserve the original sequence of $L_1$ and $L_2$ from $W$. Our construction must distinguish whether $W$ encountered $L_1$ or $L_2$ first from $\dualvertex{x_3}$ for the mapping to be one-to-one.

Our approach to this specific example is as follows: if $L_1$ appeared before $L_2$ in $W$, then $\widehat{W}$ traverses $L_2$ when it reaches $\dualvertex{x_5}$, and later traverses $L_1$ when it reaches $\dualvertex{x_4}$.
If $L_2$ appeared before $L_1$, upon reaching $\dualvertex{x_5}$ we walk along $L_2$ from $\dualvertex{x_5}$ to $\dualvertex{x_3}$, walk along the loop $L_1$, return to $\dualvertex{x_5}$ along the remaining part of $L_2$.
We describe our complete algorithm in \Cref{alg:w-to-w-hat} after presenting the formal definitions and some key lemmas.

\newcommand{\final}{N}
\begin{restatable}{lemma}{edgecannotflip}\label{lem:edge-cannot-flip}
    For undirected simple paths $P,P'$ from adjacent vertices $\dualvertex{x}$ to $\dualvertex{y}$ in $\gmnd$ that are \locally in a subgrid $\dualgraph{H}$, where $\dualvertex{x},\dualvertex{y}\not\in V(\dualgraph{H})$,
    the directed orientations $D,D'$ of $P,P'$ directed from $\dualvertex{x}$ to $\dualvertex{y}$ satisfy
    that if $\dualedge{p}{q}\in D$ and $\dualedge{q}{p}\in D'$, 
    then $\dualvertex{p},\dualvertex{q}$ must be in $\dualgraph{H}$.
\end{restatable}
\begin{proof}
    As we follow the directed orientation $D$ starting at $\dualvertex{x}$, we must repeatedly enter and exit the subgrid $\dualgraph{H}$. Let $\dualvertex{v}_1, \ldots, \dualvertex{v}_{2\final}$ be the set of vertices where $D$ enters or exits $\dualgraph{H}$ in clockwise order along the boundary of $\dualgraph{H}$ starting from the first entry point, $\dualvertex{v}_1$. Define pairs of indices $(a_i,b_i)$ with $i=1, \ldots, \final$ such that $\dualvertex{v}_{a_i}$ and $\dualvertex{v}_{b_i}$ are the $i^{th}$ entry and exit vertex in $\dualgraph{H}$ along $D$. That is, the path $D$ starts from $\dualvertex{x}$, enters $\dualgraph{H}$ at $\dualvertex{v}_1=\dualvertex{v}_{a_1}$, walks along the inside of $\dualgraph{H}$, exits $\dualgraph{H}$ at $\dualvertex{v}_{b_1}$, walks along the outside of $\dualgraph{H}$, then reenters $\dualgraph{H}$ at $\dualvertex{v}_{a_2}$, and so on, until leaving $\dualvertex{H}$ through $\dualvertex{v}_{b_\final}$ and ending at $\dualvertex{y}$. Since $P'$ agrees with $P$ outside of $\dualgraph{H}$, $D'$ has the same entry and exit vertices, though they may appear in a different order and an entry vertex may become an exit vertex or vice versa. For each directed subpath $\dualvertex{v}_{b_i} \to \dualvertex{v}_{a_{i+1}}$ of $D$, $D'$ contains either the subpath $\dualvertex{v}_{b_i} \to \dualvertex{v}_{a_{i+1}}$ or the reversed subpath $\dualvertex{v}_{a_{i+1}} \to \dualvertex{v}_{b_i}$. It therefore suffices to show that the latter case is not possible.

    We first establish a condition for whether a set of values for the indices $a_1, \ldots, a_{\final}$ and $b_1, \ldots, b_{\final}$ is possible. For any $i=1, \ldots, \final$, we know that $D$ contains some path $\dualvertex{v}_{a_i} \to \dualvertex{v}_{b_i}$ within $\dualgraph{H}$. This path, viewed as a curve in the plane between the planar embeddings of the points, separates the interior of the planar embedding of $\dualgraph{H}$ into two disjoint, connected regions of the plane $R_1$ and $R_2$ (see \cref{fig:path-interior}). For any $j$ distinct from $i$, if $\dualvertex{v}_{a_j}$ and $\dualvertex{v}_{b_j}$ are in different regions, then the paths $\dualvertex{v}_{a_i} \to \dualvertex{v}_{b_i}$ and $\dualvertex{v}_{a_j} \to \dualvertex{v}_{b_j}$ must intersect from the Jordan Curve theorem. This is a contradiction since $D$ is simple, so $\dualvertex{v}_{a_j}$ and $\dualvertex{v}_{b_j}$ must either both be in $R_1$ or both be in $R_2$.
    
    Define sets $\{I_i\}_{i=1}^N$ where $I_i$ is the set of integers between $a_i$ and $b_i$ (inclusive). 
    Since the vertices $v_1, \ldots, v_{2 \final}$ are in clockwise order, the previous observation implies that the vertices that are within $R_1$ and $R_2$ correspond to the index sets $I_i$ and $\{1,2, \ldots,2 \final\} \setminus I_i$. Therefore, we can conclude that for $j\neq i$, we have either
    \begin{enumerate}
        \item $I_j \subset I_i$, when $a_j,b_j \in I_i$, or
        \item $I_i \cap I_j = \emptyset$ when $a_j, b_j \in \{1,2, \ldots,2N\}\setminus I_i$.
    \end{enumerate}
    In other words, the set family $\{I_i\}_{i=1}^{\final}$ is laminar. We also define the sets $\{\hat{I}_i\}_{i=1}^{\final}$ where $\hat{I}_i$ denotes the integers between $b_{i}$ and $a_{i+1}$ (inclusive). By looking at the paths $\dualvertex{v}_{b_i} \to \dualvertex{v}_{a_{i+1}}$ and making the same arguments as above with the \textit{exterior} of $\dualgraph{H}$ instead of the interior (see \cref{fig:path-exterior}), we can show that 
    $\{\hat{I}_i\}_{i=1}^{\final}$ is also laminar\footnote{For this, we need to define $a_{\final+1}=a_1$ and the path $\dualvertex{v}_{b_{\final}}\to \dualvertex{v}_{a_1}$ path as the union of the path from $\dualvertex{v}_{b_{\final}}$ to $\dualvertex{y}$, the edge $\dualedge{x}{y}$, and the path from $\dualvertex{x}$ to $\dualvertex{v}_{a_{1}}$.}.
    Note that $\outerface$ is not in $\dualgraph{H}$ because $\dualgraph{H}$ is a subgrid (unless it is a trivial subgrid, in which case the statement trivially follows). It also does not affect our arguments if $\outerface$ belongs to one of the $\dualvertex{v}_{b_i} \to \dualvertex{v}_{a_{i+1}}$ paths, since this path still divides the exterior of $\dualgraph{H}$ into two regions. 

    \begin{figure}[htbp]
    \centering
    \begin{subfigure}{0.49\textwidth}
        \centering
        \includegraphics[width=0.75\linewidth]{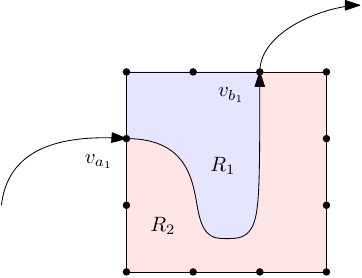}
        \caption{}
\label{fig:path-interior}
\end{subfigure}
\hfill
\begin{subfigure}{0.49\textwidth}
        \centering
        \includegraphics[width=0.55\linewidth]{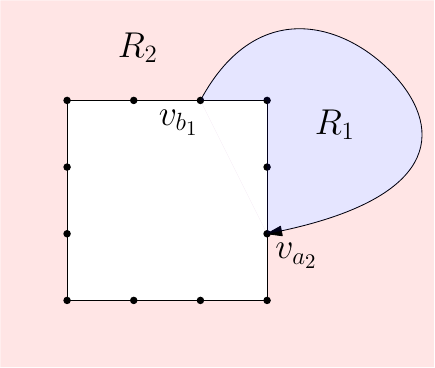}
        \caption{}
        \label{fig:path-exterior}
\end{subfigure}
\caption{The paths divide the interior and exterior of $\dualgraph{h}$ into regions $R_1$ and $R_2$.}
\label{fig:path-interior-exterior}
\end{figure}

    Lastly, we claim that each $a_i$ is odd and each $b_i$ is even. To do this, first note that each value in $\{1,2,\ldots,2 \final\}$ is used as an index exactly once. In, say $I_i$, the number of indices strictly between $a_i$ and $b_i$ is exactly $\abs{b_i-a_i}-1$. However, these indices come in pairs $(\dualvertex{v}_{a_j},\dualvertex{v}_{b_j})$ such that either both lie between $a_i$ and $b_i$, or neither do. Thus, $b_i-a_i-1$ must be even for each $1\leq i\leq N$. Similarly, $b_{i}-a_{i+1}-1$ is also even from the laminarity of $\{\hat{I}_i$\}.
    This implies that the terminal points of every set have opposite parity. Since we defined $a_1=1$,  $b_1$ must be even.
    Since $b_1-a_2-1$ is even,  $a_2$ must be odd, and so on. Thus, each $a_i$ is odd and each $b_i$ is even. In particular, it must be true that for any directed orientation with entry and exit vertices $\dualvertex{v}_1, \ldots, \dualvertex{v}_{2\final}$, the path enters at odd indices and exits at even indices.
    For example, it may be the case that in $D'$, after visiting $\dualvertex{v}_{a_1}$ the path traverses some edges in $\dualgraph{H}$ before exiting out of, say, $\dualvertex{v}_{b_3}$ (instead of $\dualvertex{v}_{b_2}$ as in $D$). 
    However, it cannot be the case that it exits out of some $\dualvertex{v}_{a_i}$ because that would imply that there are an odd number of exit/entry points between $\dualvertex{v}_{a_1}$ and $\dualvertex{v}_{a_i}$, causing a contradiction.
    Thus, for any path $\dualvertex{v}_{b_i} \to \dualvertex{v}_{a_{i+1}}$ of $D$, it must be that $D'$ also contains $\dualvertex{v}_{b_i} \to \dualvertex{v}_{a_{i+1}}$ since $\dualvertex{v}_{b_i}$ is fixed as an exit vertex.
\end{proof}
\begin{restatable}{lemma}{directedcycle}
\label{lem:directed-cycle}
    For directed simple paths $D,D'$ from vertices $\dualvertex{x}$ to $\dualvertex{y}$ in $\gmnd$ that are \locally in a subgrid $\dualgraph{H}$, where $\dualvertex{x},\dualvertex{y}\not\in V(\dualgraph{H})$, there exists a walk $W'$ whose loop-erasure is $D'$ and visits every vertex in $D\cup D'$, such that the number of edges incident on $\outerface$ in $W'$ is at most as that in $D$, and $W'$ contains at most $3\beta^2$ edges more than $D$.
\end{restatable}
\begin{proof}
    To get the requisite walk $W'$, we first follow $D$ from $\dualvertex{x}$ until we first hit a vertex $p_{0}\in H$. 
    We know that $D,D'$ completely agree outside of $\dualgraph{H}$ from \Cref{lem:edge-cannot-flip}, except for edges incident to $\outerface$.
    Thus, $W'$ contains the edges of $D$ from $\dualvertex{x}$ to $\dualvertex{p}_{0}$.
    
    Suppose that $D$ leaves $\dualgraph{H}$ at vertices $\dualvertex{p}_1,\dualvertex{p}_2,\ldots,\dualvertex{p}_{\final}$ and reenters at vertices $\dualvertex{q}_1,\dualvertex{q}_2,\ldots$.
    That is, $D$ first enters $\dualgraph{H}$ at $\dualvertex{p}_{0}$, walks some path (of possibly zero length) inside $\dualgraph{H}$, exits $\dualgraph{H}$ through $\dualvertex{p}_1$, walks some non-zero length path outside $\dualgraph{H}$, re-enters $\dualgraph{H}$ through $\dualvertex{q}_1$, walks some non-zero length path inside $\dualgraph{H}$, exits through $\dualvertex{p}_2$, and so on, until $D$ finally exits $\dualgraph{H}$ for the last time through $\dualvertex{p}_{\final}$.
    
     Recall that we described $W'$ till it reached $\dualvertex{p}_{0}$. From this point, we will follow $D$ inside $\dualgraph{H}$ from $\dualvertex{p}_{0}$ to $\dualvertex{p}_1$.
     Then, we walk along the boundary  of the subgrid $\dualgraph{H}$ from $\dualvertex{p}_1$ to $\dualvertex{q}_1$, and walk along $D$ from $\dualvertex{q}_1$ to $\dualvertex{p}_2$ and so on, until we reach $\dualvertex{p}_{\final}$.
     At this point, we backtrack along the path we took until we reach $\dualvertex{p}_{0}$. 
     In the loop-erasure, this path  will be erased because they must all be on some loop.
    Once we return to $\dualvertex{p}_{0}$, we now simply walk along $D'$, leaving through $\dualvertex{p}_{\final}$ to $\dualvertex{y}$.

    The number of vertices $\dualvertex{p}_i$ or $\dualvertex{q}_j$ is at most $\beta$, because they are distinct (since $D$ is a simple path) and since $\dualgraph{H}$ has at most $\beta$ edges (and hence, $O(\beta)$ vertices, since it is a subgrid).
    The path along the boundary from each $\dualvertex{p}$ vertex to the next $\dualvertex{q}$ vertex contains at most $\beta$ edges, again from the bound on the size of $\dualgraph{H}$. 
    Thus, the loops added from $\dualvertex{p}_{0}$ to $\dualvertex{p}_{\final}$ added at most $2\beta^2$ edges to $W'$ compared to $D$. The parts of $D'$ inside $\dualgraph{H}$ add at most another $\beta$ edges to $W'$ that were not in $D$.
    In total, we add at most $2\beta^2+\beta\leq 3\beta^2$ edges.
    This also visits every vertex in $D\cup D'$, with the exception of $\outerface$ if it belonged to $D\setminus D'$. Since $D,D'$ agreed outside $\dualgraph{H}$ except for edges incident to $\outerface$, it must be that $\outerface$ was adjacent to $\dualgraph{H}$ and it must have been visited by $D$ through some exit and re-entry vertex $\dualvertex{p}_i,\dualvertex{q}_i$. In this case, in the forward walk from $\dualvertex{p}_{0}$ to $\dualvertex{p}_{\final}$, we don't walk along the boundary from $\dualvertex{p}_i$ to $\dualvertex{q}_i$ but instead take a detour through the edges $(\dualvertex{p}_i,\outerface),(\outerface,\dualvertex{q}_i)\in D$. In the reverse walk, we don't take this detour and walk along the boundary from $\dualvertex{q}_i$ to $\dualvertex{p}_i$ as usual. This maintains that $W'$ has no more edges incident to $\outerface$ than $D$, while also ensuring that we visit $\outerface$.
\end{proof}

\subsection{Preliminaries for Analyzing Loops}
We now consider random walks and their loop-erasure more deeply.
For two walks $W_1,W_2$, if $W_1$ ends where $W_2$ starts, we use the notation $W_1W_2$ to denote the walk that continues along $W_2$ after first traversing $W_1$.
We also use the notation $W_1\subseteq W_2$ to indicate that $W_1$ is a substring of $W_2$. That is, if $W_2=w_1e_1w_2\ldots w_n$, then $W_1=w_ie_i\ldots w_j$ for some $i,j$ with $1\leq i\leq j\leq n$.
We will use $W_1=W_2[i:j]$ to indicate this.
If $j=n$, we will also say that $W_1=W_2[i:]$.

Recall that in the definition of the loop-erasure of a walk (\Cref{defn:loop-erasure}), we iterated over $i=1,\ldots,n$ and added edges $(v_i,v_{i+1})$ to our directed path, removing cycles whenever formed.
We now formally define loops.


\begin{definition}[Loops]\label{def:loop}
Given a walk $W=w_1e_1w_2e_2\ldots w_n$, 
we initialize the loop-erasure $D=\phi$. We also initialize an auxiliary set of indices $I_D=\phi$.
For $i=2,\ldots, n$, add the directed edge $e_{i-1}=(w_{i-1},w_{i})$ to $D$ and add the index $i-1$ to $I_D$.
If this created a cycle in $D$, remove the edges of the cycle from $D$ and the corresponding indices from $I_D$. Let $j$ be the smallest index among those removed. We say that $L=W[j:i]$ is a loop of $W$ starting at $j$. 
\end{definition}

A simpler and necessary, but not sufficient, definition is that a loop is a walk contained in $W$ that starts and ends at the same vertex, and whose edges were removed by the loop-erasure (not necessarily all at once). We recall the instance in \Cref{fig:loop-example} as an example. The only loop of $W$ is $W[2:5]$.
We say that $L_1=W[i_1:i_2],L_2=W[i_2:i_3],\ldots, L_m=W[i_{m}:i_{m+1}]$ are the loops at index $i_1$ if $L_1,\ldots, L_m$ are all loops of $W$, and there is no index $j$ such that $L[i_{m+1}:j]$ is a loop.
We will also say that these are the loops at $w_{i_1}$ when it is clear that we mean the appearance of $w_{i_1}$ at index $i_1$.

\begin{definition}[Maximal Loops]
    A maximal loop $L\subseteq W$ of a walk $W$ is a loop of $W$ that is not contained in any other loop of $W$. 
\end{definition}

Now, observe that we can break down any walk $W$ into its loop-erasure augmented with loops in $W$. Suppose the loop-erasure $D$ of a walk $W$ is $\{e_1,e_2,\ldots,e_m\}$ where $e_i=(w_i,w_{i+1})$.
Then in $W$, we start from $w_1$, walk along the (possibly zero) loops at $w_1$, walk along the edge $e_1$ to reach $w_2$ for the last time in the walk, walk along the loops at $w_2$ and so on.
Observe that such loops must be maximal, since the loop-erasure does not contain any cycles. Further, since these loops contain all edges erased by the loop-erasure, these must be the only maximal loops.


\begin{observation}\label{obs:maximal-loop-cannot-reverse}
    Consider any walk $W$ whose loop-erasure is $\{e_1,e_2,\ldots,e_{m}\}$ where $e_i=(w_i,w_{i+1})$. For any $i\in [m+1]$, every vertex $w_{i'}$ that appears on any maximal loop $L$ at $w_i$ satisfies $i'\geq i$.
\end{observation}
This follows from the aforementioned decomposition of a walk into the loop-erasure along with the loops. If a maximal loop at $w_i$ visited some $w_{i'}$ where $i'<i$, it would create a loop at $w_{i'}$ instead, since $w_{i'}$ is visited before $w_i$ in the loop-erasure, contradicting that this was a maximal loop at $w_i$.

\begin{definition}[Ordering $\prec_S$]
For any directed path or walk $S$, let  $\prec_S$ be an ordering on the vertices appearing in $S$ based on the sequence they appear in. We say that for $s_1\prec_S s_2$ if $s_1$ appears earlier than $s_2$ in the sequence or path $S$. If they appear multiple times, we only consider their first appearance.
\end{definition}

\begin{definition}[Splicing a walk]
    Let $W=w_1e_1w_2\ldots w_m$ be a walk and let $L=w'_1e'_1w'_2\ldots e'_{m'} w'_1$ be some walk starting and ending at the same vertex. For some  $1\leq i\leq m$, let $w_i$ be some vertex also appearing in both $W,L$.
    The walk obtained by splicing $L$ onto $W$ at position $i$ is the walk that traverses the first $i$ edges of $W$, then traverses the cyclic permutation of $L$ starting from the first appearance of $w_i$ in $L$, and then continues the traversal on $W$. That is, if $j=\min\{j': w'_{j'}=w_i\}$, it is the walk 
    $$
    w_1e_1w_2\ldots w_i\underbrace{\textcolor{red}{e'_jw'_{j+1}\ldots e'_{m'} w'_{1}e'_1w'_2\ldots e'_{j-1}w_i}}_{\text{Cyclic permutation of $L$}}e_iw_{i+1}\ldots w_m.
    $$    
\end{definition}
Note that the cyclic permutation is well-defined since $L$ is a sequence starts and ends at the same vertex.
We are ready to prove \Cref{thm:walk-collapse}.

\subsection{Proof of \Cref{thm:walk-collapse}}

    Let $\dualvertex{e}=\dualedge{x}{y}$.
    Let $D,D'$ be the directed paths from $\dualvertex{x}$ to $\dualvertex{y}$ along $C\setminus\{\dualvertex{e}\},C'\setminus\{\dualvertex{e}\}$ respectively.
    We will describe \Cref{alg:w-to-w-hat}, that takes as input a walk $W$ whose loop-erasure is $D$ and outputs a walk $\widehat{W}$ whose loop-erasure is $D'$. 
    We are given that $C,C'$ are \locally on a subgrid $\dualgraph{H}$ of size $\beta$.
    Apart from the guarantees on the number of edges in the theorem statement, we also need to guarantee that our algorithm does not output $\widehat{W}$ for any other input walk whose loop-erasure is $D$.
    We first observe that $D,D'$ satisfy \local in $\dualgraph{H}$ from \Cref{lem:edge-cannot-flip}. Hence, we can apply \cref{lem:directed-cycle}.
    The output $W'$ of \Cref{lem:directed-cycle} will be our starting point in \Cref{alg:w-to-w-hat}.
    \Cref{claim:wwhat-setdifference} proves the required properties about the number of edges in $\hat{W}$ compared to $W$. \Cref{claim:alg-wwhat-invariant} proves that the loop-erasure of $\widehat{W}$ is $D'$. Finally, \Cref{lem:alg-wwhat-oneone} proves that the mapping between $W$ and $\widehat{W}$ is one-to-one, completing the proof of \Cref{thm:walk-collapse}.
    We first walk through a run of the algorithm on an example below.
    
\renewcommand\algorithmiccomment[1]{
  \hfill\textcolor{blue}{\# \ #1}
}
\newcommand\LONGCOMMENT[2]{
  \hfill \begin{minipage}[t]\parbox{#1}{\textcolor{blue}{\# \ #2}}\strut\end{minipage}
}

\begin{algorithm}[htbp]
\caption{Algorithm for \Cref{thm:walk-collapse}}
\noindent \textbf{Input:} Paths $D,D'$ and a walk $W$ satisfying the requirements in the statement of \Cref{thm:walk-collapse}\\
\noindent \textbf{Output:} A walk $\widehat{W}=\widehat{W}(D,D',W)$ that is different from the output of any other input $(D,D',W)$.

\vspace{-0.3cm}
\hrulefill

\begin{algorithmic}[1]\label{alg:w-to-w-hat}
    \STATE Let $W'$ be the walk guaranteed by \Cref{lem:directed-cycle} for $D,D'$.
    \STATE For every vertex $u \in W'$, initialize an ordered set $\mathcal{M}_u = \{ ~ \}$.
    \STATE Initialize $\widehat{W} \gets W'$.
    \STATE Let the vertices of $D$ be indexed as $v_1,v_2,\ldots$, in order of their appearance (i.e. following $\prec_{D}$).
    
    \COMMENT{Phase 1: Identify loops from $W$ and assign them to vertices $u \in W'$.}

    \FOR{vertices $v_\ell \in D$ in ascending order of indices (i.e., following $\prec_D$)} \label{step:wwhat-iteration}
        \STATE Let $\mathcal{L}=\{L_1,L_2,\ldots,L_{j_{\max}}\}$ be the maximal loops in $W$ that start at $v_\ell$, indexed by their appearance order in $W$.
        \STATE Let $U=\{u_1,u_2,\ldots\}$ be the unique vertices from $W'$ that appear in any loop in $\mathcal{L}$, themselves indexed according to $\prec_{W'}$.
        
        \FOR{$u_i \in U$ in ascending order of indices (i.e., following $\prec_{W'}$)} 
            \IF {$\mathcal{L}$ is empty} 
                \STATE Break. 
            \ENDIF
            
            \STATE Let $j = \min \{ j' \mid 1 \le j' \le j_{\max}, u_i \in L_{j'} \}$.
            \IF[That is, if $u_i$ is not in any loop in $\mathcal{L}$.]{$j=\phi$} 
            \STATE Continue to the next iteration.
            \ENDIF
            \STATE Let $M_\ell = L_j L_{j+1} \ldots L_{j_{\max}}$.\label{step:w-what-minj}
            \STATE Add $(M_\ell,v_\ell)$ to $\mathcal{M}_{u_i}$ such that it is ordered before all previously added elements. 
            \STATE Remove $L_j, L_{j+1}, \ldots, L_{j_{\max}}$ from $\mathcal{L}$; update $j_{\max} \gets j-1$. 
        \ENDFOR
    \ENDFOR
    
    \COMMENT{Phase 2: Splice the collected loop sequences from $\mathcal{M}_u$ into $\widehat{W}$.}
    
    \FOR{vertices $u \in W'$ in order of $\prec_{W'}$} 
        \STATE Let $\mathcal{M}_u=\{(M_{i_1},v_{i_1}), (M_{i_2},v_{i_2}),\ldots (M_{i_{\min}},v_{i_{\min}})\}$. 
         \hfill \begin{minipage}[t]{0.35\textwidth}
        \parbox{\textwidth}{\textcolor{blue}{\# \ Note that $i_1 > i_2 > \ldots > i_{\min}$ and $M_{i_1}$ was added last.}}
        \end{minipage}

        \FOR{$(M_i,v_i) \in \mathcal{M}_u$ in descending order of indices}
            \STATE Let $j$ be the index of the earliest appearance of $u$ in $\widehat{W}$, ignoring any instances of $u$ added from splices.\label{step:w-what-earliest-j}
            \STATE Update $\widehat{W}$ by splicing $M_i$ into the current $\widehat{W}$ at $j$. \label{step:w-what-splice}
        \ENDFOR
    \ENDFOR
\end{algorithmic}
\end{algorithm}

\begin{figure}[htbp]
    \centering
    \begin{subfigure}{0.49\textwidth}
        \centering
        \includegraphics[width=0.75\linewidth]{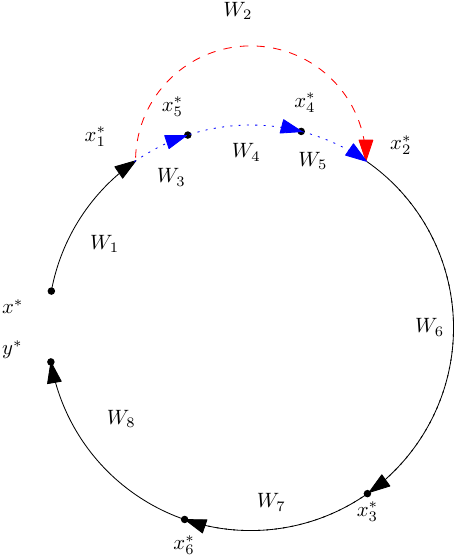}
        \caption{Both $D,D'$ share the solid black paths. The red dashed path is used only by $D$ whereas the blue dotted path is used only by $D'$.}
\label{fig:alg-wwhat1}
\end{subfigure}
\hfill
\begin{subfigure}{0.49\textwidth}
        \centering
        \includegraphics[width=0.75\linewidth]{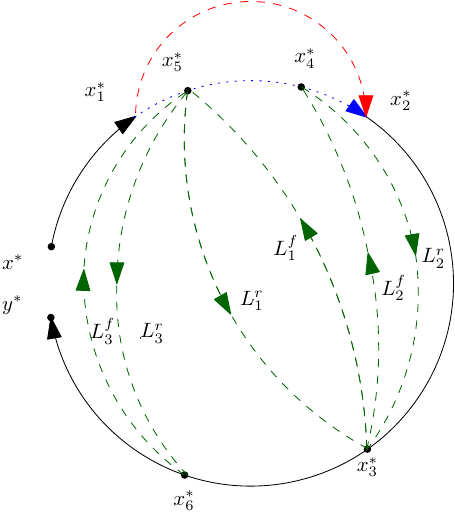}
        \caption{The walk $W$ has additional loops $L_1,L_2,L_3$. We will vary the precise location and sequence of these loops to illustrate the algorithm.}
        \label{fig:alg-wwhat2}
\end{subfigure}
\caption{Description of an instance. 
Each $W_i$ is a walk. Here, $D$ is the directed path corresponding to the walk $W_1W_2W_6W_7W_8$ and $D'$ is the directed path corresponding to the walk $W_1W_3W_4W_5W_6W_7W_8$.}
\label{fig:example-wwhat}
\end{figure}
\newcommand{\loopcol}[1]{\textcolor{ForestGreen}{#1}}

\begin{example}
    We will consider the example in \Cref{fig:example-wwhat}. For a marked walk $W_i$, let $\overline{W}_i$ denote the walk in reverse. For illustrative purposes, we assume that the $W'$ guaranteed by \Cref{lem:directed-cycle} is the walk $W_1W_2\overline{W_2}W_3W_4\ldots W_8$. 
    We describe the outputs of several cases depending on where the loops appear on $W$. 
    We break each loop $L_i$ into two walks $L_i^f$ and $L_i^r$ so that $L_i=L_i^fL_i^r$.
    \begin{itemize}
        \item $L_1,L_2$ appear in that sequence at $\dualvertex{x_3}$ and $L_3$ appears at $\dualvertex{x_6}$. That is, $W=W_1W_2W_6\loopcol{L_1L_2}W_7\loopcol{L_3}W_8$.
        The loops of the walk will be colored green in the exposition.
        \Cref{alg:w-to-w-hat} adds $(L_1L_2,\dualvertex{x_3})$ to $\mathcal{M}_{\dualvertex{x_5}}$. And then later adds $(L_3,\dualvertex{x_6})$. The output walk is 
        $$\widehat{W}=W_1W_2\overline{W_2}W_3\loopcol{L_1^rL_2L_1^fL_3^rL_3^f}W_4\ldots W_8.$$
        
        \item $L_2,L_1$ appear in that sequence at $\dualvertex{x_3}$ and $L_3$ appears at $\dualvertex{x_6}$. That is, $W=W_1W_2W_6\loopcol{L_2L_1}W_7\loopcol{L_3}W_8$.
        \Cref{alg:w-to-w-hat} first adds $(L_1,\dualvertex{x_3})$ to $\mathcal{M}_{\dualvertex{x_5}}$, then $(L_2,\dualvertex{x_3})$ to $\mathcal{M}_{\dualvertex{x_4}}$,  and then $(L_3,\dualvertex{x_6})$ to $\mathcal{M}_{\dualvertex{x_5}}$. The output walk is 
        $$\widehat{W}=W_1W_2\overline{W_2}W_3\loopcol{L_1^rL_1^fL_3^rL_3^f}W_4\loopcol{L_2^rL_2^f}W_5\ldots W_8.$$
        
        \item All loops appear at $\dualvertex{x_3}$. We consider the case where $W=W_1W_2W_6\loopcol{L_2L_1^fL_3^rL_3^fL_1^r}W_7W_8$.
        \Cref{alg:w-to-w-hat} first adds $(L_1^fL_3^rL_3^fL_1^r,\dualvertex{x_3})$ to $\mathcal{M}_{\dualvertex{x_5}}$, then $(L_2,\dualvertex{x_3})$ to $\mathcal{M}_{\dualvertex{x_4}}$. The output walk is 
        $$\widehat{W}=W_1W_2\overline{W_2}W_3\loopcol{L_3^rL_3^fL_1^rL_1^f}W_4\loopcol{L_2^rL_2^f}W_5\ldots W_8.$$
    \end{itemize}
\end{example}

We now argue that \Cref{alg:w-to-w-hat} has the requisite properties from \Cref{thm:walk-collapse}.

\begin{claim}\label{claim:wwhat-setdifference}
    In \Cref{alg:w-to-w-hat}, the multi-set difference between directed edges in the final $\widehat{W}$ and $W$ has at most $3\beta^2$ edges, and $\widehat{W}$ has no more edges to $\outerface$ than $W$.
\end{claim}
\begin{proof}
    At a high-level, \Cref{alg:w-to-w-hat} uses the output $W'$ from applying \Cref{lem:directed-cycle} as its starting point. Recall that $W'$ already satisfies two key properties -- its loop-erasure is $D'$ and it contains at most $3\beta^2$ edges more than $D$.
    Then, we note that we can decompose the edges of $W$ into the edges of the loop-erasure (which is $D$) along with the edges in the maximal loops. Since the final walk $\widehat{W}$ is obtained after splicing each maximal loop from $W$ onto $W'$, it follows that the multi-set difference between directed edges in the final $\widehat{W}$ and $W$ also has at most $3\beta^2$ edges. The required property regarding the number of edges incident on $\outerface$ also holds by a similar argument.
\end{proof}
Before we argue that the loop-erasure invariant is maintained, we first make the following observation about splicing loops onto walks.

\begin{claim}\label{obs:splicing-loops}
    Let $M$ be a walk starting and ending at some vertex $u$. For some walk $\widehat{W}$, if $u$ appears at index $i$, and the loop-erasure of $\widehat{W}[1:i]$ does not contain any vertex of $M$ (other than $u$), then the loop-erasure of $\widehat{W}[1:i]$ and the loop-erasure $\widehat{W}[1:i]M$ are the same.
    This also implies that the loop-erasure of $\widehat{W}$ and $\widehat{W}[1:i]M\widehat{W}[i:]$ are the same.
\end{claim}
\begin{proof}
    This follows from the definition (\Cref{defn:loop-erasure}) of the loop-erasure procedure. If the loop-erasure of of $\widehat{W}[1:i]$ is empty, then the loop-erasure of $\widehat{W}[1:i]M$ is also empty, since $M$ is a walk that starts and ends at the same vertex.
If the loop-erasure of $\widehat{W}[1:i]$ is some non-empty $D$, then observe that $u$ must be the last vertex of $D$. If no vertex in $M$ other than $u$ appears on $D$, then the only cycles removed after adding the edges of $M$ to $D$ will be cycles that start and end at $u$, and include no other vertex on $D$.
Further, since $M$ is a walk that starts and ends at $u$, all edges of it must be removed by the loop-erasure.
The last statement follows trivially from the definition of loop-erasures, since adding later edges only creates cycles with edges in $D$, and it is irrelevant whether some edges of $M$ were previously added and removed.
\end{proof}

\begin{claim}\label{claim:alg-wwhat-invariant}
    \Cref{alg:w-to-w-hat} always maintains the invariant that $\widehat{W}$ has loop-erasure $D'$.
\end{claim}
\begin{proof}
    Consider the first splice of some walk $M$ onto $u$ in \Cref{alg:w-to-w-hat}. At this point, $\widehat{W}=W'$.
    In step~\ref{step:w-what-minj} of \Cref{alg:w-to-w-hat}, observe that when $M$ is added to $\mathcal{M}_u$, it has the property that $u$ is the earliest vertex in $M$ according to $\prec_{W'}$.
    If there was some vertex $u'\prec_{W'} u$ that appeared in $M$, then $M$ would've been added to $\mathcal{M}_{u'}$ and hence, spliced onto $u'$ instead. From \Cref{obs:splicing-loops}, the updated $\widehat{W}$ after splicing $M$ has the same loop-erasure as the previous $\widehat{W}$ (which was $W'$).
    For further splices, in step~\ref{step:w-what-minj}, since we always add loops to the earliest appearance of $u$ that was \textit{not} added by a splice, we can repeat our arguments by noting that the loop-erasure of $\widehat{W}$ till this appearance of $u$ remains the same as that of the loop-erasure of $W'$ till the first appearance of $u$ in $W'$.
\end{proof}
\begin{lemma}\label{lem:alg-wwhat-oneone}
    For fixed directed paths $D,D'$ satisfying the requirements in \Cref{thm:walk-collapse}, no two walks $W$ with loop-erasure $D$ input to \Cref{alg:w-to-w-hat} yield the same output walk $\widehat{W}$.
\end{lemma}

We spend the rest of this section proving this lemma.
We prove this by finding an inverting procedure for \Cref{alg:w-to-w-hat}. That is, given $D,D'$, and a walk $\widehat{W}$ that is guaranteed to be the output of \Cref{alg:w-to-w-hat} with inputs $D,D'$ and a walk $W$, we will argue that we can uniquely reconstruct $W$.
We detail this procedure in \Cref{alg:one-to-one}, but first give some intuition.
Our idea will be to first reconstruct the ordered sets $\mathcal{M}_u$ that \Cref{alg:w-to-w-hat} constructed, and then argue that given these sets and $D,D'$, we can reconstruct $W$.
To construct these sets, we will need to observe the loops of $\widehat{W}$ and argue that the way they were added by \Cref{alg:w-to-w-hat} leaves a ``signature'' that we can detect, and thereafter recover the sets $\mathcal{M}_u$.
We first have the following definition, which are related to these signatures.

    \begin{definition}[Loop-Sequence and Breakdown]
        Given an ordered set of walks that start and end at the same vertex $\mathcal{L}=\{L_1,L_2,\ldots,L_m\}$, the loop-sequence of $\mathcal{L}$ is the ordered set $\{v_1,v_2,\ldots,v_m\}$ where for every $i\in [m]$, $v_i\in D$ is the earliest vertex according to $\prec_D$ appearing in $L_i$. If no such vertex exists, let $v_i=\varphi$.

        The breakdown of a loop-sequence is a maximal ordered set of indices $\{i_1,i_2,\ldots, i_n\}$ with $n\leq m$ such that for all $j\in \{1,2,\ldots,n\}$,
        \begin{enumerate}
            \item $v_{i_j}\neq \varphi$.
            \item $\forall~i_{j-1}<i'<i_{j}$, either $  v_{i_{j}}\preceq_D v_{i'}$ or $v_i'=\varphi$.
            \item $\forall~i'>i_j$, either $ v_{i_j}\prec_D v_{i'}$ or $v_i'=\varphi$.
        \end{enumerate}
    \end{definition}

    \begin{claim}
        The breakdown of any loop-sequence is unique.
    \end{claim}
    \begin{proof}
        Without loss of generality, assume that no $v_i=\varphi$, since they do not affect the breakdown. If $v_i=\varphi$ for every $i$, the breakdown is uniquely empty.

        For simplicity, replace the loop-sequence with real numbers $\{a_i\}_{i\in[m]}$  such that $v_i\prec_D v_{i'}\iff a_i< a_{i'}$. This does not affect the arguments.
        Suppose for a contradiction that $i_1,i_2,\ldots,$ and $j_1,j_2,\ldots,$ are two distinct breakdowns of the loop sequence.
        Suppose they agree until the $\ell^{th}$ element in their sequences.
        If there is no $\ell+1^{th}$ element in either sequence, they are the same sequence.
        If only one of them has an $\ell+1^{th}$ element, this contradicts the maximality of the other.
        Thus, we assume that both of them have an $\ell+1^{th}$ element, which differ.
        Without loss of generality, assume that $i_{\ell+1}<j_{\ell+1}$. 
        Then we have the following inequalities
        \begin{enumerate}
            \item $\forall~i'>i_{\ell+1}$, $ a_{i_{\ell+1}}< a_{i'}$.
            \item $\forall~i':i_{\ell}=j_\ell<i'<i_{\ell+1}<j_{\ell+1}$, $a_{j_{\ell+1}}\leq a_{i'}$.
        \end{enumerate}
        This is a contradiction since the second inequality implies that $a_{j_{\ell+1}}\leq a_{i_{\ell+1}}$ whereas the first inequality implies $a_{i_{\ell+1}}< a_{j_{\ell+1}}$.
    \end{proof}
\begin{algorithm}[htbp]
\caption{Algorithm to invert \Cref{alg:w-to-w-hat}}
\noindent \textbf{Input:} Paths $D,D'$ and a walk $\widehat{W}$ that was the output of \Cref{alg:w-to-w-hat} with inputs $D,D'$ and some unknown $W$.\\
\noindent \textbf{Output:} The walk $W$.

\vspace{-0.3cm}
\hrulefill

\begin{algorithmic}[1]\label{alg:one-to-one}
    \STATE Let $W'$ be the walk guaranteed by \Cref{lem:directed-cycle} for $D,D'$.
    \STATE Initialize $\widetilde{W}\gets \widehat{W}$.
    \STATE For every vertex $u \in W'$, initialize an ordered set $\mathcal{M}'_u = \{ ~ \}$.
    \FOR{vertices $u \in W'$ in order of $\prec_{W'}$}
    \STATE Let $\ell_u$ be the number of loops at the first appearance of $u$ on $W'$. 
    \STATE Let $i$ be the index of the earliest appearance of $u$ in $\widetilde{W}$. Let $L_1,L_2,\ldots,L_m$ be the loops at $i$ in $\widetilde{W}$.
    \IF {$m\leq \ell_u$}
    \STATE Set $\mathcal{M}'_u=\{ ~ \}$. 
    \ELSE 
    \STATE Let $\mathcal{L}=\{L_1,L_2,\ldots,L_{m-\ell_u}\}$.
    \STATE Let $\{j_1,j_2,\ldots,j_{q} \}$ be the breakdown of the loop-sequence of $\mathcal{L}$. Let $j_0:= 0$.
    \FOR{$p=1,2,\ldots,q$}\label{step:iteration-one-to-one}
    \STATE Let $N=L_{j_{p-1}+1}L_{j_{p-1}+2}\ldots L_{j_{p}}$.
    \STATE Let $M$ be the cyclic permutation of $N$ starting from the last appearance of $v_{j_p}$ in $N$.
    \STATE Add $(M,v_{j_p})$ to $\mathcal{M}'_u$ such that it is ordered before all previously added elements.
    \ENDFOR    
    \STATE Update $\widetilde{W}$ by removing the loops $L_1,L_2,\ldots,L_{m-\ell_u}$ at $i$. That is, if $L_{m-\ell_u}=\widetilde{W}[i':j']$, we set
    \[
    \widetilde{W}\gets \widetilde{W}[1:i]\widetilde{W}[j':]
    \]
    \ENDIF
    \ENDFOR
    
    \STATE Initialize  $W\gets D$.
    \STATE For every $v\in D$, define $\mathcal{S}_v=\{(M_i,u_i)\mid \exists u_i\in W', \text{ s.t. } (M_i,v)\in \mathcal{M}'_{u_i}\}$.\label{step:oneone-sv}
    \FOR{vertices $v \in D$ in order of $\prec_{D}$}
    \STATE Let $\mathcal{S}_v=\{(M_{i_1},u_{i_1}), (M_{i_2},u_{i_2}),\ldots (M_{i_{\min}})\}$ where $u_{i_{\min}}\prec_{W'}\ldots \prec_{W'} u_{i_2}\prec_{W'} u_{i_1}$. 
    \STATE Let $M=M_{i_1}M_{i_2}\ldots M_{i_{\min}}$
    \STATE Let $j$ be the index of the earliest appearance of $v$ in $W$, ignoring any instances of $v$ added from splices.
    \STATE Update $W$ by splicing $M$ into the current $W$ at $j$. \label{step:one-one-splice}
    \ENDFOR
\end{algorithmic}
\end{algorithm}

\begin{claim}\label{claim:one-one-mm'}
    For any walk $W$ such that \Cref{alg:w-to-w-hat} with input $D,D',W$ constructs sets $\{\mathcal{M}_{u}\}_{u\in W'}$ and yields a walk $\widehat{W}$, \Cref{alg:one-to-one} with input $D,D',\widehat{W}$ constructs sets $\{\mathcal{M}'_{u}\}_{u\in W'}$ such that $\forall~u\in W',\mathcal{M}'_u=\mathcal{M}_u$.
\end{claim}
\begin{proof}
    Let the vertices of $W'$ be indexed as $u_1,u_2,\ldots,u_m$ according to $\prec_{W'}$.
    We first show the following subclaim.

    \begin{subclaim}\label{subclaim:oneone}
    For any $i$,  if $\mathcal{M}_{u_{i'}}=\phi$ for all $i'<i$, then
    $\mathcal{M}_{u_i}=\mathcal{M}'_{u_i}$.    
    \end{subclaim}
    Observe that if $\mathcal{M}_{u_{i'}}=\phi$ for all $i'<i$, then $W'$ is unchanged by \Cref{alg:w-to-w-hat} until the first appearance of $u_i$, say at index $j$.
    Further, recall that in line~\ref{step:w-what-splice} of \Cref{alg:w-to-w-hat}, we always splice the walks at this first appearance of $u_i$.
    This means that if there were $\ell_{u_i}$ loops at $j$ in $W'$, and we added $m-\ell_{u_i}$ loops from the splice, these loops that were originally in $W'$ will be the last $\ell_{u_i}$ of the $m$ loops at $j$ in $\widehat{W}$.
    Thus, it suffices to consider the first $m-\ell_{u_i}$ loops at $j$ to recover $\mathcal{M}_{u_i}$,  if there are $m$ loops at $j$ in $\widehat{W}$.
    If $\mathcal{M}_{u_i}=\phi$, this also means that $m=\ell_{u_i}$ and \cref{alg:one-to-one} also sets $\mathcal{M}'_{u_i}=\phi$. We consider the alternative case.
    
    Consider step~\ref{step:w-what-minj} of \Cref{alg:w-to-w-hat}, where for some $v\in D$, we add $(M,v)$ to $\mathcal{M}_{u_i}$, where $M$ is a walk starting from $v$.
    Observe that $v$ cannot reappear in $M$ before the first appearance of $u_i$ in $M$. If it did, it would create a maximal loop that would not have been included in $M$ in step~\ref{step:w-what-minj}.
    Later, we spliced $M$ onto $W'$ at $j$, since it is the first appearance of $u_i$ in $W'$.
    Observe that from the definition of splicing, the corresponding walk $N$ added to $W'$ is the cyclic permutation of $M$ starting from the first appearance of $u_i$ in $M$.
    Since $M$ starts at $v$ and $v$ does not appear in $M$ before $u_i$ except at the start, it follows that $M$ is the cyclic permutation of $N$ from the last appearance of $v$ in $N$. 
    
    Let $\mathcal{M}_{u_{i}}=\{(M_p,v_p),(M_{p-1},v_{p-1}),\ldots,(M_1,v_1)\}$ so that $v_1\prec_D v_2\prec_D \ldots \prec_D v_p$.
    The walks that \Cref{alg:w-to-w-hat} adds to $W'$ at $j$ are the cyclic permutations of these walks $M_{p'}$ starting from the first appearance of $u_i$ in them.
    Let these cyclic permutations be $N_{p},N_{p-1},\ldots,N_1$ respectively. 
    Note that in $\widehat{W}$, we walk along $N_1$ first, then $N_2$ and so on.
    Let $\mathcal{L}=\{L_1,L_2,\ldots, L_{m-\ell_{u_i}}\}$ be the first $m-\ell_{u_i}$ loops at $j$ in $\widehat{W}$. Observe that each loop must have come from a single $N_{p}$ because each $N_{p}$ is a walk that starts and ends at $u_i$. It cannot be the case that part of some loop came from, say $N_1$, and another part came from $N_2$.
    This follows from an argument similar to \Cref{claim:alg-wwhat-invariant}. \Cref{alg:w-to-w-hat} always adds these loops to the earliest appearance of $u_i$ in $\widehat{W}$. Since $\mathcal{M}_{u_{i'}}$ is empty for all $i'<i$, this means that $W'$ is unchanged till the first appearance of $u_i$. Further, since $u_i$ is the earliest vertex in $N$ according to $\prec_{W'}$, this satisfies the conditions in \Cref{obs:splicing-loops}. Thus, adding $N$ does not affect the loop-erasure, meaning that $N$ must be composed of loops at $u_i$.
    
    Now, recall that \Cref{alg:w-to-w-hat} added these walks to $W'$ in such a way that $N_1$ appears first in $\widehat{W}$, $N_2$ appears just after that, and so on, so that $N_{p}$ appears last.
    Let $L_1,L_2,\ldots, L_{m'}$ be the loops added because of $N_{1}$. That is, $N_1=L_1L_2\ldots L_{m'}$. 
    We want to show that in the first iteration of step~\ref{step:iteration-one-to-one} in \Cref{alg:one-to-one}, we correctly identified $N_{1}$.
    For any $1\leq p\leq m'$, let $w_p$ be the smallest element in $L_p$ according to $\preceq_D$, so that $\{w_1,w_2,\ldots,w_{m'}\}$ is the loop sequence of these loops.
    Since each $L_p$ (or possibly a cyclic permutation of it) was a part of $M_1$, which was a maximal loop at $v_i$ in $W$, we can apply \Cref{obs:maximal-loop-cannot-reverse} to get that either $w_p=\varphi$ or $v_1\preceq_D w_p$.
    Further, since $N_1$ must be the cyclic permutation of $M_1$ from the first appearance of $v_1$, $v_1$ must appear on $L_{m'}$. If $v_1$ did not appear on $L_{m'}$ but only appeared on, say, $L_{m'-1}$ instead, this would mean that $N_1$ is the cyclic permutation of $M_1$ from the \textit{second} appearance of $u_i$.
    Thus, $w_{m'}=v_1$.
    We now argue that $m'$ is the first element of the breakdown of the loop-sequence of $\mathcal{L}$.
    This will complete our argument that $N_1$ was correctly identified in step~\ref{step:iteration-one-to-one}.
    To see this, observe that by construction of $\mathcal{M}_{u_i}$, $v_1\prec_D v_2 \prec _D \ldots \prec_D v_m$.
    Further, the walks $M_2,M_3,\ldots,$ were comprised of maximal loops from $v_2,v_3,\ldots$ respectively. From \Cref{obs:maximal-loop-cannot-reverse}, for any $p>m'$, it must be that $w_{m'}=v_1\prec_D w_p$, or $w_p=\varphi$.
    Combined with our earlier observation that  $0<p<m'\implies v_1\preceq_D w_p$, or $w_p=\varphi$, this means that $v_1$ is in the breakdown of the loop-sequence of $\mathcal{L}$.
    It also must be the first element of the breakdown, since if there exists some $w_p$ with $p<m'$ and $w_p \prec_D v_1$, this creates a contradiction with our observations.
    Repeating the arguments for $N_2,N_3,\ldots,$ in the same way completes the proof of \Cref{subclaim:oneone}, when $\mathcal{M}_{u_{i'}}=\phi$ for all $i'<i$. 
    
    Now, in the run of \Cref{alg:one-to-one}, let $\widetilde{W}_i$ be the walk stored in $\widetilde{W}$ before the $i^{th}$ iteration of the outermost loop, so that $\widetilde{W}_1=\widehat{W}$. We make the following subclaim.

\begin{subclaim}\label{claim:agree-after-iterations}
    For any $i$, let $j$ be the index of the earliest appearance of $u_i$ in $W'$. Then $W'[1:j]=\widetilde{W}_{i}[1:j]$.
\end{subclaim}
We prove this by induction. Let $j$ be the first appearance of $u_1$ in $\widehat{W}$. The base case is trivial since $\widetilde{W}_1=\widehat{W}$ and $\widehat{W}[1:j]=W'[1:j]$ because the earliest location that \Cref{alg:w-to-w-hat} adds loops to is at the first appearance of $u_1$.
For the inductive step for $i$, recall that after applying the induction hypothesis, we can apply \Cref{subclaim:oneone}, since from the perspective of \Cref{alg:one-to-one} and the previous proof, the situation is exactly the same as if $\mathcal{M}_{u_{i'}}=\phi$ for all $i'<i$.
Then, \Cref{subclaim:oneone} shows that $\mathcal{M}_{u_i}=\mathcal{M'}_{u_i}$. This means that \Cref{alg:one-to-one} correctly identifies the loops at $u_i$ and removes them in the $i^{th}$ iteration, proving the case of $i+1$ and completing the inductive step.

To complete the proof of \Cref{claim:one-one-mm'}, we do induction (on the statement of \cref{claim:one-one-mm'}) on $i$, noting that for the inductive step, we can use \Cref{claim:agree-after-iterations}, after which we simply repeat the proof of \Cref{subclaim:oneone}, since it is identical to the case where $\mathcal{M}_{u_{i'}}=\phi$ for every $i'<i$, from the perspective of \Cref{alg:one-to-one} and the proof.
\end{proof}

\Cref{claim:one-one-mm'} implies that given $\widehat{W},D,D'$, we can recover the ordered sets $\{\mathcal{M}_{u_i}\}_{u_i\in W'}$. These sets contain all the maximal loops of $W$. For the proof of \Cref{lem:alg-wwhat-oneone}, it only remains to argue that these sets uniquely define $W$. That is, \Cref{alg:w-to-w-hat} on two distinct walks will not create the same set of $\{\mathcal{M}_{u}\}_{u\in W}$.

Towards this last step, note that the sets $\{\mathcal{S}_v\}_{v\in D}$ constructed in step~\ref{step:oneone-sv} of \Cref{alg:one-to-one} is uniquely defined by $\{\mathcal{M}_u\}_{u\in W'}$.
Then, the following claim (which follows by construction of the walks $M$ in \Cref{alg:w-to-w-hat}) indicates that we can recover the maximal loops of $W$, as well as the sequence in which they appear in $W$.
Since $W$ uniquely decomposes into its loop-erasure (which we know to be $D$) and its maximal loops at each vertex on the loop-erasure, this recovers the decomposition of $W$, and hence, uniquely recovers $W$. 


\begin{claim}\label{claim:maximal-loops-order}
For $v\in D$, let $L_1,L_2,\ldots L_m$ be the maximal loops at $v$ in $W$ indexed in the order in which they appear in $W$. If $\mathcal{S}_v=\{(M_{i_1},v),(M_{i_2},v),\ldots, (M_{i_{\min}},v)\}$ where $(M_i,v)\in \mathcal{M}_{u_i}$ and $u_{i_{\min}}\prec_{W'}\ldots \prec_{W'}u_{i_2}\prec_{W'} u_{i_{1}}$, then 
$$M_{i_{i}}M_{i_2}\ldots M_{i_{\min}}=L_1L_2\ldots L_m.$$
\end{claim}
This concludes the proof of \Cref{lem:alg-wwhat-oneone}.
Along with 
\Cref{claim:wwhat-setdifference,claim:alg-wwhat-invariant}, this also completes the proof of \Cref{thm:walk-collapse}.

\section{Existence of Unseparating Mappings: Proof of \Cref{thm:main-exact}}\label{sec:reconnect}

For this section, we assume that we are given some cycle $C$ on the dual graph that corresponds to a partition that separates two specified vertices $u,v$.
Informally, our objective in this section is to show that this cycle can be transformed, through a constant number of modifications, to a cycle that does \textit{not} separate $u,v$, while not affecting the (im)balance of the partition too much.
We will implicitly show the second requirement in the statement of \cref{thm:main-exact}, that only a constant number of cycles separating $u,v$ get mapped to a certain cycle not separating $u,v$,
We do this by only making modifications to $C$ within a small local neighborhood (say, within distance $10$) from $u,v$. Since the number of cycles that differ from $C$ only in this local neighborhood and nowhere else is some constant, this is enough.
Our main theorem statement for this section is the following.
\begin{theorem}\label{thm:main-reconnect}
    There exists universal constants $\gamma,n_0$ such that if $m,n\geq n_0$, for any feasible partition $(S,T)$ of $\gmn$ such that $\abs{S},\abs{T}\geq n_0$,
    and for any adjacent vertices $u,v\in \gmn$ with $u\in S,v\in T$,
    there exists a feasible partition $(S',T')$ of $\gmn$ such that
    \begin{enumerate}[itemsep=0.05cm]
        \item $u,v\in S'$.
        \item Every vertex in $S\Delta S'$ is within distance $\gamma$ of $u$ and $v$.
        \item $\abs{\abs{S'}-\abs{S}}\in \{0,1,2,3\}$
        \item If $C$ is the dual cycle corresponding to $(S,T)$ and $C'$ is the dual cycle corresponding to $(S',T')$, then the degree of $\outerface$ in $C'$ is at most its degree in $C$.
    \end{enumerate}
\end{theorem}

In words, this means that any partition that separates $u,v$ can be modified locally (i.e. in the distance $\gamma=O(1)$-neighborhood of $u,v$) to create a new partition that does not separate $u,v$. Crucially, we also increase the imbalance of the initial partition by at most a constant.
We first prove \cref{thm:main-exact} assuming the correctness of \cref{thm:main-reconnect}.

\begin{proof}[Proof of \cref{thm:main-exact}]
Consider a dual cycle $C$ corresponding to a partition $(S,T)$. Since $\abs{E(S,T)}=\abs{C}$, and primal vertices of $\gmn$ have degree at most 4, it follows that $\abs{S},\abs{T}\geq \frac{\abs{C}}{4}$.
For the constant $n_0$ from \cref{thm:main-reconnect}, we ensure that $\beta\geq 4n_0$ in \cref{thm:main-exact} so that we can apply the theorem statement.

Our \unsep mapping $f$ is easily defined.
For any cycle $C$ corresponding to a partition $(S,T)$ that separates $u,v$, we apply \cref{thm:main-reconnect} to get a dual cycle $C'$ corresponding to partition $(S',T')$ that does not separate $u,v$. We let $f(C)=C'$.
We now show that this is an \unsep mapping.
The \Local property (\cref{def:local}) follows from the second property of \cref{thm:main-reconnect}.
Further, the number of edges incident on $\outerface$ does not change from the fourth property of \cref{thm:main-reconnect}.
We remark that the presence of $\outerface$ does not affect any of our arguments. In particular, the partition only changes within a distance of $\gamma$ from $u,v$ in the \textit{primal} graph.

We can think of moving from $(S,T)$ to $(S',T')$ by switching the vertices in $S\Delta S'$ from $S$ to $T'$ or from $T$ to $S'$ one at a time. Note that these switches may result in some intermediate infeasible partitions, but this does not matter.
Consider the dual edges corresponding to the primal edges incident to some vertex $w\in S\Delta S'$. 
Every such switch of $w$ requires adding or removing some or all of these dual edges from the cycle corresponding to the partition.
Further, the cycle corresponding to the initial partition $(S,T)$ is $C$ and the final partition $(S',T')$ is $C'$.
But since these are the \textit{only} dual edges affected, and since every vertex in $S\Delta S'$ is within distance $\gamma$ of $u,v$, it follows that $C,C'$ are \locally on a subgrid $\dualgraph{H}$ of dimension at most $2\gamma\times 2\gamma$ of $\gmnd$, which has $O(\gamma^2)$ edges.
This also implies that number of cycles $C$ such that $f(C)=C'$ is some constant. This is because we have the stronger property that \textit{every} cycle $C$ separating $u,v$ is modified only in $\dualgraph{H}$. The number of cycles that agree with $C'$ everywhere else except $\dualgraph{H}$ is at most the number of subgraphs of $\dualgraph{H}$, since any two such cycles must induce different subgraphs on $\dualgraph{H}$. This number is at most $2^{O(\gamma^2)}$.
The other properties of the \unsep mapping follow directly from the theorem statement. Choosing $\beta$ larger than $4n_0, 2^{O(\gamma^2)},O(\gamma^2)$ completes the proof.
\end{proof}

\subsection{Additional Definitions and Notation}

We spend the rest of this subsection proving \cref{thm:main-reconnect}.
Given any cycle $C$ on the dual graph, recall that it defines a partition of the primal vertices based on the its interior and exterior. We call this a \textit{feasible partition}, since it is a 2-partition of the graph into connected components.
Given a feasible partition, we can color the vertices of the two parts red and blue respectively.

Consider some cycle $C\in \mathcal{C}_{m,n}$ that separates $u,v$.
We assume that $u,v$ are at least one row/column away from the boundary of the grid. We will deal with the case where they are not in \Cref{sec:boundary-cases}. The local picture looks like \cref{fig:local-uv}. Square nodes are nodes of the dual graph. Since $C$ separates $u,v$, $C$ contains the dual edge corresponding to the primal edge $(u,v)$.

\begin{figure}[htbp]
    \centering
    \includegraphics[width=0.25\linewidth]{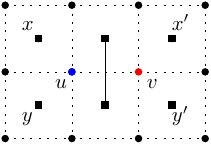}
    \caption{}
    \label{fig:local-uv}
\end{figure}

Dashed edges are edges of the primal graph, and solid edges are edges of the dual graph that are in $C$. 
A black circular node is a node whose color can be either red or blue, and for whom we currently do not know any other information. 

We do a case analysis based on the degrees of $x,y$ (and when necessary, $x'$ and $y'$). Recall that since $C$ is a cycle, every vertex in the dual graph has degree 0 or 2 in $C$. 
If a node other than $u,v$ is colored red or blue, then we mean that in the case being analyzed, that node is definitively in the part containing $u$ or $v$ respectively.
Note that if a dual edge is not drawn, it may or may not be in $C$. However, if both the endpoints of the primal edge corresponding to a dual edge are of the same color, then that naturally implies that the dual edge is not in $C$ in the case being analyzed.
We list all possible cases for $x,y$ in \cref{fig:case_analysis_uv_separated} as Cases 1-8. In some subcases, we will need to look at the degree of $x',y'$ as well. In this case, we will refer to them as `Case (3,2)' which denotes that the degrees of $x,y$ are in Case 3 and the degrees of $x',y'$ are in Case 2.
\renewcommand\thesubfigure{Case~\arabic{subfigure}}
\begin{figure}
     \centering
     \begin{subfigure}[b]{0.3\textwidth}
         \centering
         \includegraphics[width=\textwidth,cframe=green]{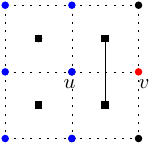}
         \caption{}
         \label{case:1}
     \end{subfigure}
     \hfill
     \begin{subfigure}[b]{0.3\textwidth}
         \centering
         \includegraphics[width=\textwidth,cframe=green]{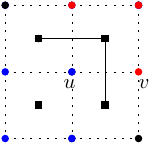}
         \caption{}
         \label{case:2}
     \end{subfigure}
     \hfill
    \begin{subfigure}[b]{0.3\textwidth}
         \centering
         \includegraphics[width=\textwidth,cframe=green]{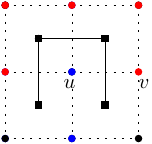}
         \caption{}
         \label{case:3}
     \end{subfigure}
     
     \bigskip
     
     \begin{subfigure}[b]{0.29\textwidth}
    \centering
\includegraphics[width=\textwidth,cframe=green]{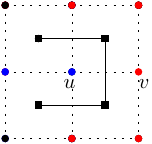}
         \caption{}
         \label{case:4}
     \end{subfigure}
     \hfill
    \begin{subfigure}[b]{0.347\textwidth}
         \centering
         \includegraphics[width=\textwidth]{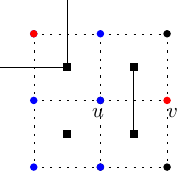}
         \caption{}
         \label{case:5}
     \end{subfigure}
     \hfill
\begin{subfigure}[b]{0.347\textwidth}
         \centering
         \includegraphics[width=\textwidth]{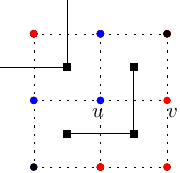}
         \caption{}
         \label{case:6}
     \end{subfigure}
     
     \bigskip 
     \begin{subfigure}[b]{0.3\textwidth}
         \centering
         \includegraphics[width=\textwidth]{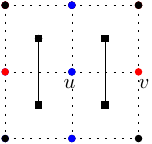}
         \caption{}
         \label{case:7}
     \end{subfigure}
     \hfill
         \begin{subfigure}[b]{0.3\textwidth}
         \centering
         \includegraphics[width=\textwidth]{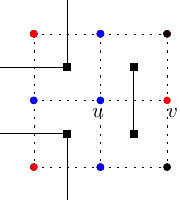}
         \caption{}
         \label{case:8}
     \end{subfigure}
     \hfill
     
    \caption{Cases for analysis. Easy cases are framed green.}
    \label{fig:case_analysis_uv_separated}
\end{figure}

\paragraph{Case Notation} 
Broadly, we have three branches for our case analysis. 
\begin{enumerate}
    \item If both $x,y$ have degree 0 - This is $\fbox{\mbox{Case 1}}$.
    \item If $x$ has degree 2 and $y$ has degree 0. This also covers the symmetric case where $y$ has degree 2 and $x$ has degree 0. We have 2 cases in this branch.
    \begin{enumerate}
        \item If $x$ has its right edge. This is $\fbox{\mbox{Case 2}}$.
        \item If $x$ does not have its right edge. This is $\fbox{\mbox{Case 5}}$.
    \end{enumerate}
    \item If both of $x,y$ have degree 2. We have 5 cases in this branch.
    \begin{enumerate}
        \item If $(x,y)\in C$ and exactly one of $x,y$ have its right edge\footnote{Note that both $x,y$ cannot have their right edge in this case, since the size of the blue component would be just 1.}. This is $\fbox{\mbox{Case 3}}$.
        \item If $(x,y)\in C$ and neither one of $x,y$ have its right edge. This is $\fbox{\mbox{Case 7}}$.
        \item If $(x,y)\not\in C$ and both of $x,y$ have its right edge. This is $\fbox{\mbox{Case 4}}$.
        \item If $(x,y)\not\in C$ and exactly one of $x,y$ have its right edge. This is $\fbox{\mbox{Case 6}}$.
        \item If $(x,y)\not\in C$ and neither one of $x,y$ have its right edge. This is $\fbox{\mbox{Case 8}}$.
    \end{enumerate}
\end{enumerate}

We number the cases so that the `easy' cases are together. In Cases 1-4,
we can directly modify the cycle to include $u$ into the red component. The graphs after the modification are drawn in \cref{fig:easy_cases_uv_separated}.

\begin{figure}[htbp]
     \centering
     \hfill
     \begin{subfigure}[b]{0.2\textwidth}
         \centering
         \includegraphics[width=\textwidth]{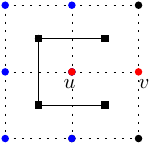}
         \caption{}
     \end{subfigure}
     \hfill
     \begin{subfigure}[b]{0.2\textwidth}
         \centering
         \includegraphics[width=\textwidth]{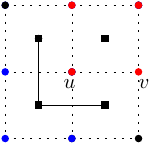}
         \caption{}
     \end{subfigure}
     \hfill
    \begin{subfigure}[b]{0.2\textwidth}
         \centering
         \includegraphics[width=\textwidth]{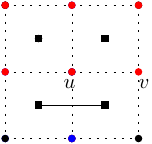}
         \caption{}
     \end{subfigure}
     \hfill
    \begin{subfigure}[b]{0.2\textwidth}
         \centering
         \includegraphics[width=\textwidth]{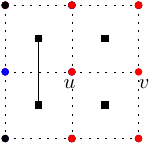}
         \caption{}
     \end{subfigure}
     \hfill
     
    \caption{Modifications for easy cases.}
    
    \label{fig:easy_cases_uv_separated}
\end{figure}

\renewcommand{\thesubfigure}{\alph{subfigure}}

We will use several tools to prove the harder cases. We first start with some formal definitions, before proving some lemmas that we will repeatedly use.

Recall that we colored the vertices of the feasible 2-partitions red and blue. We call these connected sets of vertices of the same color as regions. Though trivial for a feasible 2-partition, we will need this definition later when we consider intermediate infeasible partitions in our proof, where the all the vertices of one part are not necessarily connected.
That is, there are several `pockets' of red and blue vertices.
See \cref{fig:grid-definition-example} for an example. The coloring in \cref{fig:grid-definition-example-island} contains two blue regions. 
Note that in such infeasible 2-partitions, we have several dual cycles, instead of just one.

\begin{definition}[Region]
    A region is a maximal connected component of vertices of the same color. Any feasible partition has exactly one red and one blue region.
\end{definition}
\begin{definition}[Boundary]
The boundary of a region is the set of vertices of that region that do not have 4 adjacent vertices that are also from that region.     
\end{definition}

\begin{figure}[htbp]
     \centering
     \hfill
     \begin{subfigure}[b]{0.3\textwidth}
         \centering
         \includegraphics[width=\textwidth]{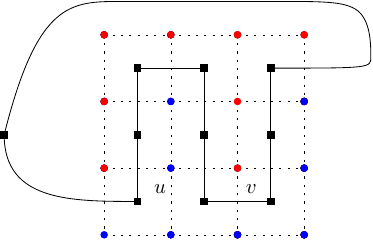}
         \caption{Starting feasible partition where $u,v$ are separated}
         \label{fig:grid-definition-example-disposable}
     \end{subfigure}
     \hfill
     \begin{subfigure}[b]{0.3\textwidth}
         \centering
         \includegraphics[width=\textwidth]{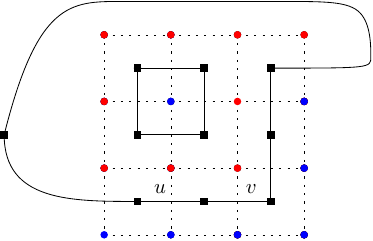}
         \caption{Intermediate infeasible partition where $u,v$ are not separated.}\label{fig:grid-definition-example-island}
     \end{subfigure}
     \hfill
    \begin{subfigure}[b]{0.3\textwidth}
         \centering
         \includegraphics[width=\textwidth]{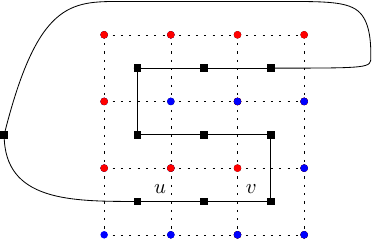}
         \caption{Final feasible partition where $u,v$ are not separated.}
     \end{subfigure}
     \hfill
     
    \caption{A full example on a $4\times 4$ grid.}
    
    \label{fig:grid-definition-example}
\end{figure}

\begin{definition}[Island]
	An island is a region encircled by vertices of the opposite color. Formally, a region $R$ is an island 
    if every vertex on the boundary of $R$ has degree 4. That is, no vertex from $R$ lies on the boundary of the grid.
\end{definition}

In \cref{fig:grid-definition-example-island}, the blue vertex that is not adjacent to any other blue vertices is an island.

\begin{definition}[Disposable Vertex and Sets]
    For a region $R$, a vertex $v\in R$ is called disposable if $R\setminus \{v\}$ is still a region. A set of vertices $R'\subseteq R$ is called disposable if $R\setminus R'$ is still a region.
\end{definition}

In \cref{fig:grid-definition-example-disposable}, $v$ is a disposable vertex for its region whereas $u$ is not. However, $u$ together with the blue vertex above it form a disposable set for the blue region.
We will usually drop the mention of the region $R$ and only mention that a vertex is a disposable when its region $R$ is obvious from context. Note that every vertex in a disposable set might not be individually disposable.

\begin{observation}\label{obs:degree1-disposable}
Every vertex $v\in R$ such that the degree of $v$ in the induced graph $G[R]$ is at most 1 is a disposable vertex.
\end{observation}

The easy cases we solved above were precisely the cases where $u$ was a disposable vertex. Changing the color of $u$ did not affect the feasibility of the partition.

\subsection{Technical Lemmas}
\label{sec:tools}
We now describe a few configurations that we can easily exploit, or rule out.
From here on out, we will consider structures that may appear in infeasible partitions.

\begin{lemma}[Cross-Structure]\label{lem:xoox}
    The coloring in \cref{fig:figure-lemma-xoox} cannot appear in any feasible partition.
    In an infeasible partition, if both vertices of one color (say, blue) vertices belong to the same region, then the two vertices of the other color (red) belong to two different regions, at least one of which is an island. 
    \begin{figure}[htbp]
        \centering
        \includegraphics[width=0.2\linewidth]{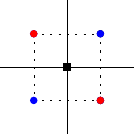}
        \caption{}
        \label{fig:figure-lemma-xoox}
    \end{figure}
\end{lemma}
\begin{proof}
    We start by proving the first part of the lemma.
    For contradiction, assume that the partition is feasible. Hence, both vertices of each color belong to one region.
    Let the red vertices be $r_1,r_2$ and the blue vertices be $b_1,b_2$.
    Consider the red vertex $r_1$ on the right. Since it is in the same region as $r_2$, there is a path $P_r$ between them in their region. Consider the vertex $r_x$ adjacent to $r_1$ in $P_r$. Without loss of generality, let $r_x$ be to the right of $r_1$. Similarly, consider the blue vertex $b_1$ on the right. If the vertex to the right of $b_1$ is on $P_b$, call it $b_x$.
    Thus, vertices $b_x$ and $r_x$ have an edge between them. In this case, the graph contains a $K_{3,3}$ minor on $r_1,r_2,r_x,b_1,b_2,b_x$  after contracting every vertex on $P_b\setminus\{b_1,b_2\}$ to $b_x$ and every vertex on $P_r\setminus\{r_1,r_2\}$ to $r_x$. This contradicts planarity.

    Suppose the vertex to the right of $b_1$ is not on $P_b$. Call it $z$.
    In this case, the vertex above $b_1$ must necessarily be on $P_b$.
    Call this $b_x$. 
    Let $y$ be the vertex to the right of $b_x$. If $y$ is not in $P_b$, we contract the edges $(y,z),(z,r_x)$. If $y$ is in $P_b$, we contract $(b_x,y)$.
    We again have a $K_{3,3}$ minor on $r_1,r_2,r_x,b_1,b_2,b_x$  after contracting every vertex on $P_b\setminus\{b_1,b_2\}$ to $b_x$ and every vertex on $P_r\setminus\{r_1,r_2\}$ to $r_x$.

    Thus, both vertices of each color cannot belong to the same region, implying that the partition is not feasible. Now, suppose that in an infeasible partition, the blue vertices $b_1$ and $b_2$ are in the same region. The red vertices must belong to different regions. Let $R_1,R_2$ be the red regions containing $r_1$ and $r_2$ respectively. We will show that at least one of them is an island. That is, both $R_1,R_2$ cannot have a vertex on the boundary of the grid. Suppose they do, let these vertices be $r'_1$ and $r'_2$. We add a new vertex to the grid, call it $z$, and add edges from $z$ to $r'_1$ and $r'_2$. This can be done without violating planarity, by adding $z$ on the outer face of the standard planar embedding of the grid. We now have a path between $r_1$ and $r_2$ through $z$, because $r_1$ and $r_2$ have paths to $r'_1$ and $r'_2$. The exact same arguments as before apply.
    
\end{proof}

\begin{lemma}[Island Walk]
	\label{lem:island-walk}
    Suppose a given (feasible or infeasible) partition has no cross-structures and there is an island $R$ of red vertices. Then there is a directed closed walk $P$ on the blue vertices such that
    \begin{enumerate}
        \item All blue vertices are from the same region.
		\item If the directed arc $(u,v)\in P$, then there is a vertex from $R$ to the right (facing the arc $(u,v)$ in the planar embedding of $G$) of $u$ or $v$.
		\item For every vertex $u$ on the boundary of $R$, there every blue vertex adjacent to $u$ is in $P$.
	\end{enumerate}
\end{lemma}
\begin{proof}
We prove this constructively. Define a \emph{spoke} as an edge from a red vertex on the boundary of the island to a neighboring blue vertex. We construct $P$ by visiting spokes one-by-one and adding edges to $P$. The spokes can be arranged in 3 different ways as below in \cref{fig:island-walk-cases}, and we show how to deal with them in each of the cases. Any other arrangement will lead to a cross-structure. The construction can be verified to satisfy the above properties.

Noting that each spoke is visited exactly once and the spokes eventually circle around the boundary, we get a closed walk. This again follows by considering the arrangements below in \cref{fig:island-walk-cases}. In each case, we always cover one spoke, and the number of spokes is finite. Furthermore, if it were the case that spokes looped around without visiting a particular spoke, we would get a walk that separates some of the red vertices from $R$ from the others, contradicting that $R$ was a region. Note that we crucially require the absence of cross-structures. 
An example construction of an island walk is drawn in \cref{fig:exmample_island_walk}.

\begin{figure}[H]
	\centering
	\hfill
	\begin{subfigure}[b]{0.1\textwidth}
		\centering
		\includegraphics[width=\textwidth]{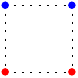}
	\end{subfigure}
	\hfill
	\begin{subfigure}[b]{0.1\textwidth}
		\centering
		\includegraphics[width=\textwidth]{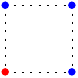}
	\end{subfigure}
	\hfill
	\begin{subfigure}[b]{0.1\textwidth}
		\centering
		\includegraphics[width=\textwidth]{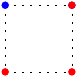}
	\end{subfigure}
	\hfill
	\bigskip

    \centering
	\hfill
		\begin{subfigure}[b]{0.1\textwidth}
		\centering
		\includegraphics[width=\textwidth]{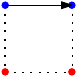}
	\end{subfigure}
	\hfill
	\begin{subfigure}[b]{0.1\textwidth}
		\centering
		\includegraphics[width=\textwidth]{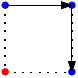}
	\end{subfigure}
	\hfill
	\begin{subfigure}[b]{0.1\textwidth}
		\centering
		\includegraphics[width=\textwidth]{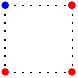}
	\end{subfigure}
    \hfill
	\bigskip 
    
	\caption{Possible arrangements of spokes and the construction of $P$ in the corresponding arrangments. Dashed edges are edges of the graph. Thick dashed edges are visited spokes.}
	
	\label{fig:island-walk-cases}
\end{figure}
\end{proof}

\begin{figure}[h]
\centering
\includegraphics[width=0.5\linewidth]{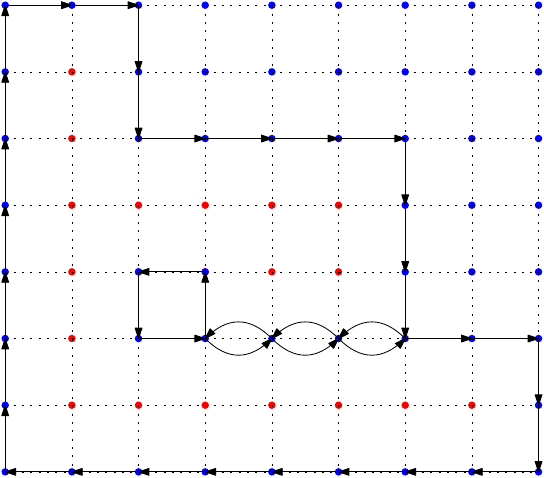}
\caption{Example of an island walk.}
\label{fig:exmample_island_walk}
\end{figure}

\begin{lemma}[1-thin structure]
\label{lem:1-thin}
    Suppose the following configuration (or a 90\textdegree  rotation of it) appears in an infeasible partition. Suppose the two red vertices are from different regions, one of which is an island. If the graph partition does not have any cross-structures, then the blue vertex is disposable and flipping its color does not add any cross-structures nor does it increase the degree of the outer face node $\outerface$.
    \begin{figure}[H]
        \centering
        \includegraphics[width=0.2\linewidth]{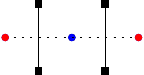}
    \end{figure}
\end{lemma}

\begin{proof}

    Let $u$ be the blue vertex and let $R$ be the region that $u$ belongs to. If the vertex above $u$ is red (or does not exist), then $u$ is disposable from \cref{obs:degree1-disposable} and we are done. Similarly for the vertex below $u$. Thus, we assume that both these vertices (say, $v,w$) are blue. Let $R_v$, $R_w$ be the sets of vertices in $R\setminus \{u\}$ reachable from $v$ and $w$ respectively in $R\setminus \{u\}$. If $R_v=R_w=R\setminus \{u\}$, then we are done and $u$ is disposable.

    Since one of the red vertices belongs to some region $R$ that is an island, then we know from \cref{lem:island-walk} that there is an island walk $P$ through $u$. If the arc $(u,v)$ and the arc $(v,u)$ are in $P$, then we have a contradiction because this would imply that both the red vertices are from $R$, since there must be a vertex from $R$ to the left of, and to the right of the one of the endpoints of the arc $(u,v)$.
	Thus, we must have only one of the arcs $(u,v)$ or $(v,u)$. 
    But in this case, it must be that $R_v=R_w$ because there is a cycle containing $u,v,w$ in $P$, because $P$ is a closed walk.

    In any of the cases above, it can be verified that flipping the color of $u$ does not add any cross-structures. Further, whether $u$ was adjacent to the boundary or not, flipping the color of $u$ to red does not add any edges to $\outerface$ regardless of the existence or color of $v,w$. 
\end{proof}

\begin{lemma}[2-thin structure]
\label{lem:2-thin}
    Suppose the following configuration (or a 90\textdegree  rotation of it) appears in an infeasible partition. Suppose the two red vertices are from different regions, one of which is an island. If the graph partition does not have any cross-structures, the blue vertices form a disposable set. Further, flipping both of them to red cannot create a cross-structure nor increase the degree of the outer face node $\outerface$.
    \begin{figure}[H]
        \centering
        \includegraphics[width=0.266\linewidth]{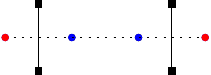}
    \end{figure}
\end{lemma}
\begin{proof}
    Let $u,v$ be the blue vertices and let $R$ be the region they both belong to. Let the red vertex to the right of $v$ belong to an island $R$. Let $a,b$ be the vertices above $u,v$ respectively and let $c,d$ be the vertices below $u,v$. If both of $a,b$ or both of $c,d$ are red, then $\{u,v\}$ is disposable. 
    
    Recall how the Island Walk $P$ was constructed in \cref{lem:island-walk}. If $b$ is a red vertex and $a$ is a blue vertex, then $P$ must contain the arcs $(v,u)$ and $(u,a)$.
    If the reverse edge $(a,u)$ is also present, then we can draw a contradiction similar to \cref{lem:1-thin}. If the reverse edge is not present, there must similarly be a cycle in $P$ that contains the edge $(u,v)$. In this case, $\{u,v\}$ is disposable. 
    If $b$ is a blue vertex and $a$ is red, $P$ must have the arc $(v,b)$. If $(b,v)$ is also present, then we again have a contradiction. If it is not present, then we have a cycle in $P$ that contains $v,b$. We can consider several cases depending on whether $u$ is in the cycle or not, and the colors of $c,d$. 
    In all the cases, we can argue that $\{u,v\}$ is disposable.
    
    Further, it can be verified that flipping them to red cannot add any cross-structures nor can it increase the degree of $\outerface$.
\end{proof}

In our case analysis, we will often deal with infeasible partitions where there is only one region of a certain color (say blue), and there are multiple regions of the other color (say red). In this case, we can flip the colors of the blue vertices in the thin structure to decrease the number of red regions. We call this \textit{resolving} the thin structure.
Note that this does not affect the blue region since these blue vertices were disposable.

\begin{lemma}[Elbow Lemma]\label{lem:elbow}
If a graph partition has a structure as in \cref{fig:elbow} and the highlighted vertex is not disposable, then its neighborhood must look like \cref{fig:elbow_resolution}.
If the graph partition is feasible and the highlighted vertex is disposable, flipping its color cannot create a cross-structure or increase the degree of $\outerface$, the outer face dual node.
\begin{figure}[H]
\hfill
\begin{subfigure}{0.2\textwidth}
    \centering
    \includegraphics[width=\linewidth]{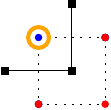}
    \caption{An elbow}
    \label{fig:elbow}
\end{subfigure}
\hfill
\begin{subfigure}{0.2\textwidth}
    \centering
    \includegraphics[width=\linewidth]{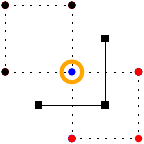}
    \caption{}
    \label{fig:elbow_borderless}
\end{subfigure}
\hfill
\begin{subfigure}{0.24\textwidth}
    \centering
    \includegraphics[width=\linewidth]{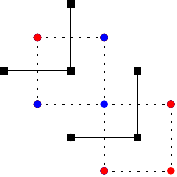}
    \caption{}
    \label{fig:elbow_resolution}
\end{subfigure}
\hfill
\bigskip
\caption{}
\end{figure}

\end{lemma}
\begin{proof}
    
This follows from a simple case analysis. If the highlighted vertex is adjacent to a border, then it looks like \cref{fig:elbow-boundary} (where the unmarked black vertex may not exist).
The blue vertex has degree at most $1$ in the induced graph of blue vertices, so it is disposable from \cref{obs:degree1-disposable}.
Otherwise, the structure resembles \cref{fig:elbow_borderless}. If the 3 black vertices in \cref{fig:elbow_borderless} are all blue (as in \cref{fig:elbow-1}), then removing the vertex cannot affect the region. Similarly if exactly one neighbor of the highlighted vertex is blue (like \cref{fig:elbow-3}), then it is disposable from \cref{obs:degree1-disposable}. If none of its neighbors are blue, then it forms a singleton region as in \cref{fig:elbow-4} and is hence, disposable.

\begin{figure}[H]

\begin{subfigure}{0.2\textwidth}
    \centering
    \includegraphics[width=\linewidth]{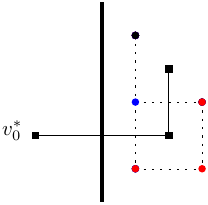}
    \caption{}
    \label{fig:elbow-boundary}
\end{subfigure}
\hfill
\begin{subfigure}{0.2\textwidth}
    \centering
    \includegraphics[width=\linewidth]{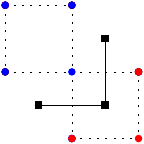}
    \caption{}
    \label{fig:elbow-1}
\end{subfigure}
\hfill
\begin{subfigure}{0.2\textwidth}
    \centering
    \includegraphics[width=\linewidth]{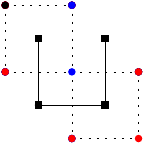}
    \caption{}
    \label{fig:elbow-3}
\end{subfigure}
\hfill
\begin{subfigure}{0.2\textwidth}
    \centering
    \includegraphics[width=\linewidth]{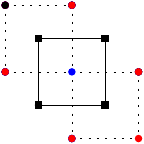}
    \caption{}
    \label{fig:elbow-4}
\end{subfigure}
\caption{}
\end{figure}

Now, assume that the graph partition is feasible and the highlighted blue vertex is disposable. Then the first part of the lemma shows that the neighborhood must look like any of \cref{fig:elbow-boundary,fig:elbow-1,fig:elbow-3,fig:elbow-4}. In this case, it can be easily verified that we can flip the color of the blue vertex without increasing the degree of $\outerface$ or creating cross-structures. The only case where we can create a cross-structure is in \cref{fig:elbow_resolution}. However, the partition must be infeasible in this case if the highlighted blue vertex was disposable.
\end{proof}

Often, in our case analysis, we will condition on whether a vertex like the highlighted one in \cref{fig:elbow} is disposable or not. If it is disposable, then \cref{lem:elbow} shows that we can safely flip its color without increasing the number of dual edges, or creating cross-structures. If it is not disposable, and given that our case analysis started from a feasible partition, \cref{lem:elbow} also implies that the highlighted vertex is not adjacent to the boundary. This crucial fact significantly decreases the number of cases, since we can primarily look at two cases: The first case where at least one of $u,v$ is adjacent to the boundary, and the second where neither of them are. This lets us avoid having to add additional cases where they are at larger distances away from the boundary.

\begin{lemma}\label{lem:create-island}
If a feasible partition contains any of the structures in \subref{case:5},\subref{case:6}, or \subref{case:7}, after flipping the vertex $u$ to red, the new partition has 1 red region and exactly 2 blue regions, of which at least 1 must be an island.

If it contains the structure in \subref{case:8}, after flipping the vertex $u$ to red, the new partition has exactly 1 red region and 3 blue regions, of which at least 2 are islands.
\end{lemma}
\begin{proof}
    For \subref{case:5} and \subref{case:6}, flipping $u$ to red creates a cross-structure. Note that adding $u$ to the red region makes it continue to be a region. Thus, since we started with a feasible partition, there continues to be exactly 1 red region.
    From \cref{lem:xoox}, the cross structure ensures that there are at least 2 blue regions (since it must not be a feasible partition) and one of them must be an island. We only need to argue that there are exactly two blue regions. 
    Let $x,y$ be the vertices above and to the left of $u$.
    Every vertex except $u$ in the region $R$ that $u$ belonged must be reachable from at least one of $x,y$, since they were all reachable from $u$ (by definition of $R$ being a region).
    For \subref{case:7}, we can come up with an argument similar to \cref{lem:xoox} to show that there must be at least 2 blue regions, and then we repeat the argument above.
    \subref{case:8} follows similarly follows from applying the arguments above from different cases simultaneously.
\end{proof}

\subsection{Illustrative Example: Case $(7,8)$}
\label{sec:case-78-mainbody}

We illustrate a simple case to showcase how we apply our tools.
Recall that Case (7,8) means that the neighborhood of $u$ is in \subref{case:7} and the neighborhood of $v$ is in \subref{case:8}.
We first flip $u$ to red. This creates two blue regions from \cref{lem:create-island}. Consider the highlighted vertices in \cref{fig:ex-fig1b}. If they were not in the same region, we have a 1-thin structure that we can resolve to get a partition with no change in the number of red or blue vertices (since we flipped 1 red vertex to blue and 1 blue vertex to red).
Thus, we assume they are in the same region.
Similarly, we assume that the highlighted vertices in \cref{fig:ex-fig1c} are in the same region.

\begin{figure}[H]
     \centering
     \hfill
     \begin{subfigure}[b]{0.3\textwidth}
         \centering
           \includegraphics[width=\linewidth]{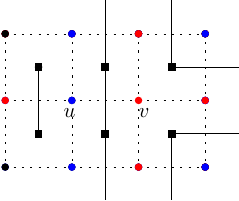}
         \caption{Step 1.}
         \label{fig:ex-fig1a}
     \end{subfigure}
     \hfill
     \begin{subfigure}[b]{0.3\textwidth}
         \centering
             \includegraphics[width=\linewidth]{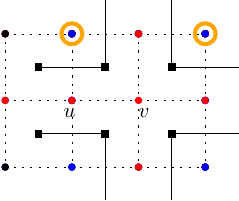}
         \caption{Step 2.}
         \label{fig:ex-fig1b}
     \end{subfigure}
     \hfill
    \begin{subfigure}[b]{0.3\textwidth}
         \centering
    \includegraphics[width=\linewidth]{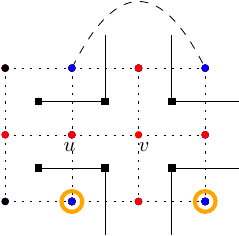}
         \caption{Step 3.}
         \label{fig:ex-fig1c}
     \end{subfigure}
     \hfill
     \bigskip
    \caption{}
    \label{}
\end{figure}

We have a 1-thin structure between the blue regions in \cref{fig:ex-fig2a}, and can flip the highlighted vertex. The final coloring is in \cref{fig:ex-fig2b}, and is a feasible partition, regardless of the color of the black vertex.

\begin{figure}[H]
     \centering
     \hfill
     \begin{subfigure}[b]{0.3\textwidth}
         \centering
           \includegraphics[width=\linewidth]{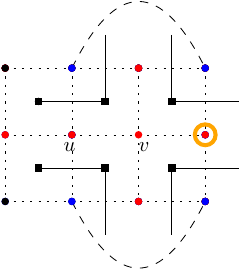}
         \caption{Step 4.}
         \label{fig:ex-fig2a}
     \end{subfigure}
     \hfill
     \begin{subfigure}[b]{0.3\textwidth}
         \centering
             \includegraphics[width=\linewidth]{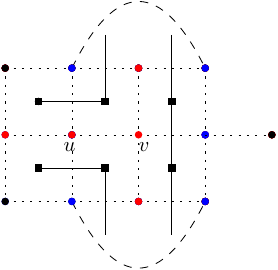}
         \caption{Final state.}
         \label{fig:ex-fig2b}
     \end{subfigure}
     \hfill
     \bigskip
    \caption{}
    \label{}
\end{figure}

This concludes Case (7,8). We look at all the other cases in \cref{sec:reconnect-cases}.
\section{Conclusion}\label{sec:discussion}

We showed that for the case of $k=2$, the $\lambda$-smooth spanning tree distribution achieves constant separation fairness. There are two main questions that arise. 

First, does our result extend to the \textit{exact} spanning tree distribution?
To extend our results to this setting, one approach would be to argue that there exists a mapping similar to an \Unsep Mapping (\Cref{def:unsep}) that does not increase the imbalance. However, unlike \unsep mappings, which we were able to create with local modifications to partitions (\Cref{thm:main-exact}), this cannot be done locally (See \Cref{fig:bad-rebalance}).
Another approach would require a mapping between partitions with imbalance $i$ and $i+1$, for $i=mn/2-O(1)$, with the restriction that the mappings are between partitions that differ only in a few edges and that every pre-image set of the mapping has constant size. As a concrete example, a result of the following kind would imply constant separation fairness on the spanning tree distribution restricted to partitions with an additive imbalance of at most 3. 
The fairness bound then follows by using this in the last part of the proof of \Cref{thm:main-intro}.

\newtheorem{Conjecture}{conjecture}
\begin{conjecture}
    Let $p_{j}$ denote the probability that the spanning tree distribution on the $m\times n$ gird samples a 2-partition where the smaller part has size $j$.
    Then for $i\in \{1,2,3\}$, $p_{mn/2}=\Theta(p_{mn/2-i})$.
\end{conjecture}

The second question is to show similar guarantees for larger $k$. Though we analyzed the spanning tree distribution on the grid only for 2-partitions, empirical suggest that the spanning tree distributions are also fair on real-world graphs, even with $k=14$ (See \Cref{fig:experiment}).
To show fairness for larger $k$ via our method, we would need to extend \Cref{thm:main-reconnect} to feasible partitions of connected subgraphs of grid graphs. Without new ideas, this would increase the number of cases in the proof of \Cref{thm:main-reconnect} by several-fold, which motivates the research question of developing techniques that go beyond case-analysis.


\begin{figure}[htbp]
    \centering
    \includegraphics[width=0.5\linewidth]{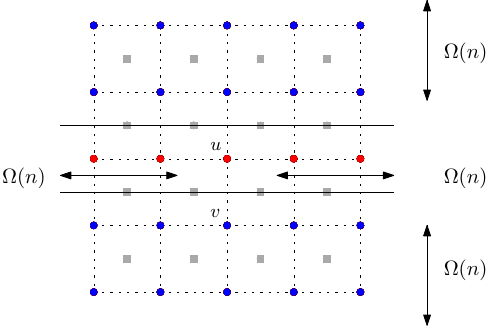}
    \caption{A close-up of a partition on the $n\times n$ grid. The arrows indicate that the red partition has a thin horizontal path of length $\Omega(n)$, that is at least $\Omega(n)$ vertically away from other red vertices. This partition cannot be modified locally near $u,v$ to a partition where $u,v$ are not separated, without increasing the imbalance.}
    \label{fig:bad-rebalance}
\end{figure}

\section*{Acknowledgements}
We thank Jamie Tucker-Foltz for helpful discussions and comments.
\newpage
\bibliography{references}
\bibliographystyle{plain}
\newpage
\appendix
\section{\lstd has Small Population Imbalance}\label{sec:imbalance}

\begin{lemma}
For large enough $m,n$ with $mn$ being even,
with probability $1-\frac{1}{mn}$, \lstd samples a 2-partition $P$ of $\gmn$ such that $\imb{P}=O(\frac{\log (mn)}{\lambda})$.
\end{lemma}
\begin{proof}

Let $\mathcal{P}_k$ be the set of 2-partitions where one side has size exactly $k$. Let $\mathcal{P}$ be the set of all 2-partitions.
Let 
\[p_k=\Pr[\spandist \text{ samples some }P\in \mathcal{P}_k] = \frac{\sum_{P\in \mathcal{P}_k}\sp{P} }{\sum_{P\in \mathcal{P}}\sp{P}}\]

From Theorem 11 in \cite{jamie_grid24}, we know 
that a uniformly sampled spanning tree has probability $\geq \frac{1}{(mn)^2}$ of having an edge whose removal splits it into a balanced 2-partition. To get a partition proportional to the spanning score, we first sample a uniformly random edge of a random spanning tree, and then remove it. Then, we reject the resulting partition, say $S,T$, with probability $\frac{1}{E(S,T)}$.
Since each tree has $mn-1$ edges, and the total number of edges in the grid is $\leq 3mn$, we get that $p_{mn/2}\geq \frac{1}{3(mn)^4}$.

Now, we sample from \lstd where we additionally have a rejection step where we accept the resulting partition $P$ with probability $e^{-\lambda \cdot \imb{P}}$. 
Then the probability that the smaller part has size at most $\delta=mn/2-\frac{6\log (mn) }{\lambda}$ is
\begin{align*}
\frac{\displaystyle\sum_{k=0}^{\delta}p_k e^{-\lambda(mn/2-k)}}{\displaystyle\sum_{k=0}^{mn/2}p_k e^{-\lambda(mn/2-k)}}.
\intertext{We want an upper bound on this quantity. We first note that the $e^{-\lambda \cdot \imb{P}}$ term for each summand in the numerator is at most $e^{-\lambda \cdot \frac{6\log (mn)}{\lambda}}$, since we are only considering partitions with imbalance more than  $\frac{6\log (mn)}{\lambda}$. Next, we can lower bound the sum in the denominator with the summand corresponding to $k=mn/2$, which is just $p_{mn/2}$. Thus, our upper bound is}
\frac{e^{-\lambda \cdot \frac{6\log (mn)}{\lambda}}\displaystyle\sum_{k=0}^{\delta}p_k }{p_{mn/2}}.
\intertext{Noting that the sum in the numerator is at most 1 since the elements $\{p_k\}$ form a distribution, and applying our earlier bound on $p_{mn/2}$, this is at most}
\frac{1}{(mn)^6 p_{mn/2}}&\leq \frac{1}{mn}.
\end{align*}
\end{proof}
\section{Missing Proofs from \Cref{sec:sep-prob}}
\label{app:missing-proofs}
\shortcycle*
\begin{proof}
    
    Intuitively, this should be $O(\abs{C}^2)$, since the maximum is obtained by enclosing a grid with the cycle. We will argue that this intuition is correct, assuming that the grid is large enough relative to $C$.

    \newcommand{\gmnh}{\widehat{\gmn}}
    Consider $\gmnd$. Instead of the outer face node, we will extend the grid one more row/column beyond the boundary to form a new graph $\gmnh$.
    We make a copy $C'$ of $C$ in $H$, where we will traverse these new edges instead of adding edges to the outer face node.
    If $C$ did not include any edges to the outer face node, we leave it unchanged in $\gmnh$. 
    If $C$ did include any such edges, we will formally describe the new cycle $C'$.
    
    Let the non-outer-face nodes of $\gmnd$ be embedded in the plane at points $(i,j)$ where $1\leq i\leq m-1$, $1\leq j\leq n-1$. Node $\outerface$ is embedded outside this region.
    We now describe the embedding of $\gmnh$. For the points common between $\gmnd$ and $\gmnh$, we embed them at the same coordinates. The vertices of the new column/row are $(i,j)$ where $i\in \{0,m\}, 0\leq j\leq n$ or $0\leq i\leq m, j\in \{0,n\}$.
    Since $C$ is a simple cycle, if it contains edges incident on $\outerface$, it contains exactly two of them.
    Let these edges be to the points that have coordinates $(i_1,j_1)$ and $(i_2,j_2)$ in the embedding of $\gmnd$. Both these vertices lie on the sides of the square formed by the lines $x=1,x=m-1,y=1$ or $y=n-1$.
    Observe that since $\abs{C}\leq \min (n,m)$, they cannot lie on opposite sides of the square. They must either lie on the same side, or on adjacent sides. Further, there is a path between them in $C$ that avoids visiting $\outerface$. This path must be at least as long as the shortest path between them in $\gmnd\setminus\{\outerface\}$ which has length exactly $\abs{i_1-i_2}+\abs{j_1-j_2}$. Thus, $\abs{C}\geq \abs{i_1-i_2}+\abs{j_1-j_2}+2$, where the $+2$ comes from including the edges to $\outerface$.
    Consider the case that the points are on the same side. Without loss of generality, let $i_1=m-1=i_2$. Then, in $C'$, in place of the edges to $\outerface$, we add the edges from $(m-1,j_1)$ to $(m,j_1)$, from $(m-1,j_2)$ to $(m,j_2)$, and the shortest path from $(m,j_1)$ to $(m,j_2)$ which has length $\abs{j_1-j_2}$. As required, $C'$ is a simple cycle.
    If the points are on adjacent sides, again without loss of generality assume that $i_1=m-1$ and $j_2=n-1$. We add the edges from
    $(m-1,j_1)$ to $(m,j_1)$, from $(i_2,n-1)$ to $(i_2,n)$, and the shortest path along the boundary from $(m,j_1)$ to $(i_2,n)$. This shortest path is the path from $(m,j_1)$ to $(m,n)$ to $(i_2,n)$
and has length $(n-j_1)+(m-i_2)=(j_2-j_1)+(i_1-i_2)$.
    In either case, we added at most $\abs{C}$ edges, since we argued that $\abs{C}\geq \abs{i_1-i_2}+\abs{j_1-j_2}+2$.
    Thus, $\abs{C'}\leq 2\abs{C}$.
    
    Now, instead of counting the primal nodes inside $C$ in $\gmnd$, we can count the dual nodes enclosed by $C'$ in $\gmnh$ instead, as a simple argument shows that this, along with the number of vertices in $C'$, form an upper bound on the number of primal nodes enclosed by $C$.
    Thus, we only need to argue that a cycle of length $k$ on a grid graph encloses $O(k^2)$ vertices.
    
Let $C$ be a cycle of length $k$ in the grid graph. Let $I$ be the set of vertices contained inside $C$.
Let $x_{min},x_{max},y_{min},y_{max}$ be the minimum and maximum coordinates of the vertices in $I$. Let $w=x_{max}-x_{min}$ and $h=y_{max}-y_{min}$.
The length of the cycle $k$ must be at least $2w+2h$, so $w+h\leq k/2$.
The interior vertices $(x,y)\in I$ must satisfy $x_{min}\leq x\leq x_{max}$ and $y_{min}\leq y \leq y_{max}$. The number of such points is $wh\leq (w+h)^2\leq k^2/4$.
Thus, the total number of vertices enclosed by $C$ is at most $\frac{(2\abs{C})^2}{4}+2\abs{C}\leq 3\abs{C}^2$.
\end{proof}

\section{Case Analysis for Cases 5-8}
\label{sec:reconnect-cases}

\newcommand{\plusblue}[1]{\textcolor{blue}{+#1}}
\newcommand{\plusred}[1]{\textcolor{red}{+#1}}

We now use the tools defined in \cref{sec:tools} to prove \cref{thm:main-reconnect}.
In several of our cases, when we apply \cref{lem:elbow} to a specific non-disposable vertex, we also implicitly use the fact that that vertex cannot be on the boundary. Whenever we change the colors of vertices, we also take care to never introduce (the possibility of) a cross-structure, although this is not explicitly pointed out.

To keep track of the imbalance, we say that a partition is \plusblue{$x$}, for $x\geq 0$ if both $u,v$ are colored blue, the number of vertices added to the blue partition is $x$, and one of $u,v$ is adjacent to a red vertex. We analogously define \plusred{$x$}.
We show that all of our partitions are \plusred{$x$} or \plusblue{$x$} for $0\leq x\leq 3$. Along with our other restrictions that we only make local changes, this is no longer possible when $u$ or $v$ are adjacent to the border (see \cref{sec:boundary-cases}).

\subsection{Cases 5 and 6}
\label{sec:case56}

We will look at these together for as much as possible due to their similarities. Combining the cases, the picture looks like the one in \cref{fig:branch1-fig1a}. 
The first step is to apply \cref{lem:elbow} to the highlighted vertex.
If the highlighted vertex in \cref{fig:branch1-fig1a} is disposable, then we can flip it to blue. In this case, $u$ becomes disposable regardless of whether we started from Case 5 or 6, and we are done by flipping $u$ to red, giving us a \plusred{0} partition.

Otherwise, after applying \cref{lem:elbow} again, we reveal more information about the neighborhood of the highlighted vertex. We again apply \cref{lem:elbow} to the highlighted vertex in \cref{fig:branch1-fig1b}. If it was disposable, then we can flip it to red. In this case, that makes the previously considered vertex (the one highlighted in \cref{fig:branch1-fig1a}) disposable, and we repeat the same argument there. This would give us a \plusred{1} partition.

If none of the cases above apply, the partition must look like the one in \cref{fig:branch1-fig1c}. Depending on the color of the highlighted vertex in \cref{fig:branch1-fig1c}, we consider two cases.

\begin{figure}[H]
     \centering
     \hfill
     \begin{subfigure}[b]{0.25\textwidth}
         \centering
           \includegraphics[width=\linewidth]{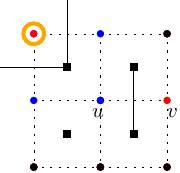}
         \caption{Step 1.}\label{fig:branch1-fig1a}
     \end{subfigure}
     \hfill
     \begin{subfigure}[b]{0.25\textwidth}
         \centering
             \includegraphics[width=\linewidth]{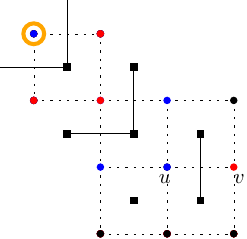}
         \caption{Step 2.}\label{fig:branch1-fig1b}
     \end{subfigure}
     \hfill
    \begin{subfigure}[b]{0.3\textwidth}
         \centering
    \includegraphics[width=\linewidth]{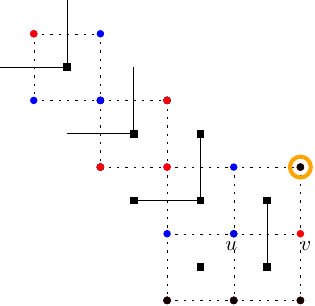}
         \caption{Step 3.}
         \label{fig:branch1-fig1c}
     \end{subfigure}
     \hfill
     \bigskip
    \caption{}
    \label{}
\end{figure}

\paragraph{Subcase 1.} The highlighted vertex in \cref{fig:branch1-fig1c} is red.
In this case, the highlighted vertex in \cref{fig:branch1-fig2a} must have a blue vertex above it. Otherwise, it is disposable (from \cref{obs:degree1-disposable}) and we can flip it to red. But after this step, $u$ becomes disposable regardless of which case we started from, and $u$ can also be flipped to red to give a \plusred{2} partition.
If this does not happen, then the picture must be as in \cref{fig:branch1-fig2b}. Now, we flip the color of $u$ and the vertex above it, giving us the coloring in \cref{fig:branch1-fig2c}. This creates two blue regions, one of which is an island, from \cref{lem:create-island}.
We also know that the blue vertices to the left of $u$ and 2 rows above $u$ are in different regions.

\begin{figure}[H]
     \centering
     \hfill
     \begin{subfigure}[b]{0.3\textwidth}
         \centering
           \includegraphics[width=\linewidth]{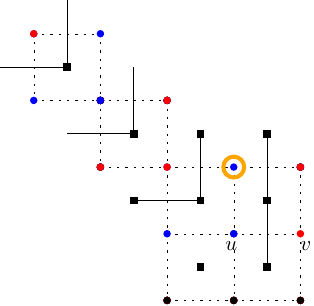}
         \caption{Step 4.}\label{fig:branch1-fig2a}
     \end{subfigure}
     \hfill
     \begin{subfigure}[b]{0.3\textwidth}
         \centering
             \includegraphics[width=\linewidth]{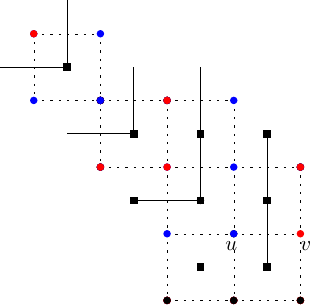}
         \caption{Step 5.}\label{fig:branch1-fig2b}
     \end{subfigure}
     \hfill
    \begin{subfigure}[b]{0.3\textwidth}
         \centering
    \includegraphics[width=\linewidth]{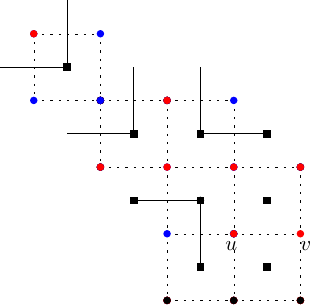}
         \caption{Step 6.}
         \label{fig:branch1-fig2c}
     \end{subfigure}
     \hfill
     \bigskip
    \caption{}
    \label{}
\end{figure}

Consider the highlighted blue vertices in \cref{fig:branch1-fig3a}.
They must be in the same region, otherwise we are done since we have a 1-thin structure. We can resolve this thin structure to get a \plusred{1} partition. 
We represent the path between the blue vertices with a dotted line between them.
Now, the highlighted vertices in \cref{fig:branch1-fig3b} must be red.
If either of them were blue (from either region), they would form 1-thin or 2-thin structures between the blue regions, or connect the two blue regions which would contradict that they were distinct blue regions. These can be resolved to get a \plusred{1} coloring.
If the highlighted vertex in \cref{fig:branch1-fig3c} has a red vertex below it, then we know that one of the blue regions is a singleton vertex. We can just flip this vertex to red to get a feasible $\plusred{3}$ partition.
Thus, we consider the case where the vertex below is blue.

\begin{figure}[H]
     \centering
     \hfill
     \begin{subfigure}[b]{0.3\textwidth}
         \centering
           \includegraphics[width=\linewidth]{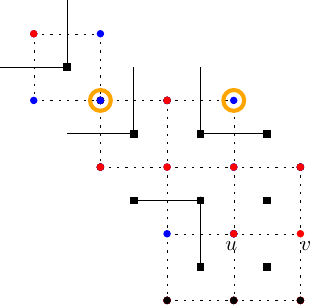}
         \caption{Step 7.}
         \label{fig:branch1-fig3a}
     \end{subfigure}
     \hfill
     \begin{subfigure}[b]{0.3\textwidth}
         \centering
             \includegraphics[width=\linewidth]{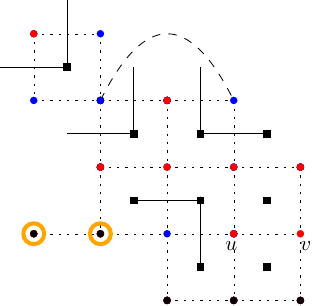}
         \caption{Step 8.}
         \label{fig:branch1-fig3b}
     \end{subfigure}
     \hfill
    \begin{subfigure}[b]{0.3\textwidth}
         \centering
    \includegraphics[width=\linewidth]{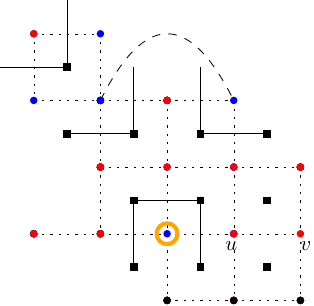}
         \caption{Step 9}
         \label{fig:branch1-fig3c}
     \end{subfigure}
     \hfill
     \bigskip
    \caption{}
    \label{}
\end{figure}

In this case, the highlighted vertices in \cref{fig:branch1-fig4a} must be red.
If they were blue (from either region), we would again have thin structures between the blue regions or the regions would not be distinct. These would resolve to a \plusred{1} or a \plusred{2} partition depending on which region the blue vertices belonged to. 
Thus, we get the coloring in \cref{fig:branch1-fig4b}, and flipping some vertices gives us the final coloring in \cref{fig:branch1-fig4c}.
Regardless of the color of the remaining black vertices, we can show that the final \plusred{0} partition is feasible.
The only issue that may arise from this change is that the red vertices in the bottom left may become disconnected from the red region that $u,v$ belong to.
However, this is not possible.
The key idea is in noting that one of the two blue regions in \cref{fig:branch1-fig4b} is an island. If the bottom blue region is an island, then the island walk (see \cref{lem:island-walk}) of this island ensures that there is a path between the two sets of red vertices that is not affected by the flips in the last step.
Similarly if the top blue region is an island. See \cref{fig:island-walk-connection-1}.


\begin{figure}[H]
     \centering
     \hfill
     \begin{subfigure}[b]{0.3\textwidth}
         \centering
           \includegraphics[width=\linewidth]{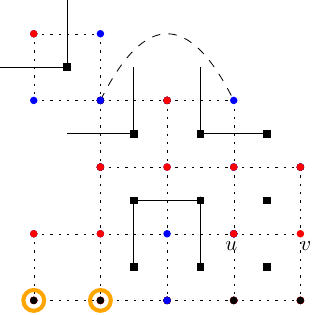}
         \caption{Step 10.}
         \label{fig:branch1-fig4a}
     \end{subfigure}
     \hfill
     \begin{subfigure}[b]{0.3\textwidth}
         \centering
             \includegraphics[width=\linewidth]{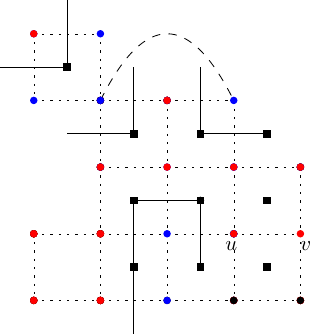}
         \caption{Step 11.}
         \label{fig:branch1-fig4b}
     \end{subfigure}
     \hfill
    \begin{subfigure}[b]{0.3\textwidth}
         \centering
    \includegraphics[width=\linewidth]{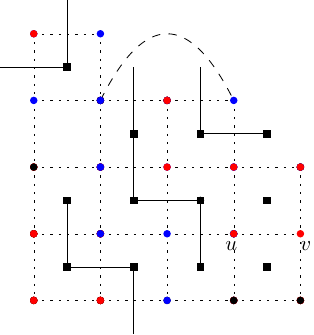}
         \caption{Final coloring.}
         \label{fig:branch1-fig4c}
     \end{subfigure}
     \hfill
     \bigskip
    \caption{}
    \label{}
\end{figure}

\begin{figure}[H]
     \centering
     \hfill
     \begin{subfigure}[b]{0.4\textwidth}
         \centering
           \includegraphics[width=\linewidth]{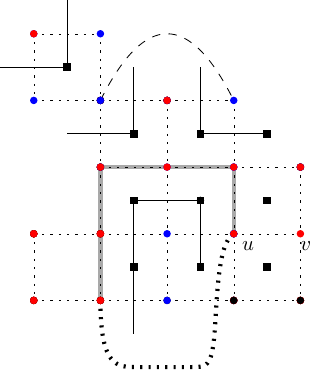}
         \caption{If the bottom blue region is an island.}
     \end{subfigure}
     \hfill
     \begin{subfigure}[b]{0.4\textwidth}
         \centering
             \includegraphics[width=\linewidth]{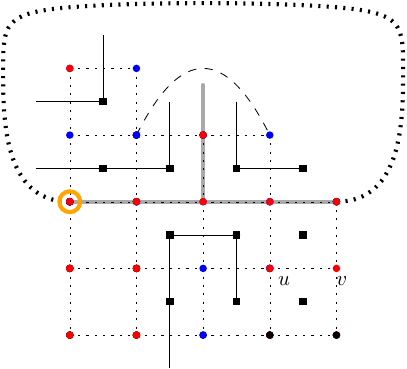}
         \caption{If the top blue region is an island. We consider the case where the highlighted vertex is red. If it is blue, then the island walk goes through the vertex below it.}
     \end{subfigure}
     \hfill
     \bigskip
    \caption{The thick gray path represents the island walk. In either case, the red vertices continue to be connected in \cref{fig:branch1-fig4c} (and hence there is only one red region) via the thick dotted curve, which is part of the island walk.}
    \label{fig:island-walk-connection-1}
\end{figure}

\paragraph{Subcase 2.} If the highlighted vertex in \cref{fig:branch1-fig1c} is blue, we have the coloring in \cref{fig:branch1-fig5a}.
If the vertex below $v$ was blue, then $v$ would be disposable from \Cref{obs:degree1-disposable}, and we can flip it to blue. In this case, we get a \plusblue{1} partition.
Thus, the vertex below $v$ must be red. A similar argument can be made if the vertex to the right of $v$ is blue, or if the vertex to the right of $v$ is blue, but the highlighted vertex in \cref{fig:branch1-fig5b} was red. In short, $v$ must be in \subref{case:6}.
We now apply the elbow lemma (\cref{lem:elbow}) as in Steps 2 and 3 (\cref{fig:branch1-fig1b,fig:branch1-fig1c}) starting from the highlighted vertex in \cref{fig:branch1-fig5b} since it is symmetric.
If the neighborhood of $u$ was in \subref{case:5}, then we have the coloring in \cref{fig:branch1-fig5c}. We note that this case is precisely the symmetric version of the previous subcase, and we can repeat the steps starting from \cref{fig:branch1-fig2a}.
The only remaining case is when the neighborhood of $u$ is in \subref{case:6}, just like the neighborhood of $v$. We now have the coloring in \cref{fig:branch2-fig6a}. 
We have two cases based on the color of the highlighted vertex in \cref{fig:branch2-fig6a}.

\begin{figure}[H]
     \centering
     \hfill
     \begin{subfigure}[b]{0.3\textwidth}
         \centering
           \includegraphics[width=\linewidth]{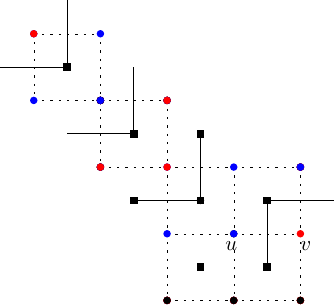}
         \caption{Step 4.}
         \label{fig:branch1-fig5a}
     \end{subfigure}
     \hfill
     \begin{subfigure}[b]{0.3\textwidth}
         \centering
             \includegraphics[width=\linewidth]{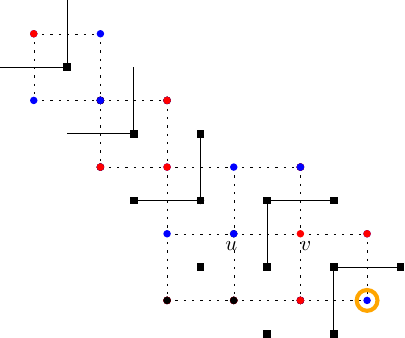}
         \caption{Step 5.}
         \label{fig:branch1-fig5b}
     \end{subfigure}
     \hfill
    \begin{subfigure}[b]{0.3\textwidth}
         \centering
    \includegraphics[width=\linewidth]{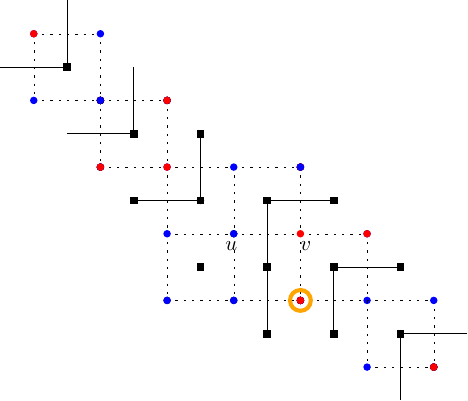}
         \caption{Step 6.}
         \label{fig:branch1-fig5c}
     \end{subfigure}
     \hfill
     \bigskip
    \caption{}
    \label{}
\end{figure}

\paragraph{Subcase 2a.} The highlighted vertex in \cref{fig:branch2-fig6a} is blue.
In this case, the highlighted vertex in
\cref{fig:branch2-fig6b} must have a red vertex to its right. Otherwise, we can flip the color of it and $v$ to blue to be done to get a \plusblue{2} partition. Thus, we get the coloring in \cref{fig:branch2-fig6c}. 
We now flip the color of $v$ and the vertex to its right. This creates two red regions, one of which must be an island, from \cref{lem:create-island}.

\begin{figure}[H]
     \centering
     \hfill
     \begin{subfigure}[b]{0.3\textwidth}
         \centering
           \includegraphics[width=\linewidth]{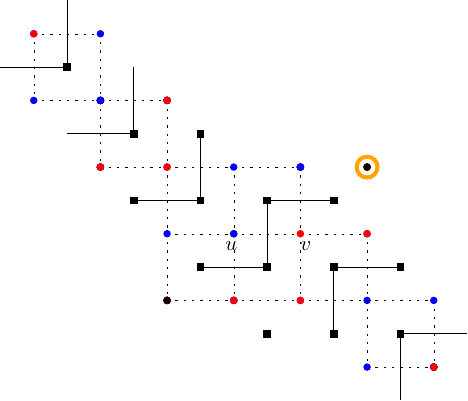}
         \caption{Step 7}
         \label{fig:branch2-fig6a}
     \end{subfigure}
     \hfill
     \begin{subfigure}[b]{0.3\textwidth}
         \centering
             \includegraphics[width=\linewidth]{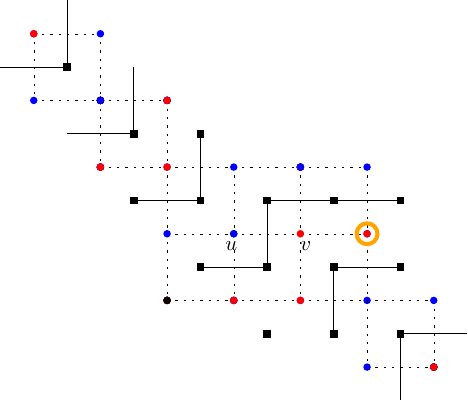}
         \caption{Step 8.}
         \label{fig:branch2-fig6b}
     \end{subfigure}
     \hfill
    \begin{subfigure}[b]{0.3\textwidth}
         \centering
    \includegraphics[width=\linewidth]{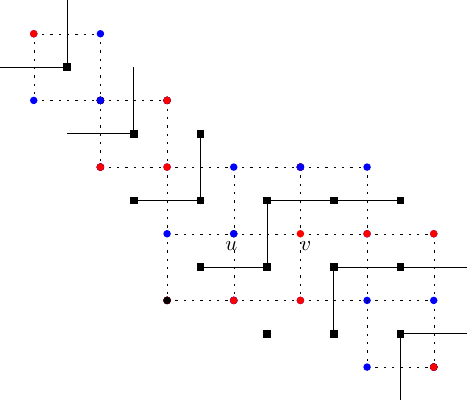}
         \caption{Step 9.}
         \label{fig:branch2-fig6c}
     \end{subfigure}
     \hfill
     \bigskip
    \caption{}
    \label{}
\end{figure}

Consider the highlighted vertices in \cref{fig:branch2-fig7a}.
If they are not in the same red region, we have a 1-thin structure, and we can resolve them to get a \plusblue{1} partition. Thus, assume that they are in the same region. We denote this with a dotted line between the two.
Next, the highlighted vertices in \cref{fig:branch2-fig7b} must be blue. Otherwise, we have a thin structure between the two red regions which can be resolved to get a \plusblue{0} or a \plusblue{1} partition.
Next, the highlighted vertex in \cref{fig:branch2-fig7c} must be red. If it was blue, the two red vertices below $u,v$ can be flipped to blue, and the vertex to the right of $v$ can be flipped back to red, to get a \plusblue{3} coloring.

\begin{figure}[H]
     \centering
     \hfill
     \begin{subfigure}[b]{0.3\textwidth}
         \centering
           \includegraphics[width=\linewidth]{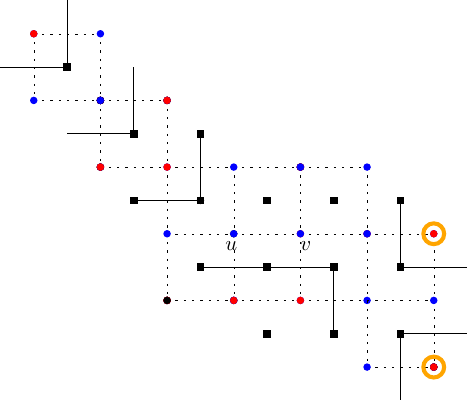}
         \caption{Step 10.}
         \label{fig:branch2-fig7a}
     \end{subfigure}
     \hfill
     \begin{subfigure}[b]{0.3\textwidth}
         \centering
             \includegraphics[width=\linewidth]{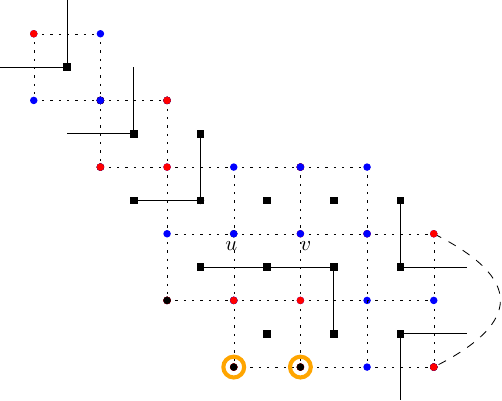}
         \caption{Step 11.}
         \label{fig:branch2-fig7b}
     \end{subfigure}
     \hfill
    \begin{subfigure}[b]{0.3\textwidth}
         \centering
    \includegraphics[width=\linewidth]{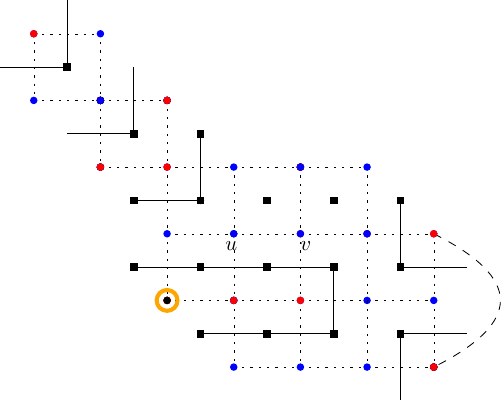}
         \caption{Step 12.}
         \label{fig:branch2-fig7c}
     \end{subfigure}
     \hfill
     \bigskip
    \caption{}
    \label{}
\end{figure}

Next, the highlighted vertex in \cref{fig:branch2-fig7d} must be blue. If it was red, $u$ and the vertex to its left can be flipped to red after flipping $v$ and the vertex to its right back to red. This is a \plusred{2} coloring.
The highlighted red vertices in \cref{fig:branch2-fig8a} must be in the same region, otherwise we have a 1-thin structure that would resolve to a \plusblue{1} partition. We will represent these two vertices with a pink highlight because of space constraints.
Now, we undo an earlier step and flip $v$ and the vertex to its right back to red to get the coloring in \cref{fig:branch2-fig8c}.
Now, we flip $u$ and the vertex to its left to red. 
From \cref{lem:create-island}, we again create two blue regions, one of which is an island.

\begin{figure}[H]
     \centering
     \hfill
     \begin{subfigure}[b]{0.3\textwidth}
         \centering
    \includegraphics[width=\linewidth]{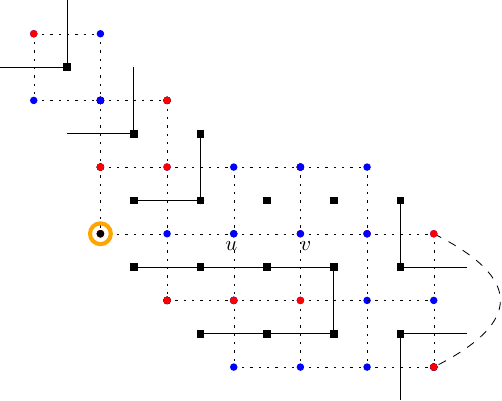}
         \caption{Step 13.}
         \label{fig:branch2-fig7d}
     \end{subfigure}
     \begin{subfigure}[b]{0.3\textwidth}
         \centering
           \includegraphics[width=\linewidth]{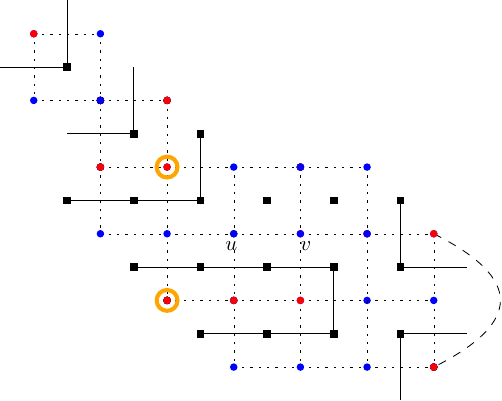}
         \caption{Step 14.}
         \label{fig:branch2-fig8a}
     \end{subfigure}
     \hfill
    \begin{subfigure}[b]{0.3\textwidth}
         \centering
    \includegraphics[width=\linewidth]{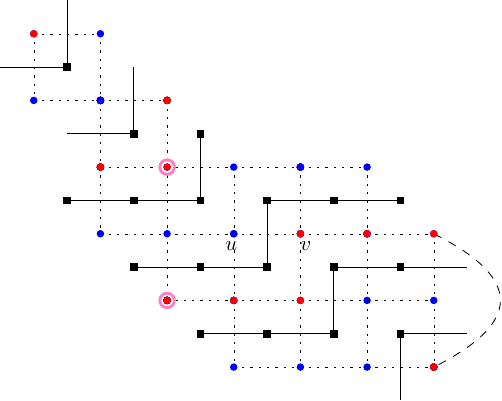}
         \caption{Step 15.}
         \label{fig:branch2-fig8c}
     \end{subfigure}
     \hfill
     \bigskip
    \caption{}
    \label{}
\end{figure}

The highlighted vertices in \cref{fig:branch2-fig9a} must be in the same region, otherwise we have a 1-thin structure that can be resolved to get a \plusred{1} partition. We represent the path between them with a dotted line.
The highlighted vertices in \cref{fig:branch2-fig9b} must also be in the same region, otherwise we again have a 1-thin structure whose resolution gives us a \plusred{1} partition. We represent this with a cyan highlight.
The highlighted vertices in \cref{fig:branch2-fig9c} must all be red. Otherwise we have some thin structure between the blue regions. These give us a \plusred{0} or a \plusred{1} partition depending on which ones are blue and which region they belong to.

\begin{figure}[H]
     \centering
     \hfill
     \begin{subfigure}[b]{0.3\textwidth}
         \centering
           \includegraphics[width=\linewidth]{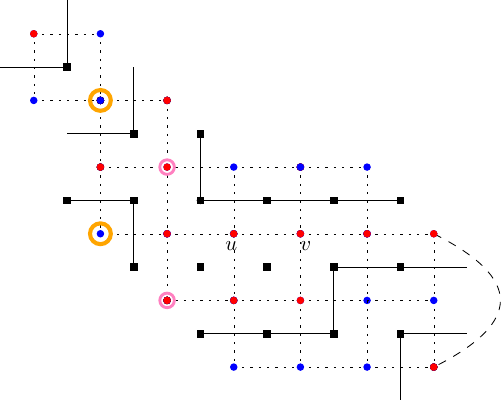}
         \caption{Step 16.}
         \label{fig:branch2-fig9a}
     \end{subfigure}
     \hfill
     \begin{subfigure}[b]{0.3\textwidth}
         \centering
             \includegraphics[width=\linewidth]{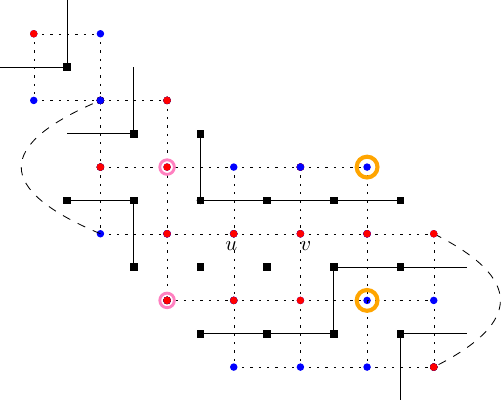}
         \caption{Step 17.}
         \label{fig:branch2-fig9b}
     \end{subfigure}
     \hfill
    \begin{subfigure}[b]{0.3\textwidth}
         \centering
    \includegraphics[width=\linewidth]{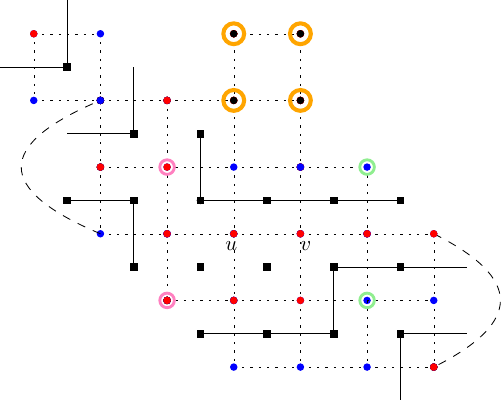}
         \caption{Step 18.}
         \label{fig:branch2-fig9c}
     \end{subfigure}
     \hfill
     \bigskip
    \caption{}
    \label{}
\end{figure}

Thus, the coloring looks like the one in \cref{fig:branch2-fig10a}. We flip some red vertices to connect the two blue regions, to get the coloring in \cref{fig:branch2-fig10b}.
Regardless of the color of the black vertex, we can show that the the resulting \plusred{0} partition is feasible.
This follows a similar argument as we made for the coloring in \cref{fig:branch1-fig4c}.
Regardless of whether the left blue region or the right blue region is the island, its island walk ensures that this last flip does not separate the red region into two disconnected regions.

\begin{figure}[H]
     \centering
     \hfill
     \begin{subfigure}[b]{0.4\textwidth}
         \centering
           \includegraphics[width=\linewidth]{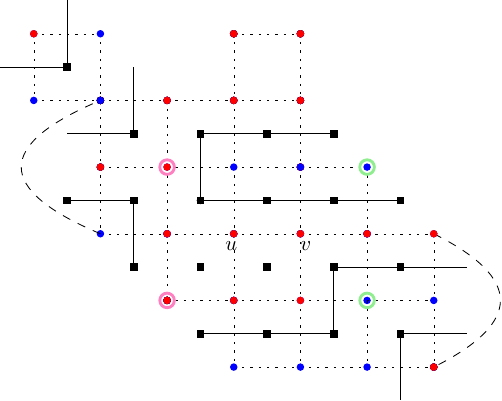}
         \caption{Step 19.}
         \label{fig:branch2-fig10a}
     \end{subfigure}
     \hfill
     \begin{subfigure}[b]{0.4\textwidth}
         \centering
             \includegraphics[width=\linewidth]{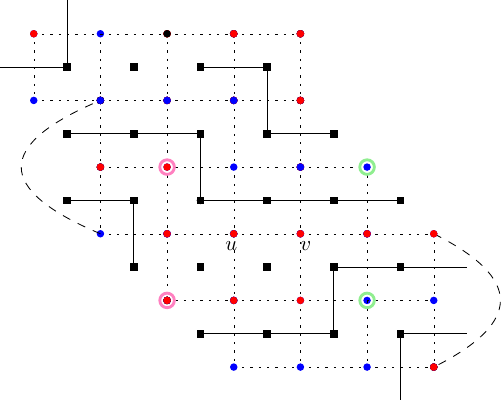}
         \caption{Final state.}
         \label{fig:branch2-fig10b}
     \end{subfigure}
     \hfill
     \bigskip
    \caption{}
    \label{}
\end{figure}

\paragraph{Subcase 2b.} Suppose the highlighted vertex in \cref{fig:branch2-fig6a} was red. We get the coloring in \cref{fig:branch2-fig11a}.
If both the highlighted vertices in \cref{fig:branch2-fig11a} were red, then we can flip $u$ and the vertices above $u$ and $v$ to red. This gives us a \plusred{3} partition. 

Thus, at least one of the highlighted vertices in \cref{fig:branch2-fig11a}  must be blue. 
We consider the case where the right one is blue (the left one may or may not be blue). The other case is similar.
We flip the color of $u$ and the two highlighted vertices in \cref{fig:branch2-fig11b} to red. From \cref{lem:create-island} creates two blue regions, one of which is an island.
The highlighted blue vertices in \cref{fig:branch2-fig11c} must be in the same region to avoid a thin structure that resolves to \plusred{1} or a \plusred{2} partition (depending on the color of the black vertex).
We represent the path between them with a dotted line.

\begin{figure}[H]
     \centering
     \hfill
     \begin{subfigure}[b]{0.3\textwidth}
         \centering
           \includegraphics[width=\linewidth]{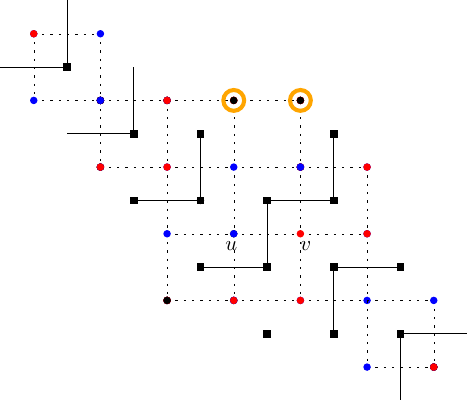}
         \caption{Step 8.}
         \label{fig:branch2-fig11a}
     \end{subfigure}
     \hfill
     \begin{subfigure}[b]{0.3\textwidth}
         \centering
             \includegraphics[width=\linewidth]{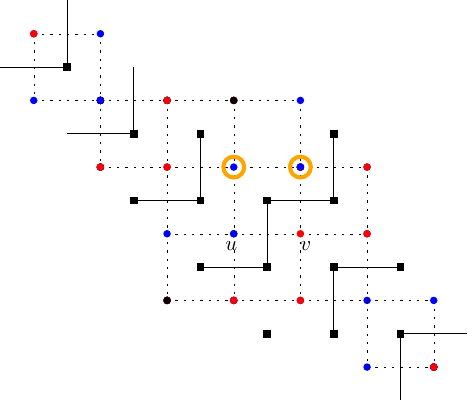}
         \caption{Step 9.}
         \label{fig:branch2-fig11b}
     \end{subfigure}
     \hfill
    \begin{subfigure}[b]{0.3\textwidth}
         \centering
    \includegraphics[width=\linewidth]{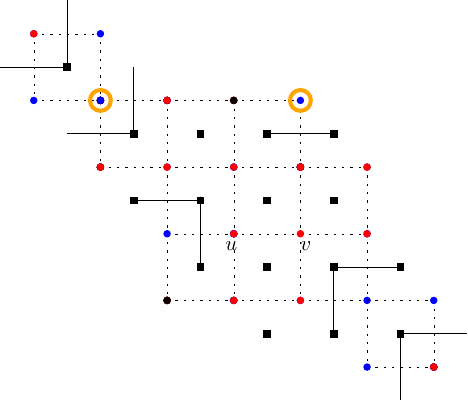}
         \caption{Step 10.}
         \label{fig:branch2-fig11c}
     \end{subfigure}
     \hfill
     \bigskip
    \caption{}
    \label{}
\end{figure}

If the highlighted vertex in \cref{fig:branch2-fig12a} was blue, we would get a thin-structure that can be resolved to get a \plusred{2} partition.
Thus, we assume it is red. This causes the highlighted vertex in \cref{fig:branch2-fig12b} to be blue. Otherwise, the vertex to the left of $u$ is the only vertex in its region, and we can flip it. The vertices above $u$ and $v$ can then be flipped back to blue to give us a \plusred{2} partition.
Next, the highlighted vertex in \cref{fig:branch2-fig12c} must be red to avoid a 2-thin structure whose resolution leads to a \plusred{1} partition.
\begin{figure}[H]
     \centering
     \hfill
     \begin{subfigure}[b]{0.3\textwidth}
         \centering
           \includegraphics[width=\linewidth]{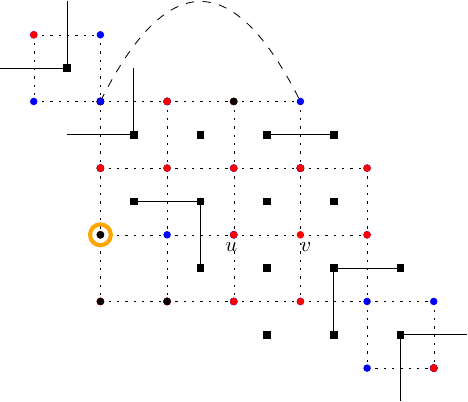}
         \caption{Step 11.}
         \label{fig:branch2-fig12a}
     \end{subfigure}
     \hfill
     \begin{subfigure}[b]{0.3\textwidth}
         \centering
             \includegraphics[width=\linewidth]{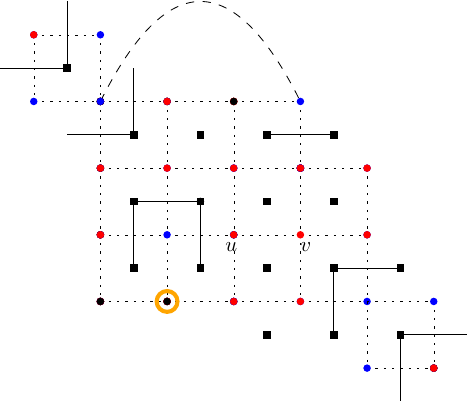}
         \caption{Step 12.}
         \label{fig:branch2-fig12b}
     \end{subfigure}
     \hfill
    \begin{subfigure}[b]{0.3\textwidth}
         \centering
    \includegraphics[width=\linewidth]{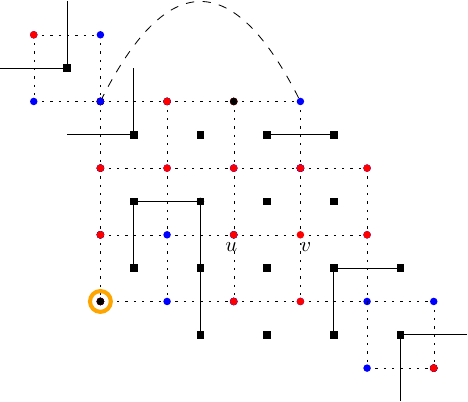}
         \caption{Step 13.}
         \label{fig:branch2-fig12c}
     \end{subfigure}
     \hfill
     \bigskip
    \caption{}
    \label{}
\end{figure}

The highlighted vertices in \cref{fig:branch2-fig13a} must be red to avoid a thin-structure between the blue regions. Again, their resolution would lead to a \plusred{1} or a \plusred{2} partition.
We flip the highlighted vertices in \cref{fig:branch2-fig13b} to blue to connect the blue regions.
This leads to the coloring in \cref{fig:branch2-fig13c}, which is \plusred{1} partition and we can argue is feasible using an argument similar to the colorings in \cref{fig:branch1-fig4c} and \cref{fig:branch2-fig10b}.

\begin{figure}[H]
     \centering
     \hfill
     \begin{subfigure}[b]{0.3\textwidth}
         \centering
           \includegraphics[width=\linewidth]{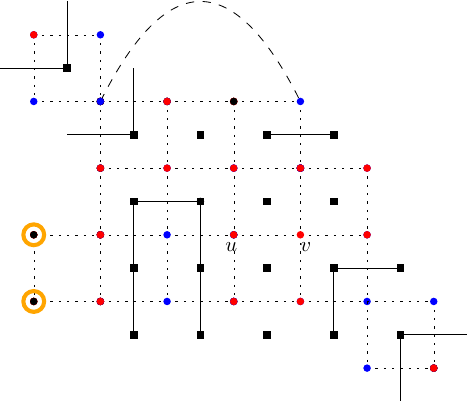}
         \caption{Step 14.}
         \label{fig:branch2-fig13a}
     \end{subfigure}
     \hfill
     \begin{subfigure}[b]{0.3\textwidth}
         \centering
           \includegraphics[width=\linewidth]{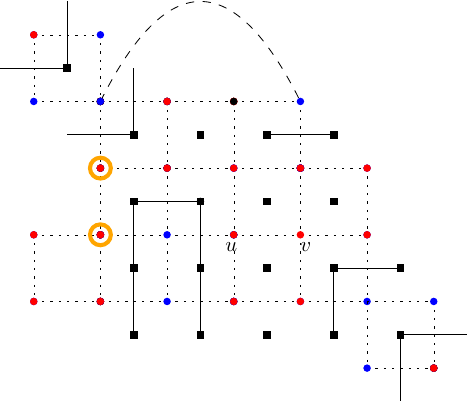}
         \caption{Step 15.}
         \label{fig:branch2-fig13b}
     \end{subfigure}
     \hfill
     \begin{subfigure}[b]{0.3\textwidth}
         \centering
             \includegraphics[width=\linewidth]{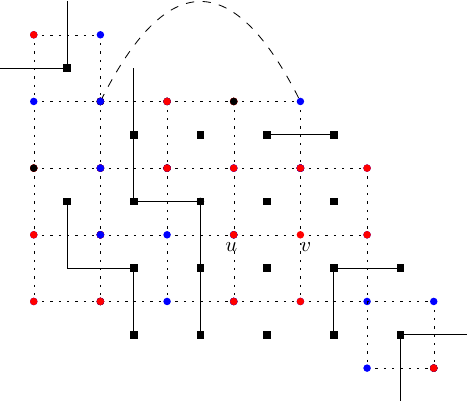}
         \caption{Final state.}
         \label{fig:branch2-fig13c}
     \end{subfigure}
     \hfill
     \bigskip
    \caption{}
    \label{}
\end{figure}

\subsection{Case 7}
\label{sec:case7}

We deal with Cases (7,7) and (7,8) separately.
We look at Case (7,7), since we dealt with Case (7,8) in \cref{sec:case-78-mainbody}.
Recall that this means that both the neighborhoods of $u$ and $v$ are in \subref{case:7}.

\paragraph{Case (7,7).}

We first flip $u$ to red. From \cref{lem:create-island}, this creates two blue regions, at least one of which is an island. Since there are 2 regions, we assume without loss of generality that the highlighted vertices in \cref{fig:branch3-fig1b} in the same region.
We represent the path between them with a dotted line.
Next, consider the highlighted vertex in \cref{fig:branch3-fig1c}. If it was blue, then it would lead to a 1-thin structure whose resolution would lead to \plusred{0} partition. Thus, it must be red.

\begin{figure}[H]
     \centering
     \hfill
     \begin{subfigure}[b]{0.3\textwidth}
         \centering
\includegraphics[width=\linewidth]{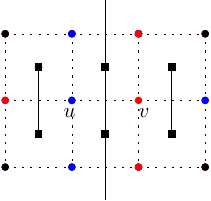}
         \caption{Step 1.}
         \label{fig:branch3-fig1a}
     \end{subfigure}
     \hfill
     \begin{subfigure}[b]{0.3\textwidth}
         \centering
             \includegraphics[width=\linewidth]{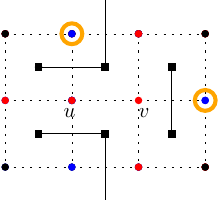}
         \caption{Step 2.}
         \label{fig:branch3-fig1b}
     \end{subfigure}
     \hfill
    \begin{subfigure}[b]{0.3\textwidth}
         \centering
    \includegraphics[width=\linewidth]{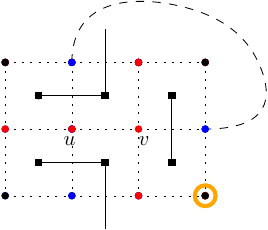}
         \caption{Step 3.}
         \label{fig:branch3-fig1c}
     \end{subfigure}
     \hfill
     \bigskip
    \caption{}
    \label{}
\end{figure}

We flip $u$ back to blue and flip $v$ to blue as well. Again, this creates two red regions, one of which is an island.
Now, consider the highlighted vertices in \cref{fig:branch3-fig2a}. If they were in the same region, it can be argued (in a similar way as in \cref{lem:xoox}) that it violates planarity.
The main observation is that since they were in the same region even after $v$ was flipped to blue, they have a path between them that does not involve $v$.
Thus, the coloring must be as in \cref{fig:branch3-fig2b}, with the leftmost red vertex being in the same region as the bottom red vertices.
Now, consider the highlighted vertex in \cref{fig:branch3-fig2b}. It must be blue. If it were red, it would make a 1-thin structure that we can resolve to get a \plusblue{0} partition.
Similarly, the highlighted vertex in \cref{fig:branch3-fig2c} must be blue.

\begin{figure}[H]
     \centering
     \hfill
     \begin{subfigure}[b]{0.3\textwidth}
         \centering
    \includegraphics[width=\linewidth]{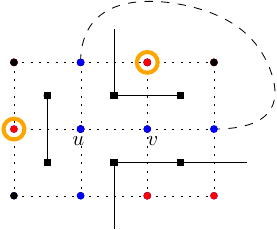}
         \caption{Step 1.}
         \label{fig:branch3-fig2a}
     \end{subfigure}
     \hfill
     \begin{subfigure}[b]{0.3\textwidth}
         \centering
             \includegraphics[width=\linewidth]{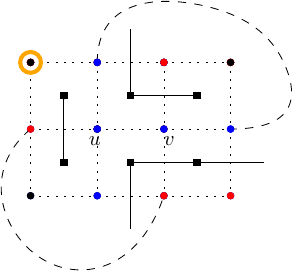}
         \caption{Step 2.}
         \label{fig:branch3-fig2b}
     \end{subfigure}
     \hfill
    \begin{subfigure}[b]{0.3\textwidth}
         \centering
    \includegraphics[width=\linewidth]{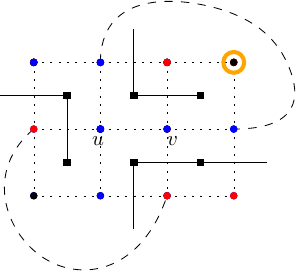}
         \caption{Step 3.}
         \label{fig:branch3-fig2c}
     \end{subfigure}
     \hfill
     \bigskip
    \caption{}
    \label{}
\end{figure}

Consider the vertex highlighted in \cref{fig:branch3-fig3a}. If it is adjacent to the boundary or if the vertex above it is blue, then the highlighted vertex would be the only red vertex in that region. It can be flipped to blue to get a \plusblue{2} partition. Thus, we assume that the vertex above exists and is red.
Now, we flip $u,v$ to red to get the coloring in \cref{fig:branch3-fig3c}.
Again, from \cref{lem:create-island}, we know that there are two blue regions, one of which is an island.
The highlighted vertex in \cref{fig:branch3-fig3c} must be red. If it were blue, it would create a 1-thin structure that we can resolve to get a \plusred{0} partition.

\begin{figure}[H]
     \centering
     \hfill
     \begin{subfigure}[b]{0.3\textwidth}
         \centering
           \includegraphics[width=\linewidth]{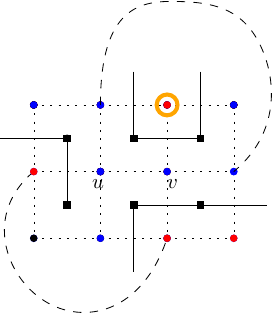}
         \caption{Step 4.}
         \label{fig:branch3-fig3a}
     \end{subfigure}
     \hfill
     \begin{subfigure}[b]{0.3\textwidth}
         \centering
             \includegraphics[width=\linewidth]{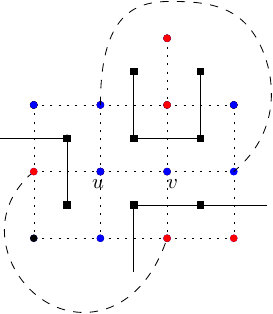}
         \caption{Step 5.}
         \label{fig:branch3-fig3b}
     \end{subfigure}
     \hfill
    \begin{subfigure}[b]{0.3\textwidth}
         \centering
    \includegraphics[width=\linewidth]{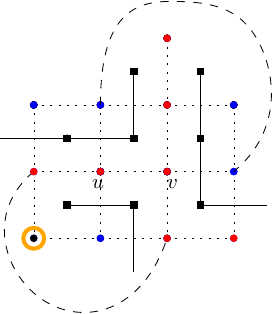}
         \caption{Step 6.}
         \label{fig:branch3-fig3c}
     \end{subfigure}
     \hfill
     \bigskip
    \caption{}
    \label{}
\end{figure}

If the highlighted vertex in \cref{fig:branch3-fig4a} was on the border or if the vertex below was red, then the highlighted vertex would be a singleton region that we can flip to red to get a \plusred{2} partition.
Thus, the vertex below exists and is blue. Consider the highlighted vertex in \cref{fig:branch3-fig4b}. If it were blue (from either region), it would create a 1-thin structure that can be resolved to get a \plusred{0} partition. Thus, we assume it is red.
We now have two cases depending on the color of the highlighted vertex in \cref{fig:branch3-fig4c}.

\begin{figure}[H]
     \centering
     \hfill
     \begin{subfigure}[b]{0.3\textwidth}
         \centering
           \includegraphics[width=\linewidth]{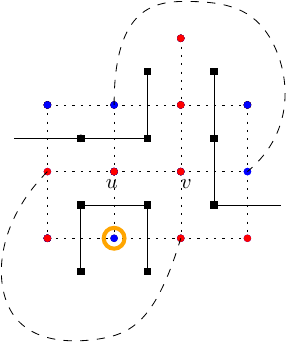}
         \caption{Step 7.}
         \label{fig:branch3-fig4a}
     \end{subfigure}
     \hfill
     \begin{subfigure}[b]{0.3\textwidth}
         \centering
             \includegraphics[width=\linewidth]{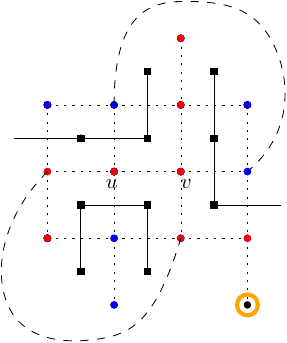}
         \caption{Step 8.}
         \label{fig:branch3-fig4b}
     \end{subfigure}
     \hfill
    \begin{subfigure}[b]{0.3\textwidth}
         \centering
    \includegraphics[width=\linewidth]{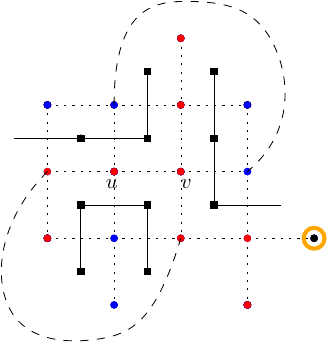}
         \caption{Step 9.}
         \label{fig:branch3-fig4c}
     \end{subfigure}
     \hfill
     \bigskip
    \caption{}
    \label{}
\end{figure}

\paragraph{Case (7,7) : Subcase 1.} The highlighted vertex in \cref{fig:branch3-fig4c} is blue. Though we have a 2-thin structure, we cannot resolve it since it would lead to us flipping more red vertices to blue than blue vertices to red.
We need to do some more work in this case.
Consider the highlighted vertex in \cref{fig:branch3-fig5a}. It must be blue to avoid a cross-structure in the original coloring.
If the highlighted vertex in \cref{fig:branch3-fig5b} was blue, we have a 2-thin structure. We can resolve this, and additionally flip the vertex below $u$ to red to get a \plusred{0} partition. Thus, we assume it is red.
We now flip $u,v$ to blue to get the coloring in \cref{fig:branch3-fig6a}.

\begin{figure}[H]
     \centering
     \hfill
     \begin{subfigure}[b]{0.3\textwidth}
         \centering
           \includegraphics[width=\linewidth]{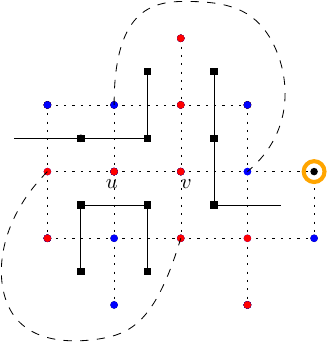}
         \caption{Step 10.}
         \label{fig:branch3-fig5a}
     \end{subfigure}
     \hfill
     \begin{subfigure}[b]{0.3\textwidth}
         \centering
             \includegraphics[width=\linewidth]{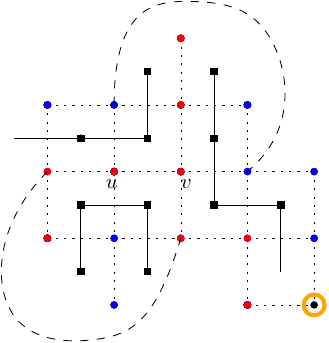}
        \caption{Step 11.}
         \label{fig:branch3-fig5b}
     \end{subfigure}
     \hfill
    \begin{subfigure}[b]{0.3\textwidth}
         \centering
    \includegraphics[width=\linewidth]{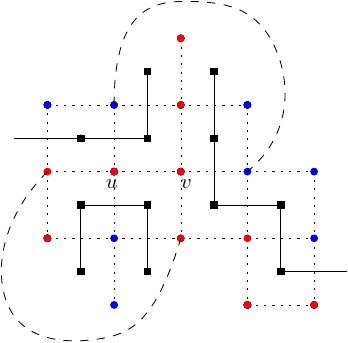}
         \caption{Step 12.}
         \label{fig:branch3-fig5c}
     \end{subfigure}
     \hfill
     \bigskip
    \caption{}
    \label{}
\end{figure}

If the highlighted vertex in \cref{fig:branch3-fig6a} was a red vertex from the bottom red region, then we have a 1-thin structure that we can resolve to get a \plusblue{0} partition. If it was a red vertex from the top red region, then we can resolve the resulting 2-thin structure, as well as flip the vertex above $v$ to blue to get a \plusblue{0} partition.
Thus, we assume that it is blue.
This makes the vertex to the right of $v$ disposable, so we flip that vertex, along with $u,v$ to red. We can now resolve the 2-thin structure between the vertices in \cref{fig:branch3-fig6b} to get the \plusred{0} partition in \cref{fig:branch3-fig6c}.

\begin{figure}[H]
     \centering
     \hfill
     \begin{subfigure}[b]{0.3\textwidth}
         \centering
           \includegraphics[width=\linewidth]{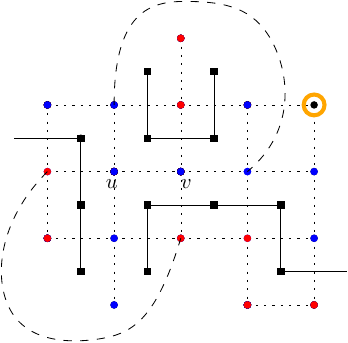}
         \caption{Step 13.}
         \label{fig:branch3-fig6a}
     \end{subfigure}
     \hfill
     \begin{subfigure}[b]{0.3\textwidth}
         \centering
             \includegraphics[width=\linewidth]{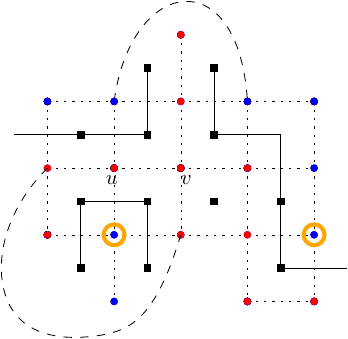}
         \caption{Step 14.}
         \label{fig:branch3-fig6b}
     \end{subfigure}
     \hfill
    \begin{subfigure}[b]{0.3\textwidth}
         \centering
    \includegraphics[width=\linewidth]{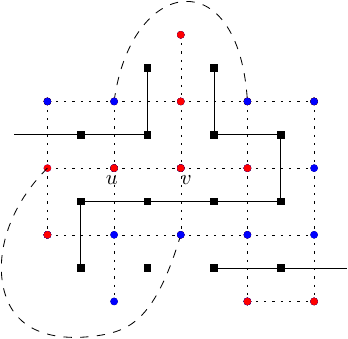}
         \caption{Step 15.}
         \label{fig:branch3-fig6c}
     \end{subfigure}
     \hfill
     \bigskip
    \caption{}
    \label{}
\end{figure}

\paragraph{Case (7,7) : Subcase 2.} The highlighted vertex in \cref{fig:branch3-fig4c} is red. Consider the highlighted vertex in \cref{fig:branch3-fig7a}. If it was red, we have a 2-thin structure that we can resolve. After flipping the  vertex above $v$ to blue, we have a \plusblue{0} partition. Thus, we assume it is blue. 
Both the highlighted vertices from \cref{fig:branch3-fig7b} must be blue for similar reasons. If either of them were a red vertex from the bottom red region, we directly have a 1-thin structure whose resolution gives us a \plusblue{0} partition. If the lower highlighted vertex was red from the top red region, we again have a 1-thin structure that resolves to a \plusblue{0} partition. If the higher highlighted vertex was red from the top red region, we have a 2-thin structure that we can resolve. After flipping the vertex above $v$ to blue, we have a \plusblue{0} partition.
Noting that all of the highlighted vertices in \cref{fig:branch3-fig7c} are disposable, we flip them to get the coloring in \cref{fig:branch3-fig8a}. Regardless of the colors of the black vertices, it is a feasible \plusblue{0} partition.

\begin{figure}[H]
     \centering
     \hfill
     \begin{subfigure}[b]{0.3\textwidth}
         \centering
           \includegraphics[width=\linewidth]{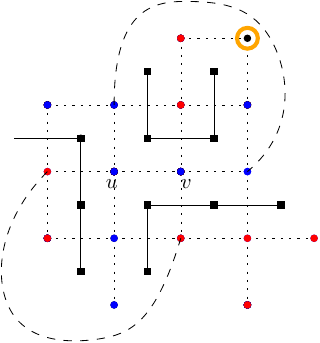}
         \caption{Step 10.}
         \label{fig:branch3-fig7a}
     \end{subfigure}
     \hfill
     \begin{subfigure}[b]{0.3\textwidth}
         \centering
             \includegraphics[width=\linewidth]{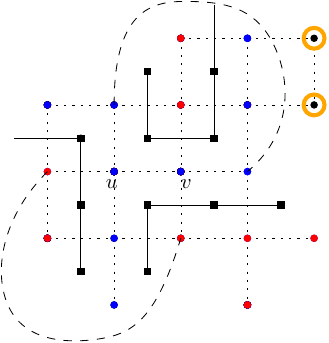}
         \caption{Step 11.}
         \label{fig:branch3-fig7b}
     \end{subfigure}
     \hfill
    \begin{subfigure}[b]{0.3\textwidth}
         \centering
    \includegraphics[width=\linewidth]{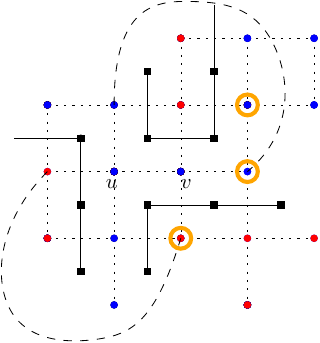}
         \caption{Step 12.}
         \label{fig:branch3-fig7c}
     \end{subfigure}
     \hfill
     \bigskip
    \caption{}
    \label{}
\end{figure}

\begin{figure}[H]
     \centering
    \includegraphics[width=0.3\linewidth]{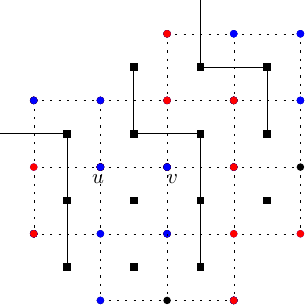}
         \caption{Final state.}
         \label{fig:branch3-fig8a}
\end{figure}

\subsection{Case 8}

Note that since we covered every other case, we only need to take care of Case (8,8).
We first consider the highlighted vertex in \cref{fig:branch4-fig1a}. If it was adjacent to a border or if the vertex to the left was red, then we can flip the highlighted vertex to red to return to Case (7,8). The \plusred{0} partitions we constructed there instead become \plusred{1} partitions.
Thus, we assume the vertex to the left exists and is blue. Now, we consider the two highlighted vertices in \cref{fig:branch4-fig1b}. We consider 3 cases depending on whether they are disposable or not.

\paragraph{Subcase 1:} If both of them are disposable, we can simply flip them both and then flip $u$ to red to get a \plusred{3} partition as depicted in \cref{fig:branch4-fig1c}.

\begin{figure}[H]
     \centering
     \hfill
     \begin{subfigure}[b]{0.3\textwidth}
         \centering
           \includegraphics[width=\linewidth]{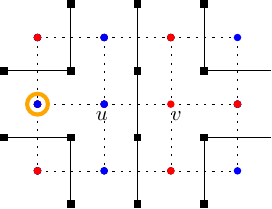}
         \caption{Step 1.}
         \label{fig:branch4-fig1a}
     \end{subfigure}
     \hfill
     \begin{subfigure}[b]{0.3\textwidth}
         \centering
             \includegraphics[width=\linewidth]{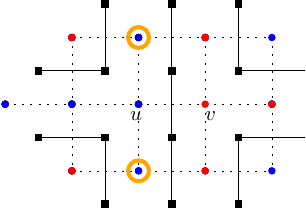}
         \caption{Step 2.}
    \label{fig:branch4-fig1b}
    \end{subfigure}
     \hfill
    \begin{subfigure}[b]{0.3\textwidth}
         \centering
    \includegraphics[width=\linewidth]{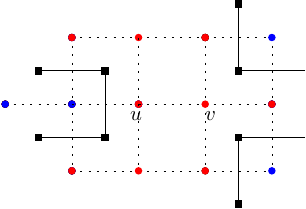}
         \caption{Step 3.}
         \label{fig:branch4-fig1c}
     \end{subfigure}
     \hfill
     \bigskip
    \caption{}
    \label{}
\end{figure}

\paragraph{Subcase 2:} If exactly one of the highlighted vertices in \cref{fig:branch4-fig1b} was disposable, then we can flip it to red to get the coloring in \cref{fig:branch4-fig3b} (or a symmetric version of it). If the highlighted vertex in \cref{fig:branch4-fig3b} was disposable, we can flip it to blue, and flip $u$ to red to get a \plusred{1} partition. Thus, we assume it is not disposable and from \cref{lem:elbow}, we get the coloring in \cref{fig:branch4-fig3c}.
We flip the colors of $u$, and the vertex above $u$ to get the coloring in \cref{fig:branch4-fig4a}. This creates two regions, one of which must be an island (\cref{lem:create-island}). Now, we must have a 1-thin structure between the 3 highlighted vertices in \cref{fig:branch4-fig4a}. Resolving this will give us a \plusred{2} partition.

\begin{figure}[H]
     \centering
     \hfill
     \begin{subfigure}[b]{0.3\textwidth}
         \centering
             \includegraphics[width=\linewidth]{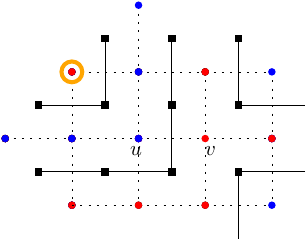}
         \caption{Step 3.}
         \label{fig:branch4-fig3b}
     \end{subfigure}
     \hfill
    \begin{subfigure}[b]{0.3\textwidth}
         \centering
    \includegraphics[width=\linewidth]{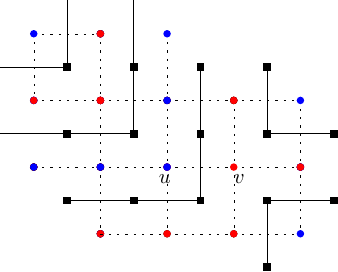}
         \caption{Step 4.}
         \label{fig:branch4-fig3c}
     \end{subfigure}
     \hfill
           \begin{subfigure}[b]{0.3\textwidth}
         \centering
           \includegraphics[width=\linewidth]{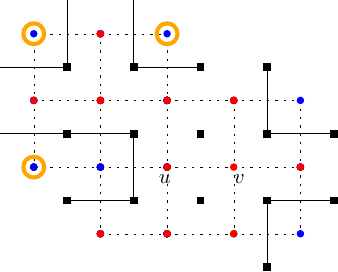}
         \caption{Step 5.}
         \label{fig:branch4-fig4a}
     \end{subfigure}
     
     \bigskip
    \caption{}
    \label{}
\end{figure}

\paragraph{Subcase 3:} If none of the highlighted vertices in \cref{fig:branch4-fig1b} were disposable, we get the figure in \cref{fig:branch4-fig4b}. We can assume that the symmetric red copies of the blue vertices we considered are also not disposable, since we can deal with them in a previous subcase otherwise. Consider the two red and two blue highlighted vertices in \cref{fig:branch4-fig4b}. If at least one of the red and at least one of the blue highlighted vertices were disposable, we can flip the red one to blue and the blue one to red. We did not affect the total number of red or blue vertices, and the neighborhood of $u$ is now in \subref{case:5}. We can continue the arguments therein.
Thus, we assume that for one color, neither highlighted vertex is disposable. Without loss of generality, assume that both the red ones are not disposable. By applying \cref{lem:elbow}, we get the coloring in \cref{fig:branch4-fig4c}.

\begin{figure}[H]
     \centering
      \hfill
     \begin{subfigure}[b]{0.3\textwidth}
         \centering
             \includegraphics[width=\linewidth]{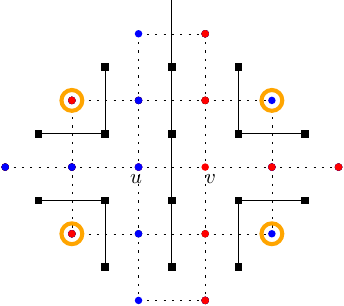}
         \caption{Step 3.}
         \label{fig:branch4-fig4b}
     \end{subfigure}
       \hfill
    \begin{subfigure}[b]{0.3\textwidth}
         \centering
    \includegraphics[width=\linewidth]{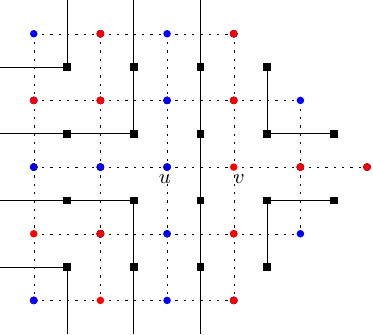}
         \caption{Step 4.}
         \label{fig:branch4-fig4c}
     \end{subfigure}
   \hfill
     \begin{subfigure}[b]{0.3\textwidth}
         \centering
           \includegraphics[width=\linewidth]{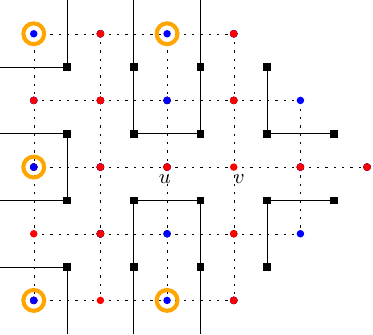}
         \caption{Step 5.}
         \label{fig:branch4-fig5a}
     \end{subfigure}
     \bigskip 
    \caption{}
    \label{}
\end{figure}

We now flip the color of $u$ and the vertex to the left of it to red. From \cref{lem:create-island}, this creates three regions. The resulting coloring is in \cref{fig:branch4-fig5a}. We note that the 5 highlighted vertices must belong to 3 regions, of which at least 2 must be islands.
This implies that there must exist two 1-thin structures between them. Resolving these gives us a \plusred{0} partition. One such way to resolve it along with the connections between the regions is shown in \cref{fig:branch4-fig5b}.

\begin{figure}[H]
     \centering
     \includegraphics[width=0.4\linewidth]{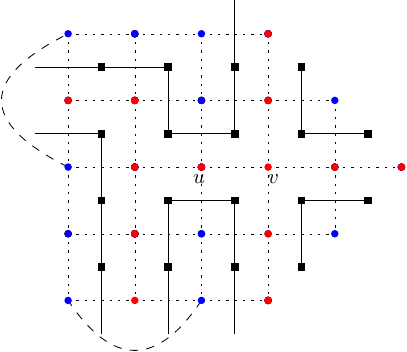}
         \caption{Final state.}
         \label{fig:branch4-fig5b}
\end{figure}

\section{Boundary Case Analysis}
\label{sec:boundary-cases}

\graphicspath{ {./img/grid_boundary} }

\newcommand{\minusred}[1]{\textcolor{red}{$-$#1}}

We now handle the cases where either $u$ or $v$ is adjacent to the border of the grid. We remark that we do not explicitly state that the degree of the outer face vertex $\outerface$ does not increase, although all our transformations do ensure this fact.

As in \cref{sec:reconnect-cases}, we say that a partition is \minusred{$x$} if both $u,v$ are colored red and the number of vertices removed from the red partition is $x$. We show that all of our partitions in these cases are \plusblue{$x$}, \plusred{$x$}, or \minusred{$x$} for $0 \leq x \leq 3$.

The border cases are depicted in \cref{fig:bcases1} and \cref{fig:bcases2}, with symmetric cases removed. The solid black line depicts the boundary of the grid. For visual simplicity, we display the graphs without the cycles in this section.

\paragraph{Case Notation} 
Broadly, we again have three branches for our case analysis. 
\begin{enumerate}
    \item Both $u,v$ are adjacent to the border. This is $\fbox{\mbox{Case 1}}$.

For the other cases, without loss of generality, we assume that $v$ is adjacent to the right border and that $u$ is to the left of $v$.
    
    \item At least one of the top and bottom neighbors of $v$ is blue. This is $\fbox{\mbox{Case 2}}$.
    \item Both the top and bottom neighbors of $v$ are red. We have several cases in this branch.
    \begin{enumerate}
        \item At least one of the top and bottom neighbors of $u$ is red. This is $\fbox{\mbox{Case 3}}$.
        \item Both the top and bottom neighbors of $u$ are blue. We further branch based on the colors of the vertices to the left of $u$. This covers $\fbox{\mbox{Cases 4-12}}$.
    \end{enumerate}
\end{enumerate}

\renewcommand\thesubfigure{Case~\arabic{subfigure}}
\begin{figure}[H]
     \centering
     
     \begin{subfigure}[c]{0.3\textwidth}
         \centering
           \includegraphics[width=0.6\linewidth]{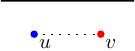}
         \caption{}
         \label{bcase:1}
     \end{subfigure}
     \hfill
     \begin{subfigure}[c]{0.29\textwidth}
         \centering
             \includegraphics[width=0.6\linewidth,cframe=green]{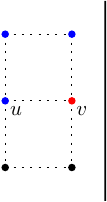}
        \caption{}
         \label{bcase:2}
     \end{subfigure}
     \hfill
    \begin{subfigure}[c]{0.3\textwidth}
         \centering
    \includegraphics[width=0.6\linewidth]{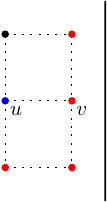}
        \caption{}
         \label{bcase:3}
     \end{subfigure}

     \bigskip

    \begin{subfigure}[c]{0.3\textwidth}
         \centering
           \includegraphics[width=\linewidth]{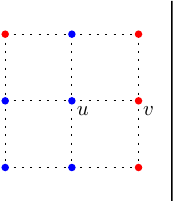}
         \caption{}
         \label{bcase:4}
     \end{subfigure}
     \hfill
     \begin{subfigure}[c]{0.3\textwidth}
         \centering
             \includegraphics[width=\linewidth]{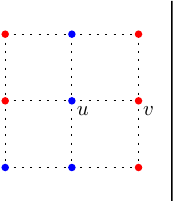}
        \caption{}
         \label{bcase:5}
     \end{subfigure}
     \hfill
    \begin{subfigure}[c]{0.29\textwidth}
         \centering
    \includegraphics[width=\linewidth,cframe=green]{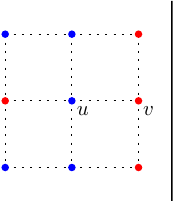}
        \caption{}
         \label{bcase:6}
     \end{subfigure}
     
    \caption{Easy cases are drawn with a green box surround them.}
    \label{fig:bcases1}
\end{figure}

\renewcommand\thesubfigure{Case \arabic{subfigure}}
\begin{figure}[H]
     \centering

     \begin{subfigure}[c]{0.29\textwidth}
        \addtocounter{subfigure}{6}
         \centering
           \includegraphics[width=0.8\linewidth,cframe=green]{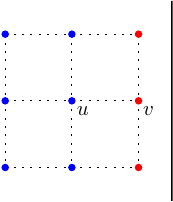}
         \caption{}
         \label{bcase:7}
     \end{subfigure}
     \hfill
     \begin{subfigure}[c]{0.29\textwidth}
         \centering
             \includegraphics[width=\linewidth]{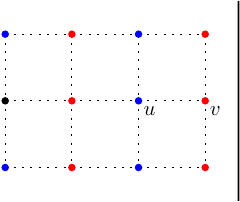}
        \caption{}
         \label{bcase:8}
     \end{subfigure}
     \hfill
    \begin{subfigure}[c]{0.3\textwidth}
         \centering
    \includegraphics[width=\linewidth]{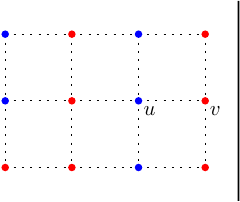}
        \caption{}
         \label{bcase:9}
     \end{subfigure}

     \bigskip

     \begin{subfigure}[c]{0.3\textwidth}
         \centering
           \includegraphics[width=\linewidth]{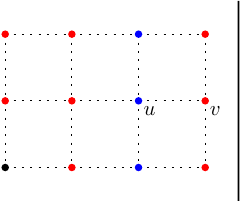}
         \caption{}
         \label{bcase:10}
     \end{subfigure}
     \hfill
     \begin{subfigure}[c]{0.3\textwidth}
         \centering
             \includegraphics[width=\linewidth]{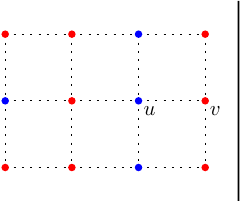}
        \caption{}
         \label{bcase:11}
     \end{subfigure}
     \hfill
    \begin{subfigure}[c]{0.3\textwidth}
         \centering
    \includegraphics[width=0.8\linewidth]{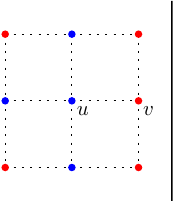}
        \caption{}
         \label{bcase:12}
     \end{subfigure}
     
    \caption{}
    \label{fig:bcases2}
\end{figure}
\renewcommand{\thesubfigure}{\alph{subfigure}}

In Case 2, Case 6, and Case 7, we can directly modify the graph with the final states drawn in \cref{fig:boundary_easy}.

\renewcommand\thesubfigure{Case \arabic{subfigure}}
\begin{figure}[H]
     \centering
     
     \begin{subfigure}[c]{0.3\textwidth}
         \centering
           \includegraphics[width=0.6\linewidth]{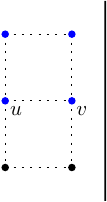}
        
        \setcounter{subfigure}{1}
         \caption{}
     \end{subfigure}
     \hfill
     \begin{subfigure}[c]{0.3\textwidth}
         \centering
             \includegraphics[width=\linewidth]{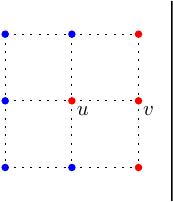}

        \setcounter{subfigure}{5}
        \caption{}
     \end{subfigure}
     \hfill
    \begin{subfigure}[c]{0.3\textwidth}
         \centering
    \includegraphics[width=\linewidth]{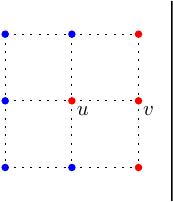}
    
        \setcounter{subfigure}{6}
        \caption{}
     \end{subfigure}

    \caption{Modifications for easy cases.}
    \label{fig:boundary_easy}
\end{figure}
\renewcommand{\thesubfigure}{\alph{subfigure}}

\subsection{Case 1}

Uniquely, Case 1 handles the scenario where both $u$ and $v$ are both adjacent to a border. Firstly, if either $u$ or $v$ in \cref{fig:boundary-1-1} were disposable, then we can flip it to obtain either a \plusred{1} or \plusblue{1} partition.
The neighborhood of $u,v$ must look like \cref{fig:boundary-1-2} since any other configuration would imply that one of $u$ or $v$ is disposable.
Next, we consider the highlighted vertices in \cref{fig:boundary-1-2}.
If, say, the red highlighted vertex in \cref{fig:boundary-1-2} was disposable, then we can flip it to blue, and flip $u$ to red to get a \plusred{0} partition.
Similarly for the blue highlighted vertex.
Assuming they are not disposable, we can again apply \cref{lem:elbow} and reason that the state of the graph resembles \cref{fig:boundary-1-3}. Next, we examine the highlighted vertices of \cref{fig:boundary-1-3}. If the rightmost highlighted vertex is blue, then we can flip $v$ and the vertex below it to obtain a \plusblue{2} partition. Likewise, if either of the left two highlighted vertices is red, then we can obtain a \plusred{2} partition. Thus, we can assume that the left two vertices are blue, and the right vertex is red.

\begin{figure}[H]
    \centering
    \begin{subfigure}[b]{0.3\textwidth}
        \centering
            \includegraphics[width=0.7\linewidth]{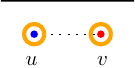}
        \caption{Step 1.}
        \label{fig:boundary-1-1}
    \end{subfigure}
    \hfill
    \begin{subfigure}[b]{0.3\textwidth}
        \centering
            \includegraphics[width=\linewidth]{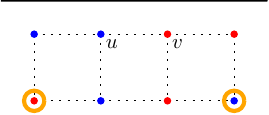}
        \caption{Step 2.}
        \label{fig:boundary-1-2}
    \end{subfigure}
    \hfill
    \begin{subfigure}[b]{0.3\textwidth}
        \centering
            \includegraphics[width=\linewidth]{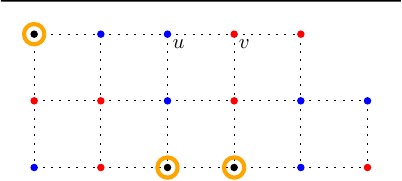}
        
        \caption{Step 3.}
        \label{fig:boundary-1-3}
    \end{subfigure}
    
    \caption{}
    \label{}
\end{figure}

Next, we flip $u$ and the vertex below to red, as depicted in \cref{fig:boundary-1-4}. From arguments similar to \cref{lem:create-island}, this creates two blue regions, one of which is an island. We then have at least one 1-thin structures at the highlighted vertices in \cref{fig:boundary-1-5}. Resolving the 1-thin structure results in a \plusred{1} partition. One such final partition is \cref{fig:boundary-1-6}.

\begin{figure}[H]
    \centering
    \begin{subfigure}[b]{0.3\textwidth}
        \centering
            \includegraphics[width=\linewidth]{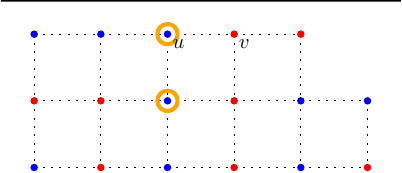}
        \caption{Step 4.}
        \label{fig:boundary-1-4}
    \end{subfigure}
    \hfill
    \begin{subfigure}[b]{0.3\textwidth}
        \centering
            \includegraphics[width=\linewidth]{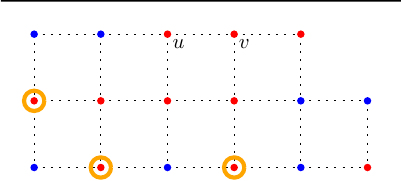}
        \caption{Step 5.}
        \label{fig:boundary-1-5}
    \end{subfigure}
    \hfill
    \begin{subfigure}[b]{0.3\textwidth}
        \centering
            \includegraphics[width=\linewidth]{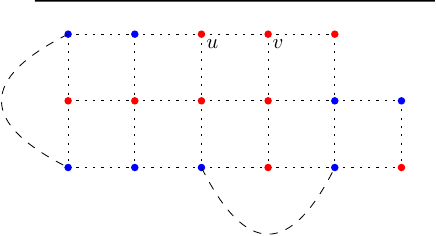}
        \caption{Final state.}
        \label{fig:boundary-1-6}
    \end{subfigure}
    
    \caption{}
    \label{}
\end{figure}

\subsection{Cases 3-5}

We handle Case 3, Case 4, and Case 5 in sequence since they resolve similarly. 

\paragraph{Case 3:} If $u$ is disposable, then we can simply flip $u$ to obtain a \plusred{1} partition. Otherwise, we can reveal more information about the neighborhood of $u$ using \cref{lem:elbow}. Repeating this logic on the highlighted vertex in \cref{fig:boundary-3-2} yields either a \plusred{0} partition or further information. We can once more apply \cref{lem:elbow} on the left highlighted vertex in \cref{fig:boundary-3-3} to get either a \plusred{1} partition, or again reveal more information about the neighboring vertices. Furthermore, if the right highlighted vertex in \cref{fig:boundary-3-3} is red, then we can flip $u$ and the vertex above it to yield a \plusred{2} partition. Thus, we can assume that the right highlighted vertex is blue.

\begin{figure}[H]
    \centering
    \begin{subfigure}[b]{0.3\textwidth}
        \centering
            \includegraphics[width=0.6\linewidth]{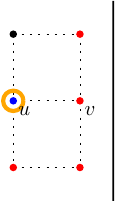}
        \caption{Step 1.}
        \label{fig:boundary-3-1}
    \end{subfigure}
    \hfill
    \begin{subfigure}[b]{0.3\textwidth}
        \centering
            \includegraphics[width=\linewidth]{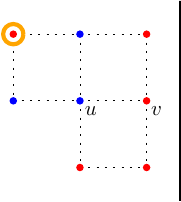}
        \caption{Step 2.}
        \label{fig:boundary-3-2}
    \end{subfigure}
    \hfill
    \begin{subfigure}[b]{0.3\textwidth}
        \centering
            \includegraphics[width=\linewidth]{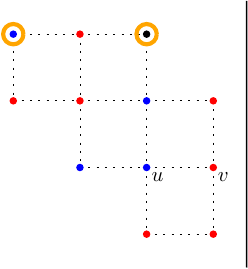}
        \caption{Step 3.}
        \label{fig:boundary-3-3}
    \end{subfigure}
    
    \caption{}
    \label{}
\end{figure}

We next flip the two highlighted vertices in \cref{fig:boundary-3-4}. By \cref{lem:create-island}, this creates two blue regions, at least one of which is an island. If the highlighted vertices in \cref{fig:boundary-3-5} belong to different regions, then we can resolve the resulting 1-thin structure to obtain a \plusred{1} partition. Otherwise, consider the highlighted vertices in \cref{fig:boundary-3-6}. If either were blue, then a 1-thin structure would form that resolves into a \plusred{1} partition. We therefore assume that they are red. 

\begin{figure}[H]
    \centering
    \begin{subfigure}[b]{0.3\textwidth}
        \centering
            \includegraphics[width=\linewidth]{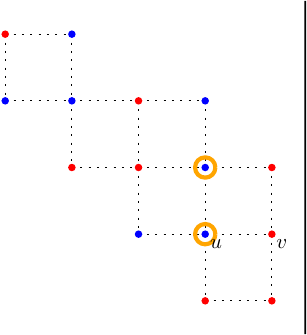}
        \caption{Step 4.}
        \label{fig:boundary-3-4}
    \end{subfigure}
    \hfill
    \begin{subfigure}[b]{0.3\textwidth}
        \centering
            \includegraphics[width=\linewidth]{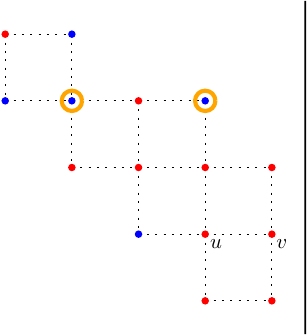}
        \caption{Step 5.}
        \label{fig:boundary-3-5}
    \end{subfigure}
    \hfill
    \begin{subfigure}[b]{0.3\textwidth}
        \centering
            \includegraphics[width=\linewidth]{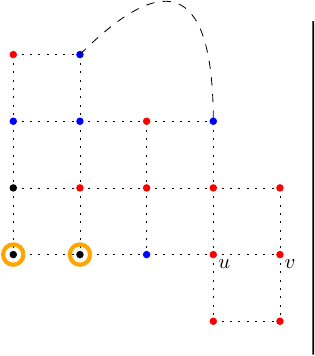}
        \caption{Step 6.}
        \label{fig:boundary-3-6}
    \end{subfigure}
    
    \caption{}
    \label{}
\end{figure}

If the highlighted vertex in \cref{fig:boundary-3-7} was red, then flipping the vertex to the left of $u$ to red
creates a \plusred{3} partition. If not, then we can examine the highlighted vertices in \cref{fig:boundary-3-8}. If any are blue, then we can resolve either a 1-thin or 2-thin structure to produce a \plusred{0} or \plusred{1} partition. If both are red, then we flip the highlighted vertices in \cref{fig:boundary-3-9} to reach a \plusred{0} partition. To justify that the right and bottom-left sets of red vertices in \cref{fig:boundary-3-10} remain connected, note that one of the two blue regions in \cref{fig:boundary-3-8} is an island. By \cref{lem:island-walk}, the island walk around that island guarantees that there remains a path connecting the two sets of red vertices after flipping the vertices in Step 9. A similar idea is used for \cref{fig:branch1-fig4c}.


\begin{figure}[H]
    \centering
    \begin{subfigure}[b]{0.3\textwidth}
        \centering
            \includegraphics[width=\linewidth]{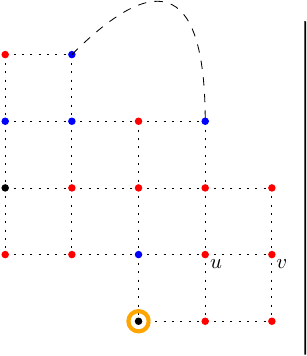}
        \caption{Step 7.}
        \label{fig:boundary-3-7}
    \end{subfigure}
    \hfill
    \begin{subfigure}[b]{0.3\textwidth}
        \centering
            \includegraphics[width=\linewidth]{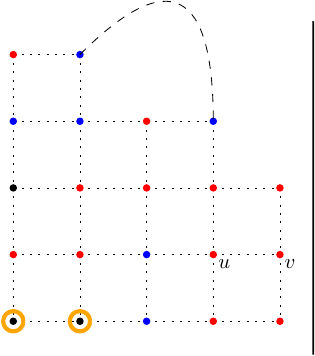}
        \caption{Step 8.}
        \label{fig:boundary-3-8}
    \end{subfigure}
    \hfill
    \begin{subfigure}[b]{0.3\textwidth}
        \centering
            \includegraphics[width=\linewidth]{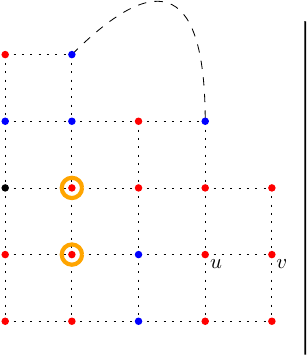}
        \caption{Step 9.}
        \label{fig:boundary-3-9}
    \end{subfigure}
    
    \caption{}
    \label{}
\end{figure}

\begin{figure}[H]
    \centering
    \begin{subfigure}[b]{0.4\textwidth}
        \centering
            \includegraphics[width=\linewidth]{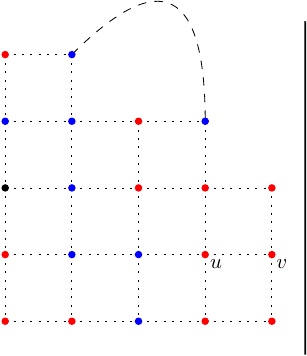}
        \caption{Final state.}
        \label{fig:boundary-3-10}
    \end{subfigure}
    
    \caption{}
    \label{}
\end{figure}

\paragraph{Case 4:} For Case 4, we focus on reaching the same state as in Case 3, after which we can follow the same logic as before. We first examine the highlighted vertex in \cref{fig:boundary-4-1}. If it were disposable, then we can flip that vertex and $u$ to obtain a \plusred{0} partition. If not, then we use \cref{lem:elbow} to reveal more information. We can repeat this line of reasoning on the left highlighted vertex in \cref{fig:boundary-4-2} to obtain either a \plusred{1} partition or more information. We also know that, if the right highlighted vertex in \cref{fig:boundary-4-2} was red, then we can flip $u$ and the vertex above $u$ to obtain a \plusred{2} partition. Thus, we assume that that vertex is blue, resulting in the state depicted in \cref{fig:boundary-4-3}. We next flip both of the highlighted vertices in \cref{fig:boundary-4-3}. Doing so creates two blue regions, one of which must be an island by \cref{lem:create-island}.

\begin{figure}[H]
    \centering
    \begin{subfigure}[b]{0.3\textwidth}
        \centering
            \includegraphics[width=\linewidth]{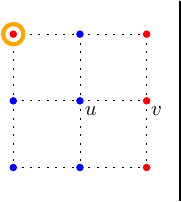}
        \caption{Step 1.}
        \label{fig:boundary-4-1}
    \end{subfigure}
    \hfill
    \begin{subfigure}[b]{0.3\textwidth}
        \centering
            \includegraphics[width=\linewidth]{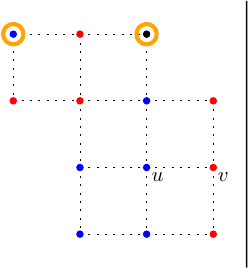}
        \caption{Step 2.}
        \label{fig:boundary-4-2}
    \end{subfigure}
    \hfill
    \begin{subfigure}[b]{0.3\textwidth}
        \centering
            \includegraphics[width=\linewidth]{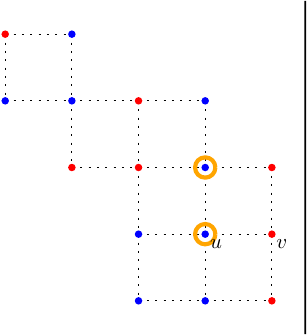}
        \caption{Step 3.}
        \label{fig:boundary-4-3}
    \end{subfigure}
    
    \caption{}
    \label{}
\end{figure}

If the two highlighted vertices in \cref{fig:boundary-4-4} belong to different regions, then we can resolve the 1-thin structure between them to achieve a \plusred{1} partition. If they belong to the same region, then we can examine the highlighted vertices in \cref{fig:boundary-4-5}. If either was blue, then we can resolve a 1-thin structure to get a \plusred{1} partition. If both are red, then we obtain the state in \cref{fig:boundary-4-6}. This is nearly identical to the state in \cref{fig:boundary-3-8}, and we can follow the exact same steps to produce either a \plusred{0} or \plusred{1} partition. One such final partition is shown in \cref{fig:boundary-4-7}.

\begin{figure}[H]
    \centering
    \begin{subfigure}[b]{0.3\textwidth}
        \centering
            \includegraphics[width=\linewidth]{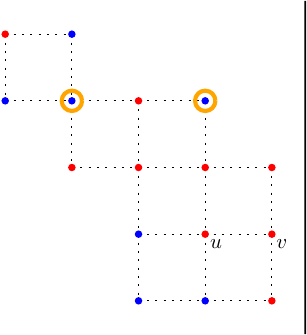}
        \caption{Step 4.}
        \label{fig:boundary-4-4}
    \end{subfigure}
    \hfill
    \begin{subfigure}[b]{0.3\textwidth}
        \centering
            \includegraphics[width=\linewidth]{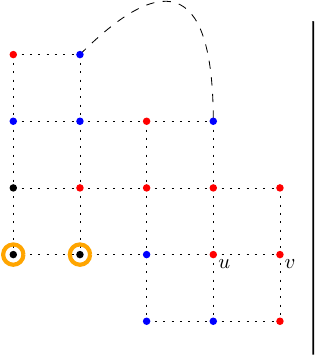}
        \caption{Step 5.}
        \label{fig:boundary-4-5}
    \end{subfigure}
    \hfill
    \begin{subfigure}[b]{0.3\textwidth}
        \centering
            \includegraphics[width=\linewidth]{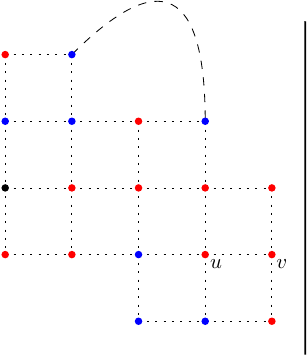}
        \caption{Step 6.}
        \label{fig:boundary-4-6}
    \end{subfigure}
    
    \caption{}
    \label{}
\end{figure}

\begin{figure}[H]
    \centering
    \begin{subfigure}[b]{0.4\textwidth}
        \centering
            \includegraphics[width=\linewidth]{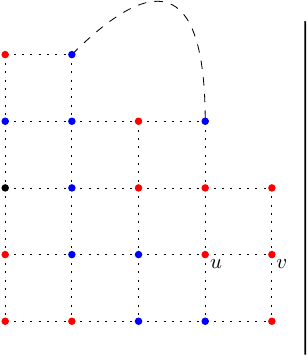}
        \caption{Final state.}
        \label{fig:boundary-4-7}
    \end{subfigure}
    
    \caption{}
    \label{}
\end{figure}

\paragraph{Case 5:} Firstly, if the highlighted vertex in \cref{fig:boundary-5-1} was disposable, then we can flip it to return to Case 4. This would give any of a \minusred{1}, \plusred{0}, or \plusred{1} partition. This allows us to assume that it is not disposable, and so \cref{lem:elbow} tells us more information about the neighboring vertices. Similarly, if the highlighted vertex in \cref{fig:boundary-5-2} was disposable, then we can flip it to red, so that the vertex to the left of $u$ becomes disposable, and by the same argument as before, we obtain a \plusred{0}, \plusred{1}, or \plusred{2} partition. If it is not disposable, then we again apply \cref{lem:elbow} to get \cref{fig:boundary-5-3}.
Next, we flip $u$ to red. By \cref{lem:create-island}, this creates two blue regions, one of which is an island. 

\begin{figure}[H]
    \centering
    \begin{subfigure}[b]{0.3\textwidth}
        \centering
            \includegraphics[width=\linewidth]{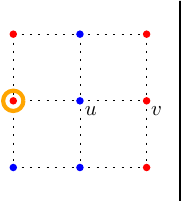}
        \caption{Step 1.}
        \label{fig:boundary-5-1}
    \end{subfigure}
    \hfill
    \begin{subfigure}[b]{0.3\textwidth}
        \centering
            \includegraphics[width=\linewidth]{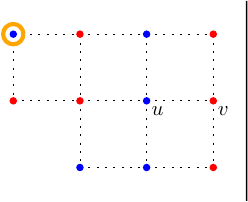}
        \caption{Step 2.}
        \label{fig:boundary-5-2}
    \end{subfigure}
    \hfill
    \begin{subfigure}[b]{0.3\textwidth}
        \centering
            \includegraphics[width=\linewidth]{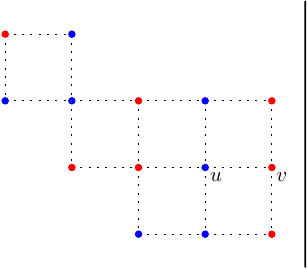}
        \caption{Step 3.}
        \label{fig:boundary-5-3}
    \end{subfigure}
    
    \caption{}
    \label{}
\end{figure}

If the two highlighted vertices in \cref{fig:boundary-5-4} are in different regions, then we can resolve the 1-thin structure between them to make a \plusred{0} partition. Otherwise, we consider the highlighted vertices in \cref{fig:boundary-5-5}. If either was blue, then this creates a 1-thin structure that can be resolved to produce a \plusred{0} partition. We therefore assume that they are red. Next, if both of the pink-highlighted vertices in \cref{fig:boundary-5-6} are red (or don't exist), then we can flip the vertices above them to red to produce a \plusred{3} partition.
Otherwise, one of the pink-highlighted vertices is blue. If any of the yellow-highlighted vertices are also blue, then this creates a 1-thin or 2-thin structure that can be resolved to make a \minusred{1} or \plusred{0} partition. We assume from here that the two yellow-highlighted vertices are red.

\begin{figure}[H]
    \centering
    \begin{subfigure}[b]{0.3\textwidth}
        \centering
            \includegraphics[width=\linewidth]{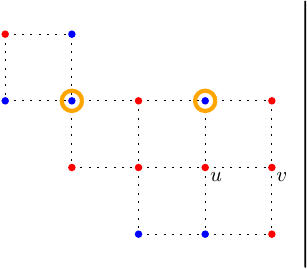}
        \caption{Step 4.}
        \label{fig:boundary-5-4}
    \end{subfigure}
    \hfill
    \begin{subfigure}[b]{0.3\textwidth}
        \centering
            \includegraphics[width=\linewidth]{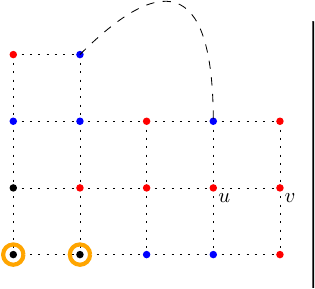}
        \caption{Step 5.}
        \label{fig:boundary-5-5}
    \end{subfigure}
    \hfill
    \begin{subfigure}[b]{0.3\textwidth}
        \centering
            \includegraphics[width=\linewidth]{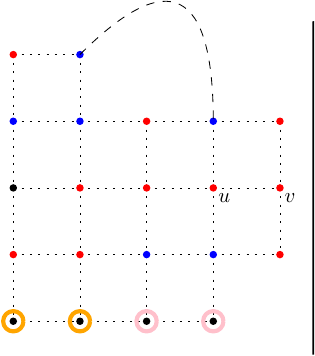}
        \caption{Step 6.}
        \label{fig:boundary-5-6}
    \end{subfigure}
    
    \caption{}
    \label{}
\end{figure}

We lastly flip the two highlighted vertices in \cref{fig:boundary-5-7} to create a \minusred{1} partition. The right and bottom-left sets of red vertices of \cref{fig:boundary-5-8} remain connected, through an argument similar to \cref{fig:boundary-3-9}.

\begin{figure}[H]
    \centering
    \begin{subfigure}[b]{0.4\textwidth}
        \centering
            \includegraphics[width=\linewidth]{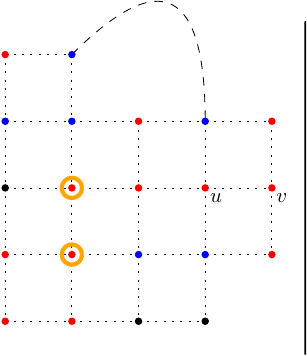}
        \caption{Step 7.}
        \label{fig:boundary-5-7}
    \end{subfigure}
    \hfill
    \begin{subfigure}[b]{0.4\textwidth}
        \centering
            \includegraphics[width=\linewidth]{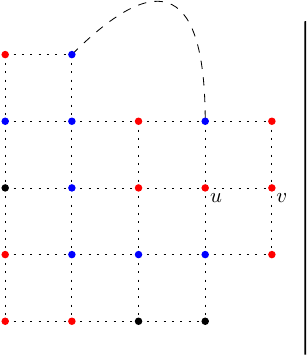}
        \caption{Final state.}
        \label{fig:boundary-5-8}
    \end{subfigure}
    
    \caption{}
    \label{}
\end{figure}

\subsection{Case 8}

We first flip $u$ to red, creating two blue regions including one island as argued by \cref{lem:create-island}. Then, at least one of the highlighted vertices in \cref{fig:boundary-8-2} must be in a 1-thin structure, so we can flip that vertex to get a \plusred{0} partition. A possible final partition is depicted in \cref{fig:boundary-8-3}.

\begin{figure}[H]
    \centering
    \begin{subfigure}[b]{0.3\textwidth}
        \centering
            \includegraphics[width=\linewidth]{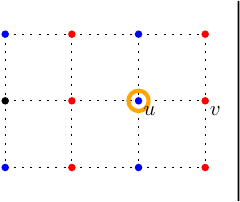}
        \caption{Step 1.}
        \label{fig:boundary-8-1}
    \end{subfigure}
    \hfill
    \begin{subfigure}[b]{0.3\textwidth}
        \centering
            \includegraphics[width=\linewidth]{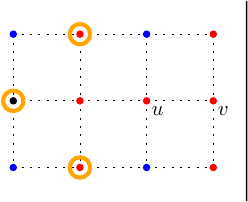}
        \caption{Step 2.}
        \label{fig:boundary-8-2}
    \end{subfigure}
    \hfill
    \begin{subfigure}[b]{0.3\textwidth}
        \centering
            \includegraphics[width=\linewidth]{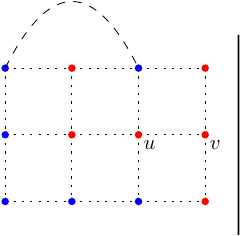}
        \caption{Final state.}
        \label{fig:boundary-8-3}
    \end{subfigure}
    
    \caption{}
    \label{}
\end{figure}

\subsection{Case 9}

We first examine the highlighted vertices in \cref{fig:boundary-9-1}. If one of them was disposable, then we could flip it and flip $u$ to obtain a \plusred{2} partition. We therefore assume that both vertices are not disposable. This can only be the case if neither is adjacent to a border and they are each adjacent to other blue vertices. This state is depicted in \cref{fig:boundary-9-1}. Next, we flip $u$ to red which creates two blue regions. One of these regions in an island by \cref{lem:create-island}. Then, if the two highlighted vertices in \cref{fig:boundary-9-3} are in different regions, then they form a 1-thin structure that we can resolve to make a \plusred{0} partition. We assume then that the two highlighted vertices belong to the same region.

\begin{figure}[H]
    \centering
    \begin{subfigure}[b]{0.3\textwidth}
        \centering
            \includegraphics[width=\linewidth]{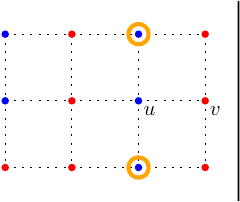}
        \caption{Step 1.}
        \label{fig:boundary-9-1}
    \end{subfigure}
    \hfill
    \begin{subfigure}[b]{0.3\textwidth}
        \centering
            \includegraphics[width=\linewidth]{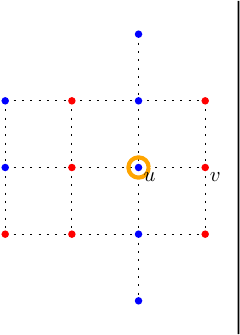}
        \caption{Step 2.}
        \label{fig:boundary-9-2}
    \end{subfigure}
    \hfill
    \begin{subfigure}[b]{0.3\textwidth}
        \centering
            \includegraphics[width=\linewidth]{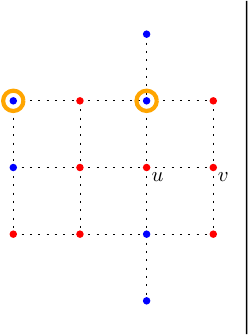}
        \caption{Step 3.}
        \label{fig:boundary-9-3}
    \end{subfigure}
    
    \caption{}
    \label{}
\end{figure}

If the highlighted vertex in \cref{fig:boundary-9-4} is blue, then there is a 1-thin structure that resolves to form a \plusred{0} partition. If the vertex is red, then we flip the highlighted vertices in \cref{fig:boundary-9-5}. Note that at least one of the blue regions remains an island. This means that flipping the highlighted vertices in \cref{fig:boundary-9-6} does not disconnect the red region, as discussed for \cref{fig:boundary-3-9}. Thus, we get a \plusred{0} partition as shown in \cref{fig:boundary-9-7}.

\begin{figure}[H]
    \centering
    \begin{subfigure}[b]{0.3\textwidth}
        \centering
            \includegraphics[width=\linewidth]{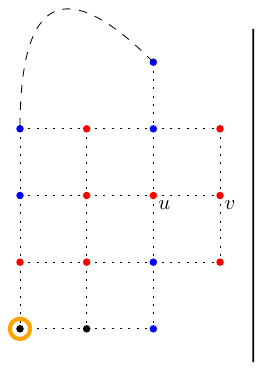}
        \caption{Step 4.}
        \label{fig:boundary-9-4}
    \end{subfigure}
    \hfill
    \begin{subfigure}[b]{0.3\textwidth}
        \centering
            \includegraphics[width=\linewidth]{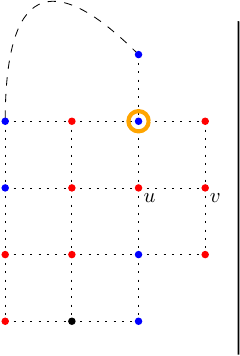}
        \caption{Step 5.}
        \label{fig:boundary-9-5}
    \end{subfigure}
    \hfill
    \begin{subfigure}[b]{0.3\textwidth}
        \centering
            \includegraphics[width=\linewidth]{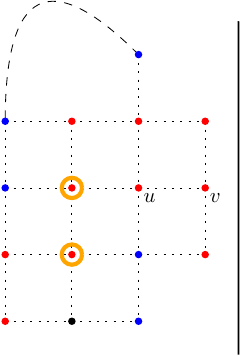}
        \caption{Step 6.}
        \label{fig:boundary-9-6}
    \end{subfigure}
    
    \caption{}
    \label{}
\end{figure}

\begin{figure}[H]
    \centering
    \begin{subfigure}[b]{0.4\textwidth}
        \centering
            \includegraphics[width=\linewidth]{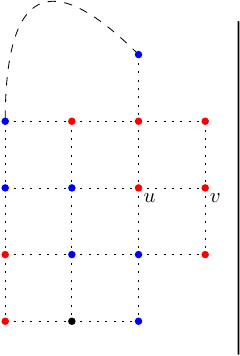}
        \caption{Final state.}
        \label{fig:boundary-9-7}
    \end{subfigure}
    
    \caption{}
    \label{}
\end{figure}

\subsection{Cases 10 and 11}

\paragraph{Case 10:} If the highlighted vertex in \cref{fig:boundary-10-1} is disposable, then we can flip it to reduce to Case 5. This produces any of a \minusred{2}, \minusred{1}, \plusred{0}, or \plusred{1} partition. 
If it is not disposable, the coloring must look like \cref{fig:boundary-10-2}, since if the highlighted vertex in \cref{fig:boundary-10-1} had a blue vertex (or nothing else) above it, it would have been disposable. Similarly, if it had a red vertex above it, but the vertex above and to the left of it was red, it would have been disposable.
We can use \cref{lem:elbow} to reveal more information for the left highlighted vertex in \cref{fig:boundary-10-2} to similarly conclude that if it is disposable, we can flip it to red, flip the highlighted vertex in \cref{fig:boundary-10-1} to blue, and then go back to a previous case can produce any of a \minusred{1}, \plusred{0}, \plusred{1}, or \plusred{2} partition. Otherwise, that vertex is not disposable and so we can reveal more information using \cref{lem:elbow}. Furthermore, if the right highlighted vertex in \cref{fig:boundary-10-2} is red, then we can flip $u$ and the vertex above it to obtain a \plusred{2} partition. We can therefore assume that that vertex is blue. We thus get the coloring in \cref{fig:boundary-10-3}. We next flip $u$, creating two blue regions. We know by \cref{lem:create-island} that one of these regions is an island.

\begin{figure}[H]
    \centering
    \begin{subfigure}[b]{0.3\textwidth}
        \centering
            \includegraphics[width=\linewidth]{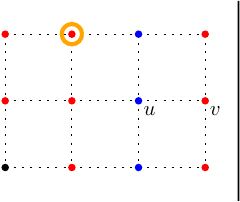}
        \caption{Step 1.}
        \label{fig:boundary-10-1}
    \end{subfigure}
    \hfill
    \begin{subfigure}[b]{0.3\textwidth}
        \centering
            \includegraphics[width=\linewidth]{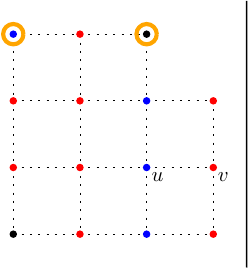}
        \caption{Step 2.}
        \label{fig:boundary-10-2}
    \end{subfigure}
    \hfill
    \begin{subfigure}[b]{0.3\textwidth}
        \centering
            \includegraphics[width=\linewidth]{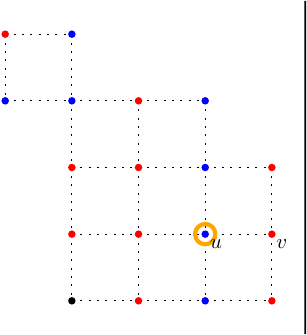}
        \caption{Step 3.}
        \label{fig:boundary-10-3}
    \end{subfigure}
    
    \caption{}
    \label{}
\end{figure}

If the two highlighted vertices in \cref{fig:boundary-10-4} are in different regions, then they form a 1-thin structure that can be flipped to produce a \plusred{0} partition. If they are in the same region, then we examine the three highlighted vertices in \cref{fig:boundary-10-5}. In the case where any of them are blue (from either region), we can resolve a 1-thin structure or 2-thin structure to form a \plusred{0} partition or a \minusred{1} partition respectively. Thus, we assume that the vertices are red. Next, we examine the two highlighted vertices in \cref{fig:boundary-10-6}. If the right vertex is red (or does not exist), then we can flip the vertex below $u$ to obtain a \plusred{2} partition. If the left vertex is blue, then we can flip it to red since it is disposable (\cref{obs:degree1-disposable}). Thus, we assume that the left vertex is red and the right vertex is blue. 

\begin{figure}[H]
    \centering
    \begin{subfigure}[b]{0.3\textwidth}
        \centering
            \includegraphics[width=\linewidth]{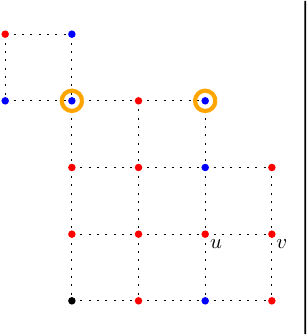}
        \caption{Step 4.}
        \label{fig:boundary-10-4}
    \end{subfigure}
    \hfill
    \begin{subfigure}[b]{0.3\textwidth}
        \centering
            \includegraphics[width=\linewidth]{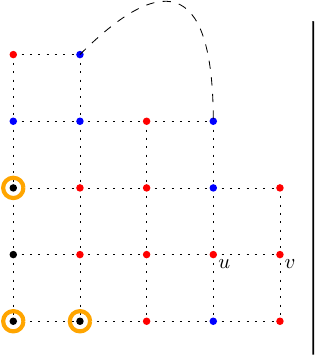}
        \caption{Step 5.}
        \label{fig:boundary-10-5}
    \end{subfigure}
    \hfill
    \begin{subfigure}[b]{0.3\textwidth}
        \centering
            \includegraphics[width=\linewidth]{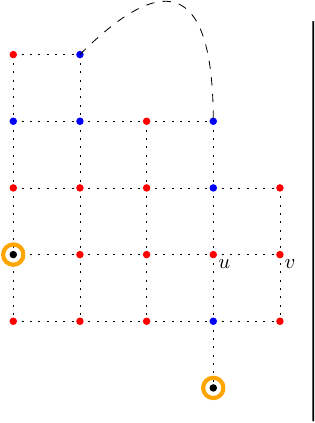}
        \caption{Step 6.}
        \label{fig:boundary-10-6}
    \end{subfigure}
    
    \caption{}
    \label{}
\end{figure}

We flip the highlighted vertex in \cref{fig:boundary-10-7}. Then, if the highlighted vertex in \cref{fig:boundary-10-8} is blue, then we can resolve a 1-thin or 2-thin structure to obtain a \minusred{2} partition. If it is red, then we can flip the highlighted vertex in \cref{fig:boundary-10-9}. We note that at least one of the blue regions remains an island.

\begin{figure}[H]
    \centering
    \begin{subfigure}[b]{0.3\textwidth}
        \centering
            \includegraphics[width=\linewidth]{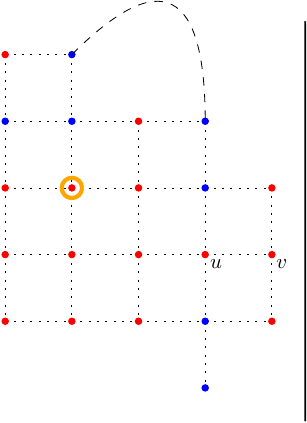}
        \caption{Step 7.}
        \label{fig:boundary-10-7}
    \end{subfigure}
    \hfill
    \begin{subfigure}[b]{0.3\textwidth}
        \centering
            \includegraphics[width=\linewidth]{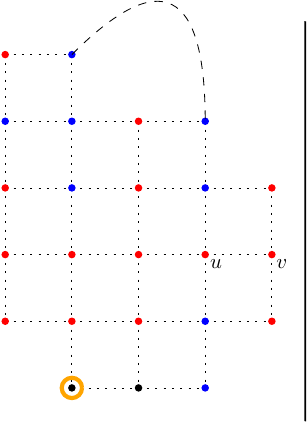}
        \caption{Step 8.}
        \label{fig:boundary-10-8}
    \end{subfigure}
    \hfill
    \begin{subfigure}[b]{0.3\textwidth}
        \centering
            \includegraphics[width=\linewidth]{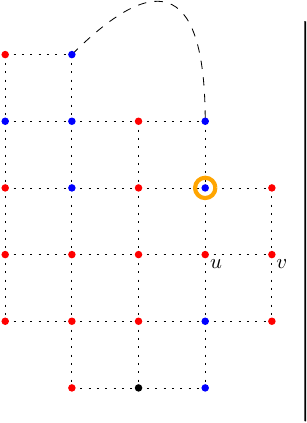}
        \caption{Step 9.}
        \label{fig:boundary-10-9}
    \end{subfigure}
    
    \caption{}
    \label{}
\end{figure}

We lastly flip all of the highlighted vertices in \cref{fig:boundary-10-10}, with the red region remaining connected as argued for \cref{fig:boundary-3-9}. This achieves a \minusred{2} partition, and the final state is depicted in \cref{fig:boundary-10-11}.

\begin{figure}[H]
    \centering
    \begin{subfigure}[b]{0.38\textwidth}
        \centering
            \includegraphics[width=\linewidth]{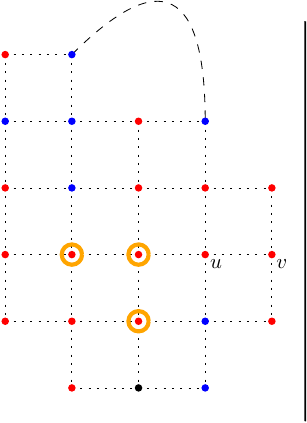}
        \caption{Step 10.}
        \label{fig:boundary-10-10}
    \end{subfigure}
    \hfill
    \begin{subfigure}[b]{0.38\textwidth}
        \centering
            \includegraphics[width=\linewidth]{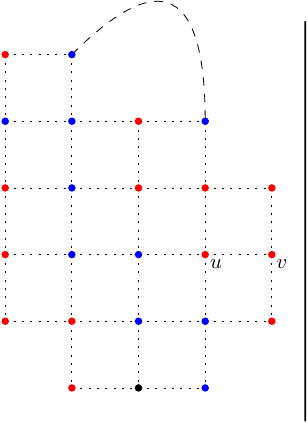}
        \caption{Final state.}
        \label{fig:boundary-10-11}
    \end{subfigure}
    
    \caption{}
    \label{}
\end{figure}

\paragraph{Case 11:} Case 11 can be easily reduced to Case 10 by flipping the highlighted vertex in \cref{fig:boundary-11-1}. This produces any of a \minusred{1}, \plusred{0}, \plusred{1}, \plusred{2}, or \plusred{3} partition.

\begin{figure}[H]
    \centering
    \begin{subfigure}[b]{0.4\textwidth}
        \centering
            \includegraphics[width=\linewidth]{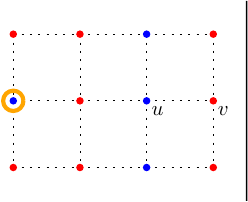}
        \caption{Step 1.}
        \label{fig:boundary-11-1}
    \end{subfigure}
    \hfill
    \begin{subfigure}[b]{0.4\textwidth}
        \centering
            \includegraphics[width=\linewidth]{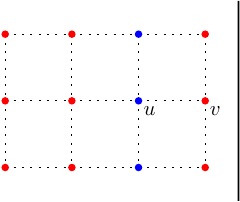}
        \caption{Step 2.}
        \label{fig:boundary-11-2}
    \end{subfigure}
    
    \caption{}
    \label{}
\end{figure}

\subsection{Case 12}

For Case 12, we primarily reduce down to the other cases whenever possible. To begin, if either of the highlighted vertices in \cref{fig:boundary-12-1} is disposable, then we can flip that vertex to go to Case 4. This produces any of a \minusred{1}, \plusred{0}, or \plusred{1} partition. If neither of them are disposable, then we can use \cref{lem:elbow} to reveal more information. Next, if the left highlighted vertex in \cref{fig:boundary-12-2} is red, then we can flip the vertex to the left of $u$ to reduce to Case 10. This creates any of a \minusred{3}, \minusred{2}, \minusred{1}, \plusred{0}, or \plusred{1} partition. Similarly, if either the top or bottom highlighted vertex in \cref{fig:boundary-12-2} is red, then we can flip the adjacent blue vertex to red to reduce to Case 3. This results in any of a \minusred{1}, \plusred{0}, or \plusred{1} partition. Thus, we assume that each of these highlighted vertices are blue. We next flip each of the highlighted vertices in \cref{fig:boundary-12-3}. By \cref{lem:create-island}, this creates three blue regions, of which two are islands.

\begin{figure}[H]
    \centering
    \begin{subfigure}[b]{0.3\textwidth}
        \centering
            \includegraphics[width=\linewidth]{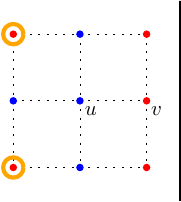}
        \caption{Step 1.}
        \label{fig:boundary-12-1}
    \end{subfigure}
    \hfill
    \begin{subfigure}[b]{0.3\textwidth}
        \centering
            \includegraphics[width=\linewidth]{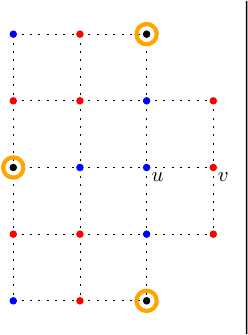}
        \caption{Step 2.}
        \label{fig:boundary-12-2}
    \end{subfigure}
    \hfill
    \begin{subfigure}[b]{0.3\textwidth}
        \centering
            \includegraphics[width=\linewidth]{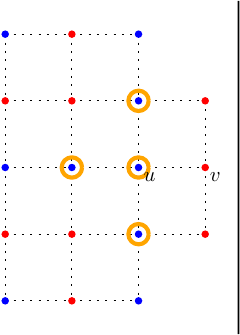}
        \caption{Step 3.}
        \label{fig:boundary-12-3}
    \end{subfigure}
    
    \caption{}
    \label{}
\end{figure}

Lastly, we note that at least two of the highlighted vertices in \cref{fig:boundary-12-4} are in 1-thin structures. These structures can be flipped to create a \plusred{2} partition. An example of a final state is shown in \cref{fig:boundary-12-5}.

\begin{figure}[H]
    \centering
    \begin{subfigure}[b]{0.4\textwidth}
        \centering
            \includegraphics[width=\linewidth]{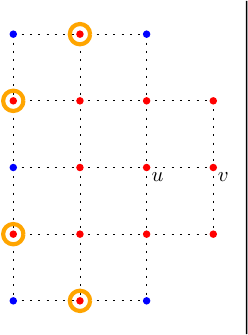}
        \caption{Step 4.}
        \label{fig:boundary-12-4}
    \end{subfigure}
    \hfill
    \begin{subfigure}[b]{0.4\textwidth}
        \centering
            \includegraphics[width=\linewidth]{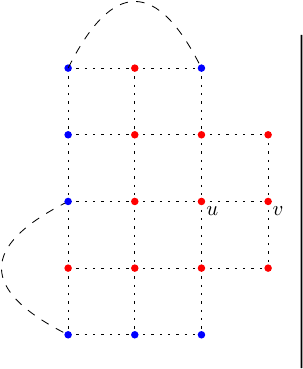}
        \caption{Final state.}
        \label{fig:boundary-12-5}
    \end{subfigure}
    
    \caption{}
    \label{}
\end{figure}

\end{document}